\documentclass{sig-alternate}

\usepackage{balance}  % for  \balance command ON LAST PAGE  (only there!)
\usepackage{times}
\usepackage[show]{chato-notes}
\usepackage{epsfig}
\usepackage{graphicx}
\usepackage[mathscr]{eucal}
\usepackage{latexsym}
\usepackage{float}
\usepackage{url}
\usepackage{color}
\usepackage{cite}
\usepackage{algorithmic}
\usepackage[ruled,linesnumbered,noend,vlined]{algorithm2e}
\usepackage{epstopdf}
\usepackage{subfigure}
\usepackage{multirow}
\usepackage{hyperref}
\hypersetup{
	colorlinks,breaklinks,
	urlcolor=[rgb]{0,0.25,0.75},
	linkcolor=[rgb]{0,0.25,0.75},
	citecolor=[rgb]{0,0.25,0.75}
}

\newtheorem{problem}{Problem}

\newtheorem{definition}{Definition}

\newtheorem{lemma}{Lemma}
\newtheorem{theorem}{Theorem}
\newtheorem{corollary}{Corollary}
\newtheorem{example}{Example}
\newtheorem{proposition}{Proposition}

\newcommand{\eat}[1]{}
\newcommand{\spara}[1]{\smallskip\noindent{\bf #1}.}
\newcommand{\mpara}[1]{\medskip\noindent{\bf #1}.}

\newcommand{\squishlist}{
 \begin{list}{$\bullet$}
  {  \setlength{\itemsep}{0pt}
     \setlength{\parsep}{3pt}
     \setlength{\topsep}{3pt}
     \setlength{\partopsep}{0pt}
     \setlength{\leftmargin}{2em}
     \setlength{\labelwidth}{1.5em}
     \setlength{\labelsep}{0.5em}
} }
\newcommand{\squishlisttight}{
 \begin{list}{$\bullet$}
  { \setlength{\itemsep}{0pt}
    \setlength{\parsep}{0pt}
    \setlength{\topsep}{0pt}
    \setlength{\partopsep}{0pt}
    \setlength{\leftmargin}{2em}
    \setlength{\labelwidth}{1.5em}
    \setlength{\labelsep}{0.5em}
} }

\newcommand{\squishdesc}{
 \begin{list}{}
  {  \setlength{\itemsep}{0pt}
     \setlength{\parsep}{3pt}
     \setlength{\topsep}{3pt}
     \setlength{\partopsep}{0pt}
     \setlength{\leftmargin}{1em}
     \setlength{\labelwidth}{1.5em}
     \setlength{\labelsep}{0.5em}
} }

\newcommand{\squishend}{
  \end{list}
}

\renewcommand{\algorithmicrequire}{\textbf{Input:}}
\renewcommand{\algorithmicensure}{\textbf{Output:}}
\renewcommand{\algorithmiccomment}[1]{// #1}

\DeclareMathOperator*{\argmax}{arg\,max}

\newcommand{\I}{\mathcal{I}}
\newcommand{\J}{\mathcal{J}}
\newcommand{\eqdef}{:=}

\newcommand{\revmax}{\textsc{RevMax}\xspace}
\newcommand{\M}{\mathcal{M}}

\newcommand{\Rev}{\mathit{Rev}}
\newcommand{\SGreedy}{\text{SL-Greedy}\xspace}
\newcommand{\GGreedy}{\text{G-Greedy}\xspace}
\newcommand{\RGreedy}{\text{RL-Greedy}\xspace}
\newcommand{\topkrating}{\text{TopRA}\xspace}
\newcommand{\topkrev}{\text{TopRE}\xspace}
\newcommand{\GlobalNo}{\text{GlobalNo}\xspace}

\newcommand{\redline}{\textcolor{red}{-----------------------------------------------}}

\newcommand{\model}{revenue model\xspace}

\newcommand{\bp}{\mathbf{p}}
\newcommand{\bq}{\mathbf{q}}

\newcommand{\ap}{\mathfrak{q}}
\newcommand{\dap}{\mathfrak{q}}
\newcommand{\edap}{\mathcal{E}}
\newcommand{\price}{\mathfrak{p}}

\newcommand{\MG}{\mathit{MR}}
\newcommand{\heap}{\mathtt{heap}}

\newcommand{\flag}{\mathtt{flag}}

\newcommand{\Top}{z}
\newcommand{\ComputeMG}{\textsf{ComputeMarginalGain}\xspace}

\newcommand{\Insert}{Insert}
\newcommand{\Remove}{DeleteMax}
\newcommand{\Add}{Add}

\newcommand{\cnt}{\mathtt{count}}
\newcommand{\List}{\mathtt{set}}

\newcommand{\userSet}{\mathtt{user\_set}}

\newcommand{\proj}{\mathrm{proj}}

\newcommand{\cl}[1]{\mathcal{#1}}
\newcommand{\val}{\mathit{val}}
\newcommand{\prob}{\mathcal{B}}

\newcommand{\LL}[1]{\textcolor{blue}{#1}}

\newcommand{\uheap}{\mathtt{upper\_heap}}
\newcommand{\lheap}{\mathtt{lower\_heap}}

\newcommand{\Var}{\mathrm{var}}
\newcommand{\Cov}{\mathrm{cov}}
\newcommand{\E}{\mathbb{E}}
\newcommand{\RandRev}{\mathit{RandRev}}
\newcommand{\bz}{\mathbf{z}}

\title{Show Me the Money: Dynamic Recommendations for Revenue Maximization\titlenote{An abridged version of this work is published in the Proceedings of the VLDB Endowment (PVLDB), volume 7, issue 14.}}

\subtitle{Full Technical Report}

\author{
\begin{tabular}{cccc}
Wei Lu  & Shanshan Chen & Keqian Li  & Laks V.S. Lakshmanan \\
\end{tabular}
\\$ $\\
\begin{tabular}{c}
\affaddr{Department of Computer Science}\\
\affaddr{University of British Columbia}  \\
\affaddr{Vancouver, B.C., Canada} \\
{\sf \{welu,cshan33,likeqian,laks\}@cs.ubc.ca} 
\end{tabular}
}

\makeatletter
\def\@copyrightspace{\relax}
\makeatother

\begin{document}
\maketitle

\begin{abstract}
Recommender Systems (RS) play a vital role in applications such as e-commerce and on-demand content streaming.
% and on-demand content streaming.
Research on RS has mainly focused on the \emph{customer perspective}, i.e., accurate prediction of user preferences and maximization of user utilities. 
As a result, most existing techniques are not explicitly built for \emph{revenue maximization}, the primary business goal of enterprises.
In this work, we explore and exploit a novel connection between RS and the profitability of a business.
As recommendations can be seen as an information channel between a business and its customers, it is interesting and important to investigate how to make strategic dynamic recommendations leading to maximum possible revenue.
To this end, we propose a novel \model that takes into account a variety of factors including prices, valuations, saturation effects, and competition  amongst products.
Under this model, we study the problem of finding revenue-maximizing recommendation strategies over a finite time horizon.
We show that this problem is NP-hard, but approximation guarantees can be obtained for a slightly relaxed version, by establishing an elegant connection to matroid theory. 
Given the prohibitively high complexity of the approximation algorithm, we also design intelligent heuristics for the original problem.
Finally, we conduct extensive experiments on two real and synthetic datasets and demonstrate the efficiency, scalability, and effectiveness our algorithms, and that they significantly outperform several intuitive baselines. 
%This work is to be presented at the 41st International Conference on Very Large Data Bases (VLDB), held in Hawaii, USA.
\end{abstract}

\section{Introduction}\label{sec:intro}
Fueled by online applications such as e-commerce (e.g., Amazon.com) and on-demand content streaming (e.g., Netflix), Recommender Systems (RS) have emerged as a popular paradigm and often as an alternative to search, for enabling customers to quickly discover needed items.
%Indeed, RS form the backbone of the business model of companies like Amazon and Netflix.
Such systems use the feedback (e.g., ratings) received from users on the items they have bought or experienced to build profiles of users and items, which are then used to suggest new items that might be of interest \cite{rssurvey05}.
%The feedback is usually given in the form of numeric ratings but can also be implicit.
The key technical problem on which most RS research has focused is rating prediction and rating-based top-$k$ recommendations:
given all ratings observed so far, %by users on items, 
predict
%We are given a matrix (indexed by users and items) containing all 
%We are given all observed ratings given so far by users to various items.
%From this, we wish to make accurate predictions of
the likely ratings users would give unrated products. %\footnote{We use the terms item and product interchangeably.} they have not yet experienced. % and hence not rated.
In practice, the data is huge, with typically hundreds of thousands of users and items, and is sparse, as users rate very few items.
%It is also very sparse as users rate very few items.
Major rating prediction techniques can be categorized into {content-based} and {collaborative filtering} (CF).
The latter can be further divided into {memory-based} and {model-based} ({\em cf.} \textsection\ref{sec:related}, also \cite{rssurvey05, mfsurvey, RS:hb10}).
%We refer the reader to \cite{rssurvey05, mfsurvey, RS:hb10} for excellent surveys and resource on contemporary RS research.
Here we simply note that for top-$k$ recommendations, a RS predicts ratings for all user-item pairs, and for each user, pushes the items with the highest predicted ratings as personalized recommendations. % to each user in a personalized manner.

%%%%%%%%%
\eat{We simply note that model-based approaches use a low rank approximation of the original rating matrix to build a model of the users and items.
Specifically, using one of a variety of loss functions (such as RMSE~\cite{mfsurvey,pmf}, NDCG \cite{RS:hb10}, etc.), they represent both users and items as vectors over a low-dimensional latent space and learn these vectors by training a model which minimizes the loss function.
Once trained, dot products between user and item vectors are used to predict the corresponding missing ratings.
This in turn is used to push the items with the highest predicted ratings as suggestions to each user in a personalized manner.}
%%%%%%%%%%

The above account is a \emph{customer-centric} view of RS, on which most existing research has focused.
We also refer to it as \emph{rating-based} recommendation.
There is an equally important, complementary \emph{business-centric} view.
From a business viewpoint, the  goal of a seller running a RS is to maximize profit or revenue. %, a central question studied in economy.
It is thus natural to ask: When a recommender system has choices of promising items but can only recommend up to $k$ items at a time, how should it make the selections to achieve better revenue?
We refer to this as \emph{revenue-driven recommendation}.

Top-$k$ rating-based recommendations are relevant to revenue, but the connection is weak and indirect: % since its main objective is to maximize users' utility. % they only indirectly try to achieve this objective.
by suggesting to each user the $k$ items with highest predicted rating, the hope is to increase the chances that some of those items are bought, and there ends the connection.
Clearly, there exists a gap: rating-based approaches ignores important monetary factors that are vital to users' product adoption decisions, namely the price and whether such price would be considered acceptable to an individual. %, by her valuation\footnote{A buyer's valuation toward an item is defined as the buyer's maximum willingness to pay for the item~\cite{}.  It is a notion widely adopted in micro-economy theory \cite{}}.
%If valuation is less than price, the buyer will not adopt the product no matter how much affinity she has~\cite{kalish85}.
%Exact valuations are typically unknown for trust and privacy reasons~\cite{}.
%A common practice is to make the \emph{Independent Private Value} (IPV) assumption, under which the valuation of each user is drawn independently at random from a public distribution, which is often learned using historical sales data.
Naturally there exists inherent uncertainty in users' purchase decisions due to people's innate valuation \cite{kalish85}, %and thus one must model such uncertainty and incorporate it into the design of revenue-driven recommendation algorithms.
and thus successful revenue-driven recommendations must be aware of these monetary factors.

There has been some recent work on revenue-driven recommendations \cite{chen08,das09,azaria13} (see \textsection\ref{sec:related} for more details).
% and the probability of adoption (purchase) \cite{}.
All these previous works focus on a static setting: generate a set of recommendations for various users for a given snapshot, corresponding to a specific time.
This is limited because e-commerce markets are highly \emph{dynamic} in nature.
For example, the price of many products in Amazon has been found to change and fluctuate frequently \cite{camel, wsj}.
Therefore, narrowly focusing on just one snapshot is restrictive.
Instead, it is necessary to be strategic about the recommendations rolled out in a short time period.
For instance, suppose a product is scheduled to go on sale in the next few days, it would be wise to strategically postpone its recommendation to low-valuation users (those not willing to pay the original price) to the sale date, so as to increase the adoption probability.
On the contrary, for high-valuation users (those willing to pay the original price), it is wiser to recommend it before the price drops.
This illustrates that sellers may benefit from taking full advantage of known %complete 
price information in near future. %, the seller can be better off.
Thus, \textsl{revenue-driven recommendations should be not only price-aware, but also sensitive to timing.}
As such, we focus on a dynamic setting where a strategic recommendation plan is rolled out over a time horizon for which pricing information is available (either deterministically or probabilistically).
We note that by myopically repeating items, the static approaches can be used to roll out recommendations over a horizon; we will empirically compare the performance of such strategies with our solutions in \textsection\ref{sec:exp}.

The dynamic setting naturally leads to two more important aspects related to revenue: competition and repeated recommendations.
Products in the same class (kind)  provide similar functionalities to customers, and thus an individual is unlikely to purchase multiple of them at once.
This signals competition and its implications for the actual revenue yielded by repeated recommendations cannot be ignored.
For instance, iPhone 5S, Nokia 820, and Blackberry Z10 are all smartphones.
While an individual may find them all appealing, she is likely to purchase at most one even if all three phones are suggested to her (potentially more than once) during a short time period.
%In other words, all recommendations of phones to a single user are likely to generate revenue of one phone's price, and overestimating such revenue may lead to suboptimal strategies.
In addition, as %evident in psychology literature 
observed recently \cite{dassarma12}, while moderate repeated recommendation may boost the ultimate chance of adoption, overly repeating could lead to boredom, or {\em saturation}. 
All these factors, not considered in previous  work \cite{chen08, das09, azaria13}, call for careful modeling and a more sophisticated algorithmic framework.
%It is thus important to factor such competitions in planning recommendation strategies. 

%%%%%%%%%%%%%%%%%%
\eat{
We argue that a holistic view of the recommendation strategy is essential for revenue maximization.
More precisely, there is an inherent uncertainty in customers adopting items, and thus we must model such uncertainty and design recommendation strategies accordingly. % to maximize the expected revenue the company can make.
The strategy should take into account pricing information of items and possible price fluctuations over the time horizon in which recommendations are made.
Also, it is necessary to be strategic about how recommendations are rolled out.
Explicitly modeling price fluctuations, saturation effects, competition, and the interdependency of repeated recommendations gives us an opportunity to reason about \textsl{what items to recommend, whether to repeat recommendations, and if yes, which items to repeat, and how often}. %, taking price fluctuations, saturation effects, and competitions into account. 
}
%%%%%%%%%%%%%%%%%%

Driven by the above, we pose the problem of designing recommendation strategies to maximize the seller's expected revenue as a novel discrete optimization problem, named \revmax (for \emph{revenue maximization}).
Our framework allows any type of RS to be used, be it content-based, memory-based CF, or model-based CF.
It takes prices as exogenous input to the RS, and nicely models adoption uncertainly by integrating the aforementioned factors: price, competition, and saturation %, all justified with literature support 
(see \textsection\ref{sec:revmax} for a detailed justification and supporting literature).  % and the saturation effects on customers w.r.t.\ various products are known.
%We also devise a method for modeling the adoption probability %(i.e., the probability that a customer will adopt)
%of a product by a customer and cast
We show \revmax is NP-hard and develop an approximation algorithm (for a relaxed version) and fast greedy heuristics.
To the best of our knowledge, no existing work on RS has a revenue model that incorporates the various factors captured by our revenue model. 
Furthermore, the discrete optimization problem proposed and the algorithmic solution framework designed in this work are unique and complement previous work nicely.

To summarize, we make the following contributions.
\enlargethispage*{\baselineskip}
\squishlist
\item
We propose a dynamic revenue model that accounts for time, price, saturation effects, and competition.
Under this model, we propose \revmax, the problem of finding a recommendation strategy that maximizes the seller's expected total revenue over a short time horizon.
To the best of our knowledge, this framework is novel and unique.

\item 
We show that \revmax is NP-hard (\textsection\ref{sec:revmax}) and its objective function is non-monotone and submodular (\textsection\ref{sec:theory}).
By establishing a connection to matroids, a relaxed version of \revmax can be modeled as submodular function maximization subject to matroid constraints, enabling a local search algorithm with an approximation guarantee of $1/(4+\epsilon)$, for any $\epsilon > 0$ (\textsection\ref{sec:theory}).

\item 
Since the local search approximation algorithm is prohibitively expensive and not practical, we design several clever and scalable greedy heuristics for \revmax (\textsection\ref{sec:algo}). 

\item We perform a comprehensive set of experiments on two real datasets (Amazon and Epinions) that we crawled, to show the effectiveness and efficiency of our algorithms.
In particular, they significantly outperform natural baselines.  
We use the data to learn the various parameters %such as price fluctuation %, saturation parameters, primitive adoption probabilities
used in the experiments (\textsection\ref{sec:exp}). 
%We do this and that and that and this that nobody has ever done and show this wonderful thing 
We also show that our best algorithm, Global Greedy, scales to a (synthetically generated) dataset that is 2.5 times as large as Netflix, the largest publicly available ratings dataset. 
\item 
We also discuss promising extensions to our framework in \textsection\ref{sec:discuss}.
In particular, we show that when exact price information is not available but only a price distribution is known, Taylor approximation can be employed for  maximizing (approximate) expected total revenue.
\squishend

\section{Background and Related Work}\label{sec:related}
%We first provide a brief overview of research on
%	recommender systems.
%Then we review previous work on profit- or revenue-aware RS and
%	discuss key differences between these work and ours.
%Related work mainly falls in one of the following categories.

%\spara{Matrix factorization}
The majority of the research in RS has focused on algorithms for accurately predicting user preferences and their utility for items in the form of ratings, from a set of observed ratings. 
%The input is users' feedback on items, usually in the form of integral ratings (stars).
%Depending on application, there are sometimes meta-data available that gives additional information such as the gender and age of a user, or the genre of a movie.
As mentioned in \textsection\ref{sec:intro}, RS algorithms can be classified into  content-based or collaborative filtering (CF) methods, with CF being further classified into  memory-based or model-based \cite{rssurvey05}. 
The state-of-the-art is collaborative filtering using Matrix Factorization (MF), which combines good accuracy and scalability~\cite{mfsurvey}.
In MF, users and items are characterized by latent feature vectors inferred from ratings.
More specifically, we learn for each user $u$ a vector $\bp_u \in \mathbb{R}^f$ and for each item $i$ a vector $\bq_i \in \mathbb{R}^f$, such that 
%the inner product of $\bp_u$ and $\bq_i$ approximates the actual rating $r_{ui}$ well.
$r_{ui} \approx \bp^T_u \bq_i$ for all observed ratings $r_{ui}$ ($f$ is the number of latent features). 
Training the model requires an appropriate loss function to measure the training error, so that using methods like stochastic gradient descent or alternating least squares, the latent feature vectors can be learned quickly in a few iterations \cite{mfsurvey}.
The most common loss functions is RMSE \cite{mfsurvey}. %, NDCG \cite{RS:hb10}, etc. %AUC (Area Under the ROC Curve) \cite{RS:hb10}, etc. 
%The work in \cite{pmf} provides a probabilistic interpretation of MF, where a graphical model is used to capture the interactions between rating patterns and latent features, under which $r_{ui}$ follows a Gaussian distribution with mean $\bp^T_u \bq_i$ and a variance closely related to the $L_2$-regularization parameters of the model.
For our purposes, existing MF methods can be used to compute predicted ratings accurately and
	efficiently, although our framework does permit the use of other methods.
Recently, Koren~\cite{koren10} considers the drift in user interests
	and item popularity over time and proposes a temporal MF model
	to further improve accuracy.
This model is for rating prediction and uses data spanning a \emph{long} period,
	since drifts usually occur gradually, while our framework addresses dynamic
	recommendations over a \emph{short} period and has a completely different
	objective (revenue).

Other related work of the customer-centric view includes \cite{zhao12, wang13}.
Zhao et al.~\cite{zhao12} compute the time interval between successive product purchases and integrate it into a utility model for recommendations, with the idea of capturing the timing of recommendations well and increasing temporal diversity. % of the recommendations.
Wang et al.~\cite{wang13} apply survival analysis to estimate from data the probability that a user will purchase a given product at a given time. %using purchase history data.
It offers a way of modeling and computing adoption probability and is in this  sense orthogonal to \revmax.
We depart from these works by focusing on strategic recommendation plans for short-term for which price information is known and repeated purchase is rare, and formulating \revmax as a discrete optimization problem.

%\textcolor{red}{===============================}
	%, e.g., user-similarity based or item-similarity based RS. 

%\spara{Other Directions of RS Research}
%There has been work on diversifying recommendations, in the hope that less homogeneous results leads to better adoption rate~\cite{yu09}. % and the natural trade-off between diversity and accuracy: 
%It has a similar flavor to ours as there is also a trade-off between revenue maximization and accurate rating prediction, since revenue-maximizing recommendations do not necessarily suggest items for which users have the highest predicted ratings.
%Goyal and Lakshmanan~\cite{recmax} describes how a RS can be utilized to launch a marketing campaign so that a specific item, after rated by a selected set of ``seed'' users, will show up in as many users' top-$k$ list as possible.
%Their work focuses on maximizing expected revenue for the manufacturer of the item.

%Other related work mainly fall into the following categories.
 
%\spara{Revenue-driven recommendations}
There has been some work on revenue-driven recommendations \cite{chen08, das09, azaria13}.
However, %as mentioned in \textsection\ref{sec:intro}, 
most of it deals with a \emph{static} setting and is limited in several senses.
%That is, they focus on generating recommendations for one fixed time point, i.e., ``one-shot'' recommendations. 
Chen et al.~\cite{chen08} model adoption probabilities using purchase frequency and user similarity, and recommend each user $k$-highest ranked items in terms of expected revenue (probability $\times$ price). %There ends the similarity to our problem. 
Important ingredients such as capacity constraints, competition, price fluctuations and saturation are all ignored. 
%Their problem can be viewed as the static version of our problem minus item %quantity constraints.
Consequently, their optimization can simply be solved independently for each user, without the combinatorial explosion inherent in our problem, which is far more challenging. It is also tied to memory-based RS (user similarity), whereas our approach is more generic and works for any type of RS. %, including sophisticated ones based on MF. 
Das et al.~\cite{das09} propose a simple profit model which predicts the adoption probability using the similarity between the user's true ratings and the system prediction.
It captures the trust between users and the RS, but does not take into account the fact that when the price or relevance of suggested items changes, user's adoption probability ought to change too.
Azaria et al.~\cite{azaria13} model the probability of a movie being viewed based on its price and rank (from a black-box RS).
Their setting is also static and focuses on movies only, which is not sufficient for most e-commerce businesses that sell a variety of products. %and look for revenue over a time period.

%\note[Laks]{Let's not cite recmax in this version.} 
\eat{ 
From a slightly different perspective, \cite{recmax} describes how a RS can be utilized to launch a marketing campaign so that a specific item, after rated by a selected set of ``seed'' users, will show up in as many users' top-$k$ list as possible.
This problem focuses more on maximizing the utility for the manufacturer of the item, rather than the business operating RS itself.
}

Static revenue-driven recommendation, on an abstract level, is related to generalized bipartite matching  \cite{gabow83, gabow89}.
The papers \cite{gianmarco11, manshadi13} propose distributed solution frameworks to scale this up to massive datasets.
%Briefly speaking, 
They are substantially different from and orthogonal to ours as the focus is on distributed computation, still within a static setting.
We defer detailed discussion to \textsection\ref{sec:revmax}.

\section{Revenue Maximization}\label{sec:revmax}
The motivation for a dynamic view of revenue-based recommendations
 	is twofold. 
First, prices change frequently, sometimes even on a daily basis, or several times a day in online marketplaces like Amazon \cite{camel, wsj}.
A static recommendation strategy is forced to be myopic and cannot
	take advantage of price fluctuations to choose the right
	\emph{times} at which an item should be recommended.
This shortcoming is glaring, since for product adoptions to occur, a user's \emph{valuation}
	of an item, i.e., the maximum amount of money she is willing to pay,
	must exceed the price. This is a well-known notion in economics literature~\cite{agtbook,klbbook,kalish85}.
And obviously, the optimal time in terms of extracting revenue is not necessarily
	when the price is at its peak or bottom!
Therefore, \textsl{only by incorporating time, we can %leverage both price fluctuation and user valuation.
	design revenue-maximizing recommendations in a systematic, holistic, and principled way.}
%In \textsection\ref{sec:exp}, we discuss how we learn their distributions from real data. 

Second, we can recommend items strategically over a time horizon,
	possibly making some repetitions, to improve the
	chances that some items are eventually bought.
However, this is delicate.
Indiscriminate repeat recommendations can lead to \emph{saturation} effects (boredom),
	turning off the user from the item or even costing her loyalty to
	the RS.
Thus, we need to choose the \emph{timing} of recommendations as well as their
	\emph{frequency} and \emph{placement} so as to maximize the expected revenue in
	the presence of saturation effects.
\textsl{These opportunities are simply not afforded when recommendations are
	made in a static, one-shot basis.}

\subsection{Problem Setting and Definition}\label{sec:revmaxdef}

We adopt a discrete time model to describe the dynamics of revenue in an e-commerce business.
Let $U$ and $I$ be the set of users and items respectively,  and define $[T] \eqdef \{1,2,\dotsc,T\}$ to be a short time horizon over which a business (or seller) plans a recommendation strategy.
The seller may choose the granularity of time at her discretion. E.g., the granularity may be a day and the horizon may correspond to a week. 

Naturally, items can be grouped into classes based on their features and functionalities, e.g.,
``Apple iPad Mini'' and ``Google Nexus 7'' belong to class {\tt tablet}, while ``Samsung Galaxy S4'' and ``Google Nexus 5'' belong to {\tt smartphone}. 
Items in the same class offer a variety of options to customers and thus compete with each other, since from an economic point of view, it is unlikely for a person to adopt multiple products with similar functionality over a short time horizon.
Thus in our model, we assume items within a class are competitive in the sense that their adoption by any user within $[T]$ is {\em mutually exclusive}.

In general, revenue to be expected from a recommendation depends mainly on two factors:
the price of the item, and the probability that the user likes and then adopts the item at that price.
%\note[Laks]{What about the probability of a user liking an item?} 
The latter, as mentioned in \textsection\ref{sec:intro}, is in turn dependent on many
factors, including but not restricted to, price.
The former, price, is conceptually easier to model, and we suggest two alternatives.
Both of the following essentially model prices to be \emph{exogenous} to the RS, meaning, they are externally determined and
provided to the RS.
First, we can use an \emph{exact price} model which assumes that for each item $i$ and
time step $t$, the exact value $\price(i,t)$ is known.
%\textcolor{blue}{
Exact pricing is studied extensively in microeconomics theory.
There is a well established theory for estimating future demand and supply, which are further used to derive future price in market equilibrium \cite{snyder08}. 
In practice, we also observe that retailers (both online and offline) do plan their pricing ahead of time (e.g., sales on Black Friday in the US and Boxing Day in Canada).
Alternatively, in case prices are not completely known -- which could be possible if the
time of interest is too far into the future, or there isn't sufficient data to determine a precise value --  they can be modeled as random variables following a certain probability distribution. In other words, there is a price prediction model that produces a probability distribution associated with price. 
As we shall see shortly, the challenge of revenue maximization in RS remains  in both
exact price and random price models.
In the bulk of this paper, we focus on the exact price model, but in \textsection\ref{sec:discuss}, give insights on how the random price model can be incorporated into our framework. 
%}

It is well known that there is uncertainty in product adoption by users.
Consider any user-item-time triple $(u,i,t) \in U\times I\times [T]$. 
We associate it with a \emph{(primitive) adoption probability} $\ap(u,i,t)\in [0,1]$, the probability that $u$ would purchase $i$ at time $t$.
%The relationship between adoption probability and price is captured by the natural assumption,
Based on economics theory, buyers are rational and utility-maximizing \cite{klbbook,agtbook} and would prefer lower price for higher utility (of consuming a product).
%Thus, the adoption probability of an item is typically a \emph{anti-monotone} function of its price.\footnote{\textcolor{red}{Our framework and algorithms do not assume this.}} 
%That is, $\price(i, t_1) \le \price(i, t_2)$ implies that $\ap(u,i,t_1) \ge \ap(u,i,t_2)$, $\forall u\in U, \forall i\in I$, $\forall t_1, t_2 \in [T]$.
%It is worth mentioning that 
Rational agents are typically assumed to possess a private, internal \emph{valuation} for a good, which is the maximum amount of money the buyer is willing to pay; furthermore, users are assumed to be price-takers who respond myopically to the prices offered to them, solely based on their privately-held valuations and the price offered.  Both the internal private valuation assumption and the price taker assumption are standard in economics literature \cite{klbbook,agtbook}.
Thus, though a lower price leads to a higher probability of adoption, it might not be optimal to recommend an item at its lowest price, since some buyers actually could be willing to pay more. In general, a seller may estimate the adoption probabilities using its own data and model, and the exact method used for estimating adoption probabilities is orthogonal to our framework.
In \textsection\ref{sec:exp}, we give one method for estimating adoption probabilities from the real datasets Amazon and Epinions. We note that the above notion is ``primitive'' in the sense that it ignores the effect of competition and saturation, which will be captured in Definition~\ref{def:dap}. 

\eat{Indeed, buyer valuation can be a way of modeling adoption probability, but for simplicity and concentration, we defer discussing our choice to \textsection\ref{sec:exp}.
We note that the exact method used for computing adoption probabilities is orthogonal to our framework, since sellers may wish to estimate them using their own data and model.
In fact, a comprehensive study on adoption probability is itself an interesting topic.} 

\eat{Should there be nothing else to consider, the expected profit yielded by any triple $(u,i,t)$ would be $\price(i,t) \times \ap(u,i,t)$.} 
Intuitively, the expected revenue yielded by recommending item $i$ to user $u$ at time $t$ is $\price(i,t) \times \ap(u,i,t)$, if this recommendation is considered in isolation. 
However, competition and repeated suggestions make the expected revenue computation more delicate.
We first formally define a \emph{recommendation strategy} as a set of user-item-time triples $S\subseteq U\times I\times [T]$, where $(u,i,t)\in S$ means $i$ is recommended to $u$ at time $t$.
\eat{Any $S\subseteq U\times I\times [T]$ is a strategy (good or bad), and $(u,i,t)\in S$ implies that if the RS adopts $S$, then it will suggest item $i$ to user $u$ at time step $t$. } 
%With that,
Let us next consider the effect of competition and repeated recommendations.
Suppose $S$ suggests multiple items (possibly repeated) within a class to user $u$, then the adoption decisions that a user makes at different times are \emph{interdependent} in the following manner: %we assume that the user will not buy the same item (or another from the same class) repeatedly just because repeated recommendations are made. 
%Instead, \textsl{we conservatively assume she will buy at most one item from a class and at most once}.
The event that a user $u$ adopts an item $i$ at time $t$ is conditioned on $u$ not adopting any item from $i$'s class before, and once $u$ adopts $i$, she will not adopt anything from that class again in the horizon (which, it may be recalled, is short in our model).
This semantics intuitively captures competition effects.

As argued earlier, repeated suggestions may lead to a boost in adoption probability, but repeating too frequently may backfire, as people have a tendency to develop boredom, which can potentially cause devaluation in user preference \cite{dassarma12, kapoor13}. We call this \emph{saturation effect}.  
This is undesirable in terms of revenue and thus the model should be saturation-aware.
We discount the adoption probability by accounting for saturation, similar in spirit to \cite{dassarma12}.
The intuition is that the more often and more recently a user has been recommended an item or an item from the same class, the fresher her memory  and consequently the more significant the devaluation. 
Let $\cl{C}(i)$ denote the class to which $i$ belongs.
The memory of user $u$ on item $i$ at any time $t>1$ w.r.t.\ a recommendation strategy $S$ is
\begin{align}\label{eqn:mem}
M_S(u,i,t) := \sum\nolimits_{j\in \mathcal{C}(i)}  \sum\nolimits_{\tau=1}^{t-1} \frac{X_S(u,j,\tau)}{t-\tau},
\end{align}
where $X_S(u,j,\tau)$ is an indicator variable taking on $1$ if $(u,j,\tau) \in S$ and $0$ otherwise. 
Also $X_S(u,i,1) := 0$ for all $u$, $i$, and $S$. % since there is no time $0$ and hence no memory occurs at $t=1$.
The discounting on adoption probability should reflect the fact that devaluation is monotone in memory.
More specifically, given user $u$, item $i$, and time step $t$, we define the \emph{saturation-aware adoption probability} under a strategy $S$ as $\ap(u,i,t) \times \beta_i^{M_S(u,i,t)}$,  where $\beta_i \in [0,1]$ is the saturation factor of item $i$.
This effectively penalizes recommending the same class of items in (nearly) consecutive time steps.
Smaller $\beta_i$ means greater saturation effect.
$\beta_i=1$ corresponds to no saturation and $\beta_i = 0$ corresponds to full saturation: any repetition immediately leads to zero probability.
In principle, $\beta_i$'s can be learned from historical recommendation logs ({\em cf.} \cite{dassarma12}).

Combining competition and saturation effects, we are now ready to define the \emph{strategy-dependent dynamic adoption probability}.
Basically, the dynamic adoption probability for later times depends on earlier times at which recommendations are made.
Intuitively, $\dap_S(u,i,t)$ is the probability that $u$  \textsl{adopts $i$ at time $t$ and does not adopt any item from its class earlier}, under the influence of $S$.

\begin{definition}[Dynamic Adoption Probability]\label{def:dap}
For any user $u$, item $i$, and time step $t$, given a strategy $S$, the dynamic adoption probability governed by $S$ is defined as follows:
\begin{align}\label{eqn:finalap}
&\dap_S(u,i,t) = \ap(u,i,t)  \times \beta_i^{M_S(u,i,t)} \times  \nonumber \\
&  \prod_{\substack{(u,j,t) \in S : \\ j\neq i, \cl{C}(j) = \cl{C}(i)}} (1-\dap(u,j,t))  \prod_{\substack{(u,j,\tau) \in S: \\ \tau < t, \cl{C}(j) = \cl{C}(i)}} (1-\dap(u,j,\tau))   .
\end{align}

Also, define $\dap_S(u,i,t) = 0$ whenever $(u,i,t)\not\in S$.
\end{definition}

\begin{example}[Dynamic Adoption Probability]\label{ex:dap}
Suppose the strategy $S$ consists of $\{(u,i,1), (u,j,2), (u,i,3)\}$, where $\cl{C}(i) = \cl{C}(j)$,
and for all three triples, the (primitive) adoption probability is $a$.
Then,
$\dap_S(u,i,1) = a$, 
$\dap_S(u,j,2) = (1-a) \cdot a \cdot \beta_i^{\frac{1}{1}}$, and 
$\dap_S(u,i,3) = (1-a)^2 \cdot a \cdot \beta_i^{\frac{1}{1}+\frac{1}{2}}$. 
\end{example}

We next define the \emph{expected revenue} of a recommendation strategy, as the expected amount of money the business (seller) can earn from recommendation-led adoptions by following the strategy.
%Unless stated otherwise, by strategy, we mean a valid strategy that respects both the display and capacity  constraints. 

\begin{definition}[Revenue Function]\label{def:rev}
Given a set $U$ of users, a set $I$ of items, and time horizon $[T]$, the \emph{expected revenue} of any strategy $S\subseteq U\times I\times [T]$ is
\begin{align}\label{eqn:rev}
\Rev(S) = \sum\nolimits_{(u,i,t)\in S}  \price(i,t) \times \dap_S(u,i,t),
\end{align}
where $\dap_S(u,i,t)$ is defined in Definition~\ref{def:dap}.
\end{definition}

The \emph{revenue maximization problem} asks to find a strategy $S$ with maximum expected revenue, subject to
the following two natural constraints:
First, the \emph{display constraint}, following standard practice, requires that no more than $k$ items be recommended to a user at a time.
Second, as items all have limited quantities in inventory, we impose a \emph{capacity constraint}: no item $i$ may be recommended to more than $q_i$ distinct users at any time, where $q_i$ is a number determined based on current inventory level and demand forecasting \cite{Porteus90}.
The intention is to avoid potential customer dissatisfaction caused by out-of-stock items.
In general, $q_i$ can be somewhat higher than the actual inventory level, due to uncertainty in product adoption.
We call a strategy $S$ \emph{valid} if $S$ satisfies the following constraints:
(i) for all $u \in U$ and all $t \in [T]$, $|\{i \mid (u, i, t) \in S \}| \leq k$;
(ii) for all $i \in I$, $|\{u \mid \exists t\in [T]: (u,i,t) \in S \} \mid \leq q_i$.

\begin{problem}[Revenue Maximization (\revmax)]
\label{prob:revmax} 
Given a set $U$ of users, $I$ of items, time horizon $[T]$, display limit $k$; for each $i\in I$, capacity $q_i$, price $\price(i,t)$ $\forall t\in [T]$, and saturation factor $\beta_i$; adoption probabilities $\ap(u,i,t)$ for all $(u,i,t)\in U\times I\times [T]$,
	find the optimal valid recommendation strategy, i.e., $S^* \in  \argmax_{S \text{ is valid }} \Rev(S) $.
%\footnote{Note that the optimal strategy may not be unique.} 
\end{problem}

\subsection{Hardness of Revenue Maximization}
As we will show shortly, \revmax is in general NP-hard.
However,  when $T=1$, it is PTIME solvable since it can be cast as a special case of the \emph{maximum-weighted degree-constrained subgraph} problem (Max-DCS) \cite{gabow83, gabow89}.
In Max-DCS, we are given a graph $G=(V,E)$, a weight $w_e$ for each edge $e\in E$, and a number $d_v$ for each node $v\in V$.
The task is to find a subgraph of $G$ that has maximum weight and for each $v\in V$, it contains no more than $d_v$ edges incident to $v$.
The best known combinatorial algorithm for Max-DCS takes time $O(D(|E|+|V|\log|V|))$  \cite{gabow83, gabow89}. 
%\note[Laks]{Double check if it's indeed the best known.} 

Given an instance $\I$ of \revmax with $T=1$, we create an instance $\J$ of Max-DCS, by building a bipartite graph $G'=(U\cup I, E')$, where $E' \subseteq U\times I$.
We create one node per user, and per item.
For each user-node $u$, set $d_u = k$, and for each item-node $i$, set $d_i = q_i$.
There is an edge $(u,i)\in E'$ if $\ap(u,i,1) > 0$ and set the weight $w_{(u,i)} = \price(i,1) \times \ap(u,i,1)$.
%Thus we have transformed the \revmax instance into a Max-DCS instance. 
It is easy to see the optimal subgraph in $\J$ corresponds to the optimal strategy in $\I$. Thus, this special case of \revmax can be solved efficiently using known algorithms for Max-DCS. 

We now show that the general \revmax is NP-hard.
The reduction is from a restricted version of the Timetable-Design problem (RTD), which is NP-complete \cite{gareyJohnson, even75}, to the decision-version of \revmax (D-\revmax).
An instance of the RTD problem consists of a set $C$ of craftsmen, a set $B$ of jobs, and a set of $H$ of hours, where $|H| = 3$. Each craftsman $c\in C$ is available for a subset $A(c) \subseteq H$ of hours, thus $|A(c)|$ is the number of hours that $c$ is available. There is a function $R: C \times B \to \{0, 1\}$, such that $R(c,b)$ specifies the number of hours craftsman $c$ is required to work on job $b$.
A craftsman is a {\em $\tau$-craftsman} if $|A(c)| = \tau$.
Furthermore, he is {\em tight}, if he is a $\tau$-craftsman for some $\tau$ and is required to work on $\tau$ jobs, i.e., $\sum_{b\in B}R(c,b) = |A(c)| = \tau$. In RTD, every craftsman is either a $2$-craftsman or a $3$-craftsman and is tight. A timetable is a function $f: C \times B \times H \to \{0, 1\}$. It is \emph{feasible} if all of the following hold:
\begin{enumerate} 
%\vspace{-1mm} 
\item $f(c,b,h) = 1$ implies $h\in A(c)$;
%\vspace{-2mm} 
\item for each $h\in H$ and $c\in C$ there is at most one $b\in B$ such that $f(c,b,h) = 1$;
%\vspace{-2mm} 
\item for each $h\in H$ and $b\in B$ there is at most one $c\in C$ such that $f(c,b,h) = 1$; and 
%\vspace{-2mm} 
\item for each pair $(c,b)\in C\times B$ there are exactly $R(c,b)$ values of $h$ for which $f(c,b,h) = 1$.
%\vspace{-1mm} 
\end{enumerate} 
The task in RTD is to determine if a feasible timetable exists.

\begin{theorem}\label{thm:nph}
Problem \ref{prob:revmax} (\revmax) is NP-hard.
\end{theorem}

\begin{proof}%[of Theorem~\ref{thm:nph}]
Given an instance $\I$ of RTD, we create an
instance $\J$ of D-\revmax as follows. Craftsmen and hours in $\I$ correspond to
users and time steps in $\J$. 
For each job $b$, create three items of class $b$: $i^b_1, i^b_2, i^b_3$.
The subscripts correspond to the three hours in $H$.
For every item the capacity constraint is $1$, i.e., $q(i^b_\tau) = 1$, $\forall b\in B, \tau\in H$. 
The display constraint $k = 1$. 
Set price $\price(i^b_\tau, t) = 1$ if $t = \tau$, and $0$ otherwise. 
For user $c$, item $i^b_\tau$, and any time step $t \in H$, set $\ap(c, i^b_\tau, t) = R(c,b)$.
For every user $c\in C$, we create a unique expensive item $e_c$ and set $\price(e_c,t) = E$ for all time $t$, where $E$ is an arbitrary number such that $E > N \eqdef \sum_{c\in C, b\in B} R(c,b)$.
For any time at which $c$ is unavailable, i.e., if $t \in H \setminus A(c)$, set $\ap(c, e_c, t) = 1$ and for all times at which $c$ is available, i.e., for $t\in A(c)$, $\ap(c, e_c, t) = 0$.
Finally, let $\Upsilon = \sum_{c\in C} |H\setminus A(c)|$, i.e., the total number of unavailable hours of all craftsmen. 

{\sc Claim}: $\J$ admits a recommendation strategy of expected revenue
$\geq N + \Upsilon E$ if and only if $\I$ admits a feasible timetable. 

($\Longleftarrow$): Suppose $f$ is a timetable for $\I$.
Define the recommendation strategy as $S = \{(c, i^b_t, t)
\mid f(c, b, t) = 1\} \cup \{(c, e_c, t) \mid t \in H \setminus A(c)\}$. It
is straightforward to verify that every recommendation will lead to an
adoption, including the expensive items. Notice that by feasibility of $f$, no
craftsman is assigned to more than one job at a time and no job is assigned to
more than one craftsman at any time, so the display and capacity constraints are satisfied by $S$. It follows that the expected
revenue is exactly $N + \Upsilon E$. 

($\Longrightarrow$): Let $S$ be a valid recommendation strategy that leads to an
expected revenue $\geq N + \Upsilon E$. Unless every recommendation in $S$
leads to adoption, the expected revenue cannot reach $N + \Upsilon E$. This implies that $S$ cannot recommend more than one item from a certain class to
any user, i.e., $S$ cannot contain both $(c,i^b_{t_1},t_1)$ and $(c,i^b_{t_2},t_2)$ for
any user $c$, class $b$ and times $t_1 \neq t_2$, for if it did, one of the
recommendations (the one that occurred later) would be rejected by user $c$ for
sure, so the opportunity to make a profit at that future time is lost. 
Define $f(c, b, t) = 1$ whenever $(c, i^b_t, t) \in S$, and $f(c, b, t) = 0$ otherwise.
We will show that $f$ is a feasible timetable.
By construction, expensive items do not correspond to
jobs and must have been recommended to each user $c$ when $c$ is unavailable according to $A(c)$. Any optimal strategy will contain these expensive item recommendations which account for a profit of $\Upsilon E$.  
Thus, the remaining recommendations in $S$ must net an expected revenue of at least $N$. Since each of them recommends an inexpensive item corresponding to a job, it can lead to an expected revenue of at most $1$. 
We shall next show that $f$ is feasible.
(1) Let $f(c,b,t) = 1$. Then we have  $(c,i^b_t,t)\in S$ by construction. Since every recommendation must make a non-zero revenue in an optimal strategy, we must have $t \in A(c)$ by construction. 
(2) Suppose $f(c,b,t) = f(c',b,t) = 1$, for some distinct users $c,
c'$. This implies $(c,i^b_t,t), (c',i^b_t,t) \in S$. But this is impossible,
since $q(i^b_t)=1$,
which would make $S$ invalid.
(3) Suppose $f(c,b,t) = f(c,b',t) = 1$, for some distinct jobs $b,
b'$. This is impossible as it would make $S$ violate the display constraint $k=1$, as
it would recommend both $i^b_t$ and $i^{b'}_t$ to user $c$ at time $t$.
(4) Finally, we shall show that for each $(c,b) \in C \times B$, there are exactly $R(c,b)$ hours, $t$, for which $f(c,b,t) = 1$.
We will prove this by contradiction.
The key observation is that every recommendation in $S$ must make a non-zero revenue.
Specifically, $N$ of them should make a revenue of $1$ each and $\Upsilon$ of them should make a revenue of $E$ each. It is easy to see that if this is not the case, then $\pi(S) < N + \Upsilon E$.

The following cases arise. 

\noindent\underline{Case 1}: $R(c,b) = 1$.
Suppose $f$ does not assign $c$ to job $b$ even once.
In this case, either $S$ has fewer than $N$ recommendations making a revenue of $1$ each or has a recommendation making no profit.
In both cases, $\pi(S) < N + \Upsilon E$. 
Suppose $f$ assigns $b$ to $c$ more than once, say $f(c,b,t) = f(c,b,t') = 1$, and assume w.l.o.g.\ that $t < t'$. Then $S$ must contain the recommendations, say $(c,i^b_t,t)$ and $(c,i^b_{t'},t')$.
By construction, however, $\dap_S(c,i^b_{t'},t') = 0$ so this is a wasted recommendation. 
Thus, $\pi(S) < N + \Upsilon E$.

\noindent\underline{Case 2}: $R(c,b) = 0$. 
Suppose $f(c,b,t) = 1$ for some $t\in H$. This implies $(c,i^b_t,t) \in S$, but $\dap_S(c,i^b_t,t) = 0$ by construction, which again makes it a wasted recommendation, so $\pi(S) < N + \Upsilon E$.

This shows that the timetable function $f$ is indeed feasible, which completes the proof. 
\end{proof} 

We stress that the proof does not involve saturation, which means even when the saturation parameter $\beta_i = 1$ for all $i\in I$ (no saturation at all!), \revmax remains NP-hard.

%Our proof made use of the notion of item classes, with the natural
%property that a user adopts at most one item from a class of similar (and thus competing) items.

%\subsection{Model Parameters Discussion}
%\note[Wei]{Better move to exp section.}

\section{Theoretical Analysis \& Approximation Algorithm}\label{sec:theory}
%\note[Wei]{Need significant revision.}

\eat{Despite solving \revmax optimally is not possible in
	polynomial time unless P = NP,} 
Given the hardness of \revmax, a natural question is whether we can design an approximation algorithm.
In this section, we present an approximation algorithm by establishing connections to the problem of
	\emph{maximizing non-negative submodular functions subject to a 
	matroid constraint}.
To this end, we first show that the revenue function $\Rev(\cdot)$
	is submodular (albeit non-monotone).
Next, we transform the display constraint of \revmax into a {\em partition matroid} constraint. 
Finally, we ``push'' a ``smoothened'' version of the capacity constraint into the objective function. This results in a version of the \revmax problem that has a relaxed capacity constraint. Solving it reduces to maximizing a non-negative non-monotone submodular function under a matroid constraint. This can be approximated to
	within a factor of $1/(4+\epsilon)$, for any $\epsilon > 0$,
	achieved by applying a local search algorithm due to Lee et al.\
	\cite{lee10}. 
%\note[Laks]{Fix the approx. factor above.} 

%Approximation guarantees are obtained by establishing a connection to
%	matroid theory, which implies that the local search algorithm
%	proposed by Lee et al.~\cite{lee10} gives a $1/3$-approximation
%	(Sec.~\ref{sec:approx}).

\subsection{Submodularity of Revenue Function}\label{sec:submod}
We show that the revenue function $\pi(\cdot)$ is submodular.
Let $X$ be any finite ground set.
A non-negative set function $f: 2^X \to \mathbb{R}_{\ge 0}$
	is \emph{monotone} if $f(S) \leq f(S')$ whenever
	$S\subseteq S'\subseteq X$.
Also, $f$ is \emph{submodular} if for any two sets
	$S\subset S'\subseteq X$ and any $w\in X\setminus S'$,
	$f(S\cup\{w\}) - f(S) \geq f(S'\cup\{w\}) - f(S')$.
Submodularity captures the principle of \emph{diminishing marginal
	returns} in economics. % and provides a way of obtaining approximation guarantees.
$\Rev(\cdot)$ is nonnegative by definition. %since the revenue of any strategy can never be negative.
Next, we describe how to compute the marginal revenue
	of a triple $z=(u,i,t)$ w.r.t.\ a strategy $S$, denoted
	$\Rev_S(z) \eqdef \Rev(S\cup\{z\}) - \Rev(S)$.
This is useful both for showing submodularity and
	illustrating algorithms in \textsection\ref{sec:algo}.

The marginal revenue $z$ brings to $S$ consists of two parts:
	a {\em gain} of $\dap_{S\cup\{z\}}(z) \cdot \price(i,t)$ from $z$
	itself, and a {\em loss} as adding $z$ to $S$ drops the dynamic adoption
	probability of triples $(u,j,t')\in S$ with $j \in \cl{C}(i)$ and $t' > t$.
%If $t' < t$, no effective change occurs: $\dap_{S\cup\{z\}}(u,j,t') = \dap_S(u,j,t')$.
More specifically, $\dap_{S\cup\{z\}}(u,j,t') \leq \dap_S(u,j,t')$
	due to the increased memory caused by $z$ and the interdependency rule.
%As a result, their marginal contribution to total revenue drops.
%On the other hand, $z$ itself brings a positive marginal revenue
%	which is equal to $\dap_{ui}(t, S\cup\{z\}) \cdot \price_i(t)$.
The exact formula for computing $\Rev_S(z)$ is as follows.

\begin{definition}[Marginal Revenue]
Given any valid strategy set $S$ and a triple $z=(u,i,t)\not\in S$, the marginal revenue of $z$ with respect to $S$ is
\begin{align}\label{eqn:mr}
&\Rev_S(z) = \price(i,t) \times \dap_{S\cup\{z\}}(z)  \\
&+ \price(i,t') \times \sum\nolimits_{\substack{(u,j,t') \in S: \\ j\in \cl{C}(i), t'>t}} (\dap_{S\cup\{z\}}(u,j,t') - \dap_S(u,j,t'))  . \nonumber 
\end{align}
%where $\List_{ui}(S) = \{t\in [T] | (u,i,t)\in S\}$. 
\end{definition}

%We now state and prove the submodularity of the objective function, together with a useful lemma below.

It can be shown that the revenue function satisfies submodularity.
\begin{theorem}\label{thm:sm}
The revenue function $\Rev(\cdot)$ is non-monotone and submodular
	w.r.t.\ sets of user-item-time triples.
\end{theorem}

First, we give a useful lemma that will be used later for proving Theorem~\ref{thm:sm}.
\begin{lemma}\label{lemma:dap}
Given any triple $(u,i,t)\in S$, its dynamic adoption probability
	$\dap_S(u,i,t)$ is \emph{non-increasing} in the strategy set $S$.
That is, given two strategy sets $S, S'$ such that $(u,i,t)\in S\subseteq S'$, we have
$\dap_S(u,i,t) \geq \dap_{S'}(u,i,t). $
\end{lemma}

\begin{proof}
By definition, all triples with $t'>t$ do not affect
	dynamic adoption probability at $t$.
Since every $(u,j,*)$-triple with $t'<t$ and $j\in \cl{C}(i)$ in $S$
	is also in $S'$, by Equation \eqref{eqn:finalap}, the lemma follows
	naturally.
\end{proof}

\begin{proof}[of Theorem~\ref{thm:sm}]
For non-monotonicity, consider an instance which has 
$U=\{u\}$, $I=\{i\}$, $T=2$, $k=1$, $q_i=2$,
	$\ap(u,i,1) = 0.5$, $\ap(u,i,2) = 0.6$, $\price(i,1) = 1$, $\price(i,2) = 0.95$,
	and $\beta_i = 0.1$.
Consider two sets $S=\{(u,i,2)\}$ and $S'=\{(u,i,1), (u,i,2)\}$.
Clearly $\Rev(S') = 0.5285 < \Rev(S) = 0.57$, implying that $\Rev(\cdot)$ is non-monotone. 
%Submodularity can be shown by simply using the definition and case analysis.
%For lack of space, we omit the proof and will include it in a full-version technical report.

For submodularity, consider two strategies $S \subseteq S'$, and a triple $z =(u,i,t)\not\in S'$.
We say that $z$ \emph{succeeds} triple $z'=(u', i', t')$ if $t > t'$, or \emph{precedes} $z'$ if $t < t'$.
The following cases arise.

\underline{Case 1}:
$z$ succeeds all $(u,j)$-triples in $S'$ such that $\cl{C}(j) = \cl{C}(i)$.
Since $S\subseteq S'$, $z$ succeeds all such triples in $S$, too, and thus adding
$z$ to $S$ or $S'$ will not cause any loss in revenue.
By Lemma~\ref{lemma:dap}, $\dap_{S'\cup\{z\}}(z) \leq \dap_{S\cup\{z\}}(z)$,
and thus $\Rev_S(z) \geq \Rev_{S'}(z)$.

\underline{Case 2}:
$z$ precedes all $(u,j)$-triples in $S'$ such that $\cl{C}(j) = \cl{C}(i)$.
Since $S\subseteq S'$, $z$ also precedes all such triples in $S$,
in which case $z$ brings the same amount of revenue gain to both sets
since $\dap_{S\cup\{z\}}(z) = \dap_{S'\cup\{z\}}(z) = \ap(z)$.
However, the number of triples $z$ precedes in $S'$ is no less than that in $S$,
	so is the revenue loss $z$ causes.
Thus, $\Rev_S(z) \geq \Rev_{S'}(z)$.  

\underline{Case 3}:
$z$ precedes a few $(u,j)$-triples in $S'$ where $\cl{C}(j) = \cl{C}(i)$
	and succeeds the rest of such triples in $S'$.
The argument for this case combines the reasoning for the two cases above.
First, for revenue gain, by Lemma~\ref{lemma:dap},\
$\dap_{S'\cup\{z\}}(z) \leq \dap_{S\cup\{z\}}(z)$.
Second, for revenue loss, the number of triples $z$ precedes in $S'$ is
	no less than that in $S$, so is the revenue loss.
Hence $\Rev_S(z) \geq \Rev_{S'}(z)$ holds.

This completes the proof.
\end{proof}

Theorem~\ref{thm:sm} not only helps us in the approximation analysis
	below, but also lays the foundation for the lazy forward technique
	used in our greedy algorithm in Section \ref{sec:algo}.

\subsection{Matroid and Approximation Guarantees}

Matroids are a family of abstract structures widely studied in
	combinatorics.
A {\em matroid} $\M$ is a pair $(X,\I)$ where $X$ is a finite
	ground set and $\I \subseteq 2^X$ is a collection of subsets of
	$X$ called \emph{independent sets}, whose membership satisfies:
	(1).\ $\emptyset\in \I$;
	(2). Downward closure: \ $\forall T\in \I$, $S\subset T$ implies $S\in \I$;
	(3). Augmentation: $\forall S, S'\in \I$, if $|S| < |S'|$, then
	$\exists w\in S'\setminus S$ s.t.\ $S\cup\{w\}\in \I$.

The problem of maximizing $f: 2^X\to \mathbb{R}_+$
	subject to a matroid constraint $\M=(X,\I)$ is to find
	$S^* \in \argmax_{S\in \I} f(S)$, which is in general NP-hard 
	for nonnegative submodular functions~\cite{lee10}.
%%%%%%%%%%%%% 
\eat{Furthermore, the problem can be extended to having multiple matroid
	constraints: let $\M_i=(X,\I_i)$, $i=1,2, \dotsc, n$, be $n$
	matroids on $X$, and the task is to find $S^* \in \argmax_{S\in \cap_{i=1}^{n} \I_i} f(S)$.
For our purpose, we will show that the display constraint ($C_1$) and
	the capacity constraint ($C_2$) can be transformed into a
	partition matroid and projection matroid (which we will define
	shortly) respectively.
} 
%%%%%%%%%%%%%

We deal with the display constraint first. 
%\spara{Display Constraint and Partition Matroid}
A \emph{partition matroid} $(X, \I)$ is defined by $m$ disjoint subsets of $X$
	with $\cup_{j=1}^{m} X_j = X$, with each $X_j$ having a maximum
	cardinality constraint $b_j$.
$\I$ contains $S\subseteq X$ iff $|S\cap X_j| \le b_j$ for all $j$.

\begin{lemma}\label{lemma:pat-mat}
The display constraint in \revmax is equivalent to a partition
	matroid constraint.
\end{lemma}

\begin{proof}
Let the ground set $X= U\times I\times [T]$.
We project $X$ onto all user-time pairs
	$(u^*, t^*)\in U\times [T]$ to obtain a collection of subsets
	$X(u^*, t^*) := \{(u,i,t) | u = u^*, i\in I, t = t^*\}$.
Clearly, the sets  $X(u,t)$, $u\in U, t\in [T]$,  are pairwise disjoint and
	$\cup_{u,t} X(u,t) = X$.
Next, set $b(u,t) = k$ for all $(u,t)$.
This gives a partition matroid $\M = (X, \I)$ and any
	$S\subseteq X$ satisfies the display constraint iff $S\in \I$.
\end{proof}

%\noindent\textbf{Handling the Capacity Constraint.}
Unlike the display constraint, the capacity constraint does not correspond to a matroid constraint. The observation is that while downward closure and inclusion of empty set are satisfied, augmentation is not satisfied by the capacity constraint, as we next show with an example. 

\begin{example}[Capacity Constraint] 
\label{ex:projection} 
Consider two strategies $S' = \{(u_1, i_2, t_1), (u_1, i_2, t_2), (u_2, i_1, t_1), (u_2, i_1, t_2)\}$ and $S = \{(u_1, i_1, t_1), (u_2, i_2, t_2)\}$. Assume the display constraint is $k=1$ and the capacity constraint is $q_{i_1} = q_{i_2} = 1$.
While $|S'| > |S|$, there is no triple in $S'\setminus S$ which can be added to $S$ without violating the capacity constraint. 
\end{example} 

Thus, we resort to a different method for obtaining an approximation algorithm. In \revmax, the capacity constraint is a hard constraint: no item $i$ can be recommended to more than $q_i$ distinct users. This can in principle result in fewer than $q_i$ adoptions, because of the inherent uncertainty involved. Consider making a calculated move in making a few more recommendations than the capacity would allow, should all recommended users wish to adopt. Since they are not all guaranteed to adopt, it is possible that such a strategy may result in a greater expected profit than a strategy that \revmax would permit. We next formalize this idea. 

Let $S$ be any strategy and $(u,i,t) \in S$ be a triple.
If $i$ is not recommended to more than $q_i$ distinct users up to time $t$, $\dap_S(u,i,t)$ should remain exactly the same as in \eqref{eqn:finalap}. % [Check!].
Suppose $i$ has been recommended to $\ge q_i$ users \emph{besides} $u$ up to time $t$. Two cases arise w.r.t. adoptions. If $q_i$ of the previously recommended users adopt $i$, then we must conclude $u$ cannot adopt it since it is not available any more!
On the other hand, if fewer than $q_i$ previously recommended users adopt $i$, then $u$ may adopt $i$ with probability $\dap_S(u,i,t)$. We refer to this relaxed version of \revmax as $R$-\revmax. Compared to \revmax, $R$-\revmax has no hard capacity constraint on number of recommendations of an item. This constraint is essentially  ``pushed'' inside the \emph{effective dynamic adoption probability} for $R$-\revmax, defined as follows. 

\begin{definition}[Effective Dynamic Adoption Prob.] 
\label{def:edap} 
Let $S$ be a strategy and suppose $(u,i,t)\in S$ is a triple and let $S_{i,t} = \{(v,i,\tau)\in S \mid v\neq u, \tau\leq t\}$ be the set of recommendations of $i$ made to users other than $u$ up to time $t$.
Let $\prob_S(i,t) = \Pr[\mbox{at most } (q_i-1) \mbox{ users in } S_{i,t} \mbox{ adopt i}]$.
Then the {\em effective dynamic adoption probability} of $(u,i,t)$ given $S$ is: 

\begin{align}\label{eqn:edap}
&\edap_S(u,i,t) := \ap(u,i,t) \times \beta_i^{M_S(u,i,t)}
\times \prod_{\substack{(u,j,t) \in S \\ j\neq i, \cl{C}(j) = \cl{C}(i)}} (1-\dap(u,j,t)) \nonumber \\
  &\quad \times \prod_{\substack{(u,j,\tau) \in S, \tau < t, \cl{C}(j) = \cl{C}(i)}} (1-\dap(u,j,\tau)) \times \prob_S(i,t)  %\nonumber \\ 
%&\quad \Pr[\mbox{at most } (q_i-1) \mbox{ users in } S_{i,t} \mbox{ adopt i}]. 
\end{align}

\end{definition}

The following example illustrates Definition~\ref{def:edap}.

\begin{example}[Exceeding Capacity]\label{ex:excap} 
Consider one item $i$, three users $u,v,w$, display constraint $k=1$, capacity $q_i = 1$, and saturation parameter $\beta_i = 0.5$. Consider the strategy $S = \{(u,i,1), (v,i,2), (w,i,1), (w,i,2)\}$. Then the effective dynamic adoption probability of $(w,i,2)$ is $\edap(w,i,2) = \dap(w,i,2) \times (1-\dap(w,i,1)) \times (1-\dap(u,i,1)) \times  (1-\dap(v,i,2)) \times 0.5^{1/1}$.
\end{example} 

The only change to the definition of $\edap_S(u,i,t)$ in \eqref{eqn:edap} compared with that of $\dap_S(u,i,t)$ is the factor $\prob_S(i,t)$.
If $|S_{i,t}| < q_i$, this probability is $1$. Computing it exactly in the general case can be hard but can be computed exactly in worst-case exponential time in $q_i$. We can use Monte-Carlo simulation for estimating this probability. The point is that given an oracle for estimating or computing probability, we can define $R$-\revmax as follows. 

An instance of $R$-\revmax is identical to that of \revmax without a hard capacity constraint. The revenue function $\Rev(S)$ is defined exactly as in \eqref{eqn:rev}, except that instead of $\dap_S(u,i,t)$, $\edap_S(u,i,t)$ is used.
A strategy is now called valid if it satisfies the display constraint, and the problem is to find a strategy $S^*$ that maximizes $\Rev(S)$ amongst all valid strategies $S$. % $S^* \in \argmax\{\Rev(S) \mid S \mbox{ is a valid strategy}\}$. 

Given an oracle for computing or estimating $\prob_S(i,t)$, we can devise an approximation algorithm for $R$-\revmax as follows.
As shown in Lemma \ref{lemma:pat-mat}, the display constraint corresponds to a matroid constraint. The revenue function $\Rev(S)$ for $R$-\revmax can be easily shown to be non-monotone and submodular.  %along the same lines as the proof of Theorem \ref{thm:sm}.
Thus, \textsl{solving $R$-\revmax corresponds to maximizing a non-negative non-monotone submodular function subject to a partition matroid constraint.}
This can be solved using local search to give a $1/(4+\epsilon)$-approximation, for any $\epsilon > 0$ \cite{lee10}. 
However, unfortunately, the time complexity of this approximation algorithm, even assuming unit cost for each invocation of the oracle, is $O(\frac{1}{\epsilon}|X|^4 \log|X|)$ where $X = U \times I \times [T]$ in our case.
This is prohibitive for our application and motivates the quest for good and efficient heuristics that perform well in practice, a topic addressed in the next section.

\section{Greedy Algorithms}\label{sec:algo}
We propose three intelligent greedy algorithms for \revmax: {\em Global Greedy} (\GGreedy),\
	{\em Sequential Local Greedy} (\SGreedy), and {\em Randomized Local Greedy} (\RGreedy). %, which are more efficient and scale better to large data
%compared to the local search approximation algorithm. 
They all start with an empty strategy set and incrementally grow it in a greedy manner.
As we shall see shortly, the main difference is that  \GGreedy operates on the entire ground set $U \times I \times [T]$ and
	makes recommendations disregarding time order, while \SGreedy and \RGreedy finalize recommendations in a predetermined chronological order. 
%We demonstrate the efficiency and effectiveness of the algorithms in
%	Sec.~\ref{sec:exp}.

\subsection{The Global Greedy Algorithm} \label{sec:gg}

\spara{Overview}
We first give a natural hill-climbing style algorithm called Global Greedy (\GGreedy for short; pseudo-code in Algorithm~\ref{alg:gg}).
%, which	grows the strategy set $S$ one triple at a time.
In a nutshell, the algorithm starts with $S$ being $\emptyset$, and in
	each iteration, it adds to $S$ the triple that  provides the
	largest positive marginal revenue w.r.t.\ $S$ without violating 
	the display or capacity constraint. % after being added to $S$.
More formally, let $\mathcal{V}(S) \subset U \times I \times[T]$ be the set of triples which,
	when added to $S$, would {\em not} violate the display or capacity constraints.
Thus, in every iteration, \GGreedy selects the triple satisfying:
\begin{align}\label{eqn:bestMG}
%z^* \in \argmax_{(u,i,t) \in \mathcal{V}(S) \setminus S: \Rev_S((u,i,t)) > 0} \Rev_S((u,i,t)),
z^* \in \argmax \{\Rev_S(z) > 0 \mid z \in \mathcal{V}(S) \setminus S\}.
\end{align} 
%as long as $\Rev_S((u,i,t))$ is positive.
Recall that $\Rev_S(z)$ represents the marginal revenue of $z$ w.r.t.\ $S$, defined as $\Rev(S\cup\{z\}) - \Rev(S)$ ({\em cf.} Equation \eqref{eqn:mr}).
Also, we use priority queues in the implementation to 
	support efficient operations in the greedy selection process.

To enhance efficiency, we also employ two implementation-level optimizations.
First, we propose the idea of {\em two-level heaps} data structure to reduce the overhead of heap operations in the greedy selection process.
Second, the \emph{lazy forward} scheme \cite{minoux78, leskovec07} is used for computing and ranking triples w.r.t. their marginal revenue.

The algorithm also maintains several auxiliary variables to facilitate
constraint enforcement.
First, counter variables are used to keep track of the
	number of items recommended to each user $u$ at each time $t$,
	facilitating the enforcement of display constraint.
Second, we also keep track of the set of users to whom item $i$
	has been recommended so far, facilitating the enforcement of capacity constraint.
Third, for each user $u$ and each item class $c$, we bookkeep the set
	of triples in $S$ involving $u$ and items of class $c$.
That is, $\List(u,c) := \{(u,i,t) \in S \mid \cl{C}(i) = c, t\in [T]\}$.
This is needed for marginal revenue computation and lazy forward.

\begin{algorithm}[t!]\label{alg:gg}
\caption{\GGreedy (Two-Level Heaps \& Lazy Forward)}
\SetKwInOut{Input}{Input}
\SetKwInOut{Output}{Output}
\Input{$U, I, T, k, \{q_i\}, \{\price(i,t)\}, \{\ap(u,i,t)\}, \{\beta_i\}$.}
\Output{A valid strategy $S \subseteq U\times I\times [T]$.}
	
$S \gets \emptyset$; \tcc*[f]{initialization}\; 
$\uheap \gets$ an empty maximum binary heap\;
\ForEach{$(u,i)\in U\times I$ {\em such that} $\exists \ap(u,i,t) > 0$} {
	$\lheap_{u,i} \gets$ an empty maximum binary heap\;  \label{line:gg1}
}
%\lForEach{$(u,t)\in U\times [T]$} {
%	$\cnt_{u,t} \gets 0$\;
%}
%\lForEach{$i\in I$}{
%	$\userSet_i \gets \emptyset$\;
%}
\ForEach{$u \in U$, {\em item class} $c$}{
	$\List_{u,c} \gets \emptyset$\;
}

\ForEach{$(u,i,t) \in U\times I\times [T]$ {\em with} $\ap(u,i,t) > 0$} {
	$\lheap_{u,i}$.\Insert($(u,i,t), \ap(u,i,t) \cdot \price(i,t)$) \;  \label{line:gg2} 
	$\flag((u,i,t)) \gets 0$; \tcc*[f]{for lazy forward}\;  \label{line:gg2b} 
}
$\uheap$.Heapify() \tcc*[f]{populate and heapify it with roots of all lower-level heaps}\;  \label{line:gg3} 
	
\While {$|S| < k\cdot T\cdot |U| \wedge \uheap$ {\em is not empty} } {
	$\Top \gets \uheap$.FindMax(); \tcc*[f]{root}\;
	\lIf{$\Rev_S(z) < 0$} {\textbf{break}; \tcc*[f]{negative case}\;}
	\If {$S\cup \{\Top\}$ {\em doesn't violate any constraint}} {  
		\If{$\flag(\Top) < |\List_{\Top.u, \cl{C}(\Top.i)}|$} {  \label{line:gg4a}
			\ForEach{\emph{triple $z' \in \lheap_{\Top.u, \Top.i}$}} {
				calculate $\Rev_S(z')$;  \tcc*[f]{Eq.~\eqref{eqn:mr}}\;  
				$\flag(z') \gets |\List_{\Top.u, \cl{C}(\Top.i)}|$\;
			}
			update $\lheap_{\Top.u, \Top.i}$ and $\uheap$; \tcc*[f]{using Decrease-Key on heaps}\;  \label{line:gg4b}
		}
		\ElseIf{$\flag(\Top) == |\List_{\Top.u, \cl{C}(\Top.i)}|$}{       \label{line:gg5a}
			$S \gets S\cup \{\Top\}$\;						
			$\List(\Top.u, \cl{C}(\Top.i))$.Add($\Top$)\;
			$\uheap$.DeleteMax()\;  						 \label{line:gg5b} 
		}
	} \Else(\tcc*[f]{remove from considerations}){
              $\uheap$.DeleteMax()\;
		delete $\lheap_{\Top.u, \Top.i}$\;
       }
}
\end{algorithm}

\spara{Two-Level Heaps Data Structure}
%The two-level heap structure is as follows.
For each user-item pair $(u,i)$ with a positive primitive adoption probability for some time step,
	we create a priority queue (implemented as a maximum binary heap)
	to store the marginal revenue of all triples $(u,i,*)$, where $*$ denotes
	any applicable time steps (line~\ref{line:gg1}).
Such heaps are initially populated with the revenue of all triples
	computed using primitive adoption probabilities (line~\ref{line:gg2}).
They form the {\em lower level} and have no direct involvement in
	seed selection.
Then, note that the best candidate triple at any point of the execution
	must be the root of one of those heaps.
Precisely, it is the one with largest marginal revenue amongst all roots.
Hence, we sort the roots of all lower-level heaps in a master, {\em upper-level}
	priority queue, which is directly involved in seed selection (line~\ref{line:gg3}).
Ties are broken arbitrarily.

The intuition is that if we were to use one ``giant'' heap that contains all
	triples, it would incur larger overhead in heap operations like Decrease-Key, or
	Delete-Max, as updated keys will have to traverse a %much 
	taller binary tree.
Conversely, in the two-level structure, each low-level heap contains at most $T$ elements,
	and thus the overhead will be almost negligible as long as $T$ is reasonable ($7$ in our experiments),
	while the upper-level heap has at most $|U|\cdot|I|$ elements, a factor of
	$T$ smaller than the ``giant one''.
%\textcolor{red}{Indeed, empirical results confirm this intuition and show that the running time can be reduced
%	by approximately $2 x$ (\textsection\ref{sec:exp}).
%}
%\note[Wei]{If we only compare GG and GG' in responses, but not the paper, this red sentence should be removed as well.}

\spara{Lazy Forward and Greedy Selection Details}
%The intuitions are as follows.	
By model definition, observe that after a triple is added to $S$, the marginal
	revenue of all triples with the same user and same class of items should be updated before they can be considered for selection.
For each update made, a {\em Decrease-Key} operation is needed
	in both the lower-level and upper-level heap to reflect the change.
However, if a triple's marginal revenue is small, chances are it will never
	be percolated up to the root of the upper-level heap. % (and hence be considered as the candidate for next selection).
For such triples, an eager update will be wasted and result in inefficiency.
Thanks to submodularity of the revenue function (Theorem~\ref{thm:sm}),
	it is possible to avoid unnecessary computation by using the
	{\em lazy forward} optimization proposed in \cite{minoux78} and  recently used in \cite{leskovec07}.

More specifically, we associate with each triple $z$ a flag variable,
	$\flag(z)$, initialized as $0$ (line \ref{line:gg2b}).
When a triple $z=(u,i,t)$ is percolated up to the root of the upper-level heap,
	we first examine if adding it to $S$ will violate any constraint.
If no violation is caused, and $\flag(z) = |\List(u, \cl{C}(i))|$ holds,
	then $z$'s marginal revenue is {\em up-to-date} and
	will be added to $S$ (lines \ref{line:gg5a} to \ref{line:gg5b}).
If, however, $\flag(z) < |\List(u, \cl{C}(i)))|$,
	then we retrieve the corresponding lower-level heap, re-compute
	all stale triples $\not\in S$, and insert the 
	updated root back to the upper-level heap (lines \ref{line:gg4a} to \ref{line:gg4b}).

%%%%%%%%%%%%%%%%%%%%%%%%%%%%%%
\eat{
Triple selection starts from line \ref{gg5}: the
	algorithm keeps examining the root element of the heap, denoted $\Top$.
If $\flag(\Top) < |\List(u,\cl{C}(i))|$, then $\Top$ is deemed
	{``dirty''} as its marginal revenue is not up-to-date.
	w.r.t.\ the current $S$. % and hence should be re-computed.
We then invoke Equation \eqref{eqn:mr} to compute the latest marginal revenue
	and re-heapify the heap.
%Also we set its flag value to $|\List(u,\cl{C}(i))|$ (lines \ref{gg6} to \ref{gg7}).
If $\flag(\Top) = |\List(u,\cl{C}(i))|$, meaning the marginal
	revenue is up-to-date, then we add $\Top$ to $S$ as long as neither
	constraint is violated  and increment corresponding
	counter variables (lines \ref{gg8} to \ref{gg9}).
}
%%%%%%%%%%%%%%%%%%%%%%%%%%%%%%

%It is sound to use lazy forward in \GGreedy due to
%	submodularity:
%Given two strategies $S\subset S'$ and any triple $z\notin S'$,
%	$\Rev_S(z) \ge \Rev_{S'}(z)$.
%Thus, if the marginal revenue of a triple is not the highest
%	after a triple $(u,j,\tau)$ with $\cl{C}(j) = \cl{C}(i)$ is added
%	to $S$, $(u,i,t)$ should not be added to $S$ and its marginal
%	revenue will only go down as \GGreedy continues growing $S$.

The soundness of lazy forward in \GGreedy stems from submodularity.
More specifically, if the root's flag is up-to-date, then its marginal revenue
	is indisputably the highest regardless of if others are up-to-date
	(consider any triple ranked lower that is not up-to-date, its actual marginal revenue
	can be no more than the stale value, due to submodularity).
A similar idea was used in the greedy algorithms proposed in
	\cite{minoux78,leskovec07}, where
	flags are compared to the overall solution size $|S|$.
But in our case, the revenues from different user-class
	pairs do not interfere, so we need to check flag values against the size
	of the corresponding $\List$ ({\em cf.} lines \ref{line:gg4a} and \ref{line:gg5a}).
%Also, whenever a triple is selected
%	into $S$, the corresponding $(u,c)$-set is immediately
%	updated and we always check flag values before adding any new
%	element to $S$. 

\spara{Termination}
The algorithm terminates when one of the following conditions is met:
(i).\ The upper-level heap is exhausted (empty);
(ii).\ All users have received $k$ recommendations in all time steps;
(iii).\ None of the remaining triples in upper-level heap has a positive marginal revenue.\ 
Regardless of the sequence in which the output strategy
	$S$ is formed, the final recommendation results are presented to
	users in natural chronological order.

\spara{Space and Time Complexity}
The upper-level heap has at most 
$|U| \times |I|$ triples, while the lower-level
	heaps, between them, have at most $|X| = |U| \times |I| \times T$.
Thus total space complexity is $O(|X|)$.
However, users are typically interested merely in a small subset of items,  and thus \textsl{the actual triples in consideration can be much fewer than} $|U|\times |I|\times T$.
For time complexity, it is difficult to analytically estimate how many calls to marginal revenue re-computations
	are needed with lazy forward, thus we give a \emph{worst-case upper bound} using
	the number of calls \textsl{would have been made without lazy forward}.
	Let $y = kT|U|$ be the total number of selections made and the total time complexity is thus $O(y(|I| T \log T + |I| \log(|U|\times |I|)  ))$.
%Therefore, there would be $O(yT)$ calls, each of which costs $O(T^2 + \log x)$. 
% in which case  at most $O(kT^2|U|)$.
%In total the time complexity is $O(kT^2|U|)$.
%Each call of marginal revenue computation runs in a $O(T^2)$  (which dominates Decrease-Key in lower-level heaps since it is only $O(\log T)$) and
This expression is an upper bound corresponding to the worst case scenario where lazy forward is not used (as it is difficult to reason about its actual savings) and all items are in one class (that is, after each addition to $S$, all $|I|$ lower level heaps associated with $u$ need to be updated).  
That said, we expect \GGreedy to run much faster in practice because of lazy-forward (a \emph{700 times speedup} was reported in \cite{leskovec07}).

%Additionally, we anticipate that recommendation strategies will be planned over relatively short time horizons, making $T$ small, since planning for long time horizons based on presently available feedback may render the strategy highly sub-optimal. 

\eat{
Let $X \eqdef \{(u,i,t) \in U\times I\times [T] \mid \ap(u,i,t) >0\}$.
The size of the heap is $O(|X|)$, i.e., effectively linear in the input size.
The vanilla \GGreedy without lazy forward needs at most
	$O(T|X|)$ calls to the marginal revenue computation
	procedure, each of which takes time $O(T^2)$;
	it also needs $O(|X|)$ heap operations, each of
	which takes $O(\log|X|)$.
Therefore, the overall complexity is $O(|X|\log|X| + T^3 |X|)$.
Note that usually users are only interested in a subset of
	items much smaller than $|I|$, so in practice
	the term $|X|$ tends to be a small fraction of $|U\times I\times [T]|$. 
Additionally, we anticipate that recommendation strategies will be planned over relatively short time horizons, making $T$ small. The reason is that planning a strategy for long time horizons based on presently available feedback may make the strategy sub-optimal. 
} 
%\GGreedy can be very efficient (\cf Sec.~\ref{sec:exp}).

\subsection{Two Local Greedy Algorithms}

\GGreedy evaluates all candidate triples together and 
	determines a  recommendation strategy in a holistic manner.
%This comes at the expense of increased space and time complexity. 
\eat{
may lead to high space complexity ({\em cf.} all lower-level heaps),
	which in turn increases running time due to heavier
	heap operations.
}
It is interesting to investigate more lightweight heuristics 
	hopefully leading to similar performance. We propose two ``local''
	greedy algorithms that generate recommendations on a
	{\em per-time-step} basis.
\eat{In contrast to \GGreedy which ``jumps'' freely in
	the entire time horizon, a local greedy algorithm}
Unlike \GGreedy, these algorithms first finalize 
	$k$ recommendations to all users for a {\sl single} time step $t$, 
	before moving on to another time step $t'$ until recommendations for
	all $T$ time steps are rolled out.
%This leads to a heap whose size is smaller than that of \GGreedy by a factor of $T$. % and thus gives
	%improvement in efficient. 

%\note[Laks]{We use time horizon and time horizon interchangeably. Horizon sounds better but there may be more occurrences of period.} 

\begin{algorithm}[t!]
\caption{\SGreedy} \label{alg:sg}
\SetKwInOut{Input}{Input}
\SetKwInOut{Output}{Output}
\Input{$U, I, T, k, \{q_i\}, \{\price(i,t)\}, \{\ap(u,i,t)\}, \{\beta_i\}$.}
\Output{A valid strategy $S \subseteq U\times I\times [T]$.}

$S \gets \emptyset$;  \tcc*[f]{initialization}\; 
\ForEach{$u \in U$, {\em item class} $c$}{
	$\List_{u,c} \gets \emptyset$\;
}
	
\For{$t = 1$ to $T$} {
	$\heap \gets$ an empty maximum binary heap\; \label{line:sg1}
	%\lForEach{$u\in U$} {
	%	$\cnt(u) \gets 0$\;
	%}
	\ForEach{$(u,i,t) \in U\times I\times [T]$} { 
		compute $\dap_S(u,i,t)$; \tcc*[f]{Eq.~\eqref{eqn:finalap}} \label{sg3} 
		$\heap$.\Insert($(u,i,t), \price(i,t) \times \dap_S(u,i,t)$) \;
		%$\heap$.\Heapify()\; 
	}

	\While{$\heap$ {\em is not empty} } { \label{sg4}
		$\Top \gets \heap$.FindMax(); \tcc*[f]{root of heap}\;
		\lIf{$\Rev_S(\Top) \le 0$} {
			break\;
		}
		\If{$S\cup \{\Top\}$ \emph{does not violate either constraint}} {
			$S \gets S\cup \{\Top\}$\;
			%$\cnt(\Top.u) \gets \cnt(\Top.u) + 1$\;
			%\If{$\Top.u\not\in \userSet(\Top.i)$}{
			%	$\userSet(i) \gets \userSet(i) \cup \{u\}$\;
			%}
			$\List(\Top.u, \cl{C}(\Top.i)).\Add(\Top)$\;
			Compute $\Rev_{S\cup\{\Top\}}(\Top.u, j, t)$, $\cl{C}(j) = \cl{C}(\Top.i)$\;
		}
		$\heap$.\Remove()\; \label{line:sg0}
	}
}
\end{algorithm}

%%%%%%%%%%%%%%%%%%%%%%%%%%
%%%%%%%%%%%%%%%%%%%%%%%%%%
\eat{
\begin{algorithm}[t!]
\caption{\RGreedy}
\label{alg:rg}
\SetKwInOut{Input}{Input}
\SetKwInOut{Output}{Output}
\Input{$U, I, T, k, \{q_i\}, \{\price(i,t)\}, \{\ap(u,i,t)\}, \{\beta_i\}$, $N$}
\Output{A valid strategy $S \subseteq U\times I\times [T]$.}
	
Randomly generate $N$ distinct permutation of $[T]$\;
%\ForEach{sampled sequence $\Seq_j$} 
\For{$j\gets 1$ \KwTo $N$}
{
	$S_j \gets \emptyset$\; \label{rg1}
	%\lForEach{$i\in I$}{
	%	$\userSet(i) \gets \emptyset$\;
	%}
	\ForEach{$u \in U$, {\em item class} $c$}{
		$\List_{u,c} \gets \emptyset$\;
	}
	\For{$\tau \gets 1$ \KwTo $T$} {
		$t \gets \tau$-th time step in the $j$-th permutation\;
		execute line \ref{line:sg1} to line \ref{line:sg0} in Algorithm \ref{alg:sg} with $t$\;
		$S_j \gets S_j \cup$ the set of greedy selections for $t$\;
	}
}
$S \gets \argmax_{j=1\dotso N} \Rev(S_j)$\;
\end{algorithm}
}
%%%%%%%%%%%%%%%%%%%%%%%%%%
%%%%%%%%%%%%%%%%%%%%%%%%%%

\spara{Sequential Local Greedy (\SGreedy)}
As suggested by its name, this algorithm (presented in Algorithm~\ref{alg:sg}) follows the
	natural chronological order $t = 1, 2, \dotsc, T$ to form  recommendations. 
\eat{The pseudo-code is presented in Algorithm~\ref{alg:sg}.} 
Note that the key data structures such as $S$ and $\List_{u,c}$ are
	still maintained as global variables for correct computation of
	marginal revenue.
The outer-loop iterates $T$ times, each corresponding to one time
	step, and a priority queue (maximum binary heap) is used to sort and store marginal revenue values.

In \SGreedy, in each iteration $t$, the heap only needs to store triples
	of $t$ (thus the two-level heaps in \GGreedy are not necessary here) and is initially populated with marginal revenue
	values computed using dynamic adoption probability given $S$, which
	already contains recommended triples up to $t-1$
	(line \ref{sg3}).
The selection procedure is done in a similar fashion to that in \GGreedy, and
lazy forward can be applied within each round (i.e., each single time step, lines \ref{line:sg1}-\ref{line:sg0} in Algorithm \ref{alg:sg}).
For lack of space, the detailed operations of lazy forward (cf. Algorithm~\ref{alg:gg}) is omitted.
%The lazy forward scheme is {\em not} applicable here as the dynamic
%	adoption probability of every triple is only computed once in the
%	round determined by its time coordinate.
%In each round $t$, the size of the max heap is $H \eqdef |\{(u,i,t) \mid u\in U, i\in I, \ap(u,i,t) > 0\}|$. %\footnote{We could drop time from $H$ since in any iteration, $t$ is fixed. We prefer to show it explicitly for clarity.} 
\SGreedy takes $O(|U|\times|I|)$ space and $O(y |I| \log (|U|\times |I|) )$ time (same reasoning has been applied here as in the case of \GGreedy). %, but it could  potentially  be slower than \GGreedy due to the inapplicability of lazy forward.

\spara{Randomized Local Greedy}
The natural time ordering used in $\SGreedy$ (from $1$ to $T$) may not be  
	optimal in terms of revenue achieved.
To illustrate it, we revisit the example used in the proof of Theorem~\ref{thm:sm}, reproduced below for completeness.

\begin{example}[Chronological is Not Optimal] 
\label{ex:randlg}
Let $U=\{u\}$, $I=\{i\}$, $T=2$, $k=1$, $q_i=2$,
	$\ap(u,i,1) = 0.5$, $\ap(u,i,2) = 0.6$,
	$\price(i,1) = 1$, $\price(i,2) = 0.95$,
	$\beta_i = 0.1$.
\SGreedy follows the chronological order $\langle 1,2 \rangle$ and outputs a strategy $S = \{(u,i,1), (u,i,2)\}$ with $\Rev(S) = 0.5285$.
However, if the order of $\langle 2,1 \rangle$ were followed, we would have a better strategy $S' = \{(u,i,2)\}$ with $\Rev(S') = 0.57$.
This is because the marginal revenue of $(u,i,1)$ w.r.t.\ $S'$ is negative and hence will not be added to $S'$.
\end{example}

Ideally, we shall determine an optimal permutation of $[T]$ such that
	the recommendations generated in that ordering yields the best revenue.
However, there seems to be no natural structure which we can exploit to avoid
	enumerating all $T!$ permutations for finding the optimal permutation. 
To circumvent this, we propose \RGreedy that first repeatedly samples $N \ll T!$
	distinct permutations of $[T]$ and executes greedy selection for each permutation.
Let $S_j$ be the strategy set chosen by the $j$-th execution.
In the end, \RGreedy returns the one that yields the largest revenue, i.e., $S = \argmax_{j=1\dotso N} \Rev(S_j)$.
%Algorithm~\ref{alg:rg} presents the pseudo-code of \RGreedy.
%We sketch most of the parts since the algorithm is similar to \SGreedy.
Under any permutation, the seed selection is done on a per-time-step basis, following lines \ref{line:sg1}-\ref{line:sg0} in Algorithm \ref{alg:sg}.
Due to lack of space and similarity to \SGreedy, we omit the pseudo-code of \RGreedy.
\RGreedy is a factor of $N$ slower than \SGreedy due to repeated sampling and has the same space complexity as \SGreedy.

\section{Empirical Evaluation}\label{sec:exp}

\begin{table}[t!]
%\vspace{-5mm}
\centering
\smallskip
	\begin{tabular}{|l | c | c | c|}
		\hline
		 & \textbf{Amazon} & \textbf{Epinions} & \textbf{Synthetic} \\ \hline
		{\bf \#Users} & 23.0K & 21.3K  & 100K -- 500K  \\ \hline
		{\bf \#Items} & 4.2K & 1.1K & 20K   \\ \hline
	    {\bf \#Ratings} & 681K & 32.9K & N/A \\ \hline
		{\bf \#Triples with positive $\ap$} & {\bf 16.1M} & {\bf 14.9M} & {\bf 50M -- 250M} \\ \hline
	    {\bf \#Item classes} & 94 & 43 & 500 \\ \hline
	    {\bf Largest class size} & 1081  & 52 &  60 \\ \hline
		{\bf Smallest class size} & 2 & 10  & 24\\ \hline
		{\bf Median class size} & 12 & 27  &  40 \\ \hline
		
	\end{tabular}
\caption{Data Statistics (K: thousand)} \label{table:dataset}
\end{table}

We conduct extensive experiments on two real datasets --
	Amazon and Epinions (\url{http://www.epinions.com/}) --
	to evaluate \revmax algorithms.
Statistics of the datasets are in Table~\ref{table:dataset}. 
%We also run scalability test on synthetic datasets.

In the experiments, an important task is to pre-compute the primitive adoption probabilities for each applicable $(u,i,t)$ triple.
The intuition is that if the user is predicted to like the item a lot, i.e., if the predicted rating $\hat{r}_{ui}$ from an established classical RS is high, then $\ap(u,i,t)$ should be high.
Intuitively, $\ap(u,i,t)$ should be anti-monotone w.r.t.\ price\footnote{Our framework and algorithms do not assume this.}, and we use the notion of buyer valuation to instantiate this intuition.
Let $\val_{ui}$ be user $u$'s valuation on item $i$, which is the maximum amount of money $u$ is willing to pay for getting $i$.
For trust or privacy reasons users typically do not reveal their true valuations, thus we make the independent value (IPV) assumption which says that $\val_{ui}$ is drawn from a common probability distribution and is independent of others \cite{klbbook, agtbook}.
Thus, we use a simple definition to estimate $\ap(u,i,t)$ to be $\Pr[\val_{ui} \geq \price(i,t)] \cdot \hat{r}_{ui} / r_{\max}$, where $r_{\max}$ is the maximum rating allowed by the system.

For real data, the first step is to compute predicted ratings using a ``vanilla'' MF model (we used the stochastic gradient descent algorithm) \cite{mfsurvey}.
Then, for all users we select 100 items with the highest predicted ratings and compute primitive adoption probabilities (if the rating is too low, the item is deemed to be of little interest).
Naturally, only triples with nonzero adoption probability will be considered by (and given as input to) any \revmax algorithm, and thus \textsl{the number of such triples is the true input size} (as an analogy, recall that the number of known ratings is the true input size that affects the running time of RS algorithms such as matrix factorization).
We highlight this number in Table~\ref{table:dataset} in bold font.

\subsection{Data Preparations and Experiments Setup}\label{sec:dataset}

\spara{Amazon}
We selected 5000 popular items from the Electronics category and crawled their prices from August 31, 2013 to November 1, 2013 via Amazon's Product Advertising API\footnote{\scriptsize\url{https://affiliate-program.amazon.com/gp/advertising/api/detail/main.html}}.
The reason to select popular items is that they receive enough ratings for computing predicted ratings in a reliable manner.
The reason to focus on one category is that we want the buyers of different items to have a reasonable overlap.
The items include, e.g., Amazon Kindle and accessories, Microsoft Xbox 360 and popular Xbox games, etc.
%We focus on one particular category as we want to get asmany overlapping
%	adopters, leading to a denser rating matrix.
For all items, we record one price per day.
%Since the price is perfectly accurate ground-truth data, we can use
%	them \emph{as is} in experiments.
In addition, we gathered all historical ratings of these items and the users
	providing them.
Items with fewer than 10 ratings are filtered out.
%After filtering out items with less than 10 ratings, we are left with 4.2K
%	items, 23K users, and 681K ratings.
%The crawling contains product category information, yielding 
%There are 94 classes;
%the largest, smallest, and median class size is 1081, 2, and 12 respectively.
We then train a low-rank matrix model using the implementation in MyMediaLite~\cite{mymedialite} to obtain predicted ratings.
The model yields a RMSE of
	$0.91$ on five-fold cross validation, which is reasonably good
	by RS standards.
We also set $T=7$ to simulate a horizon of one week.
%Valuation distributions are estimated using the kernel density estimation
%	(KDE) method (details below).

\spara{Epinions}
When giving reviews, Epinions users can optionally
	report the price they paid (in US dollars).
This makes Epinions a valuable source for obtaining price data.
We extracted item information from a public Epinions dataset
	\cite{richardson02} and followed the provided URLs (of the
	product pages) to crawl all reviews and prices.
For accurate estimations of price and valuation distributions,
	items having fewer than 10 reported prices were filtered out.
%This leaves us with 1.1K items, 21.3K users, and 32.9K ratings.
%There are 43 item classes; the largest, smallest, and median class
%	size is 52, 10, and 27 respectively.
We also trained a matrix factorization model
	with a RMSE of $1.04$ on five-fold cross validation.
%Epinions is known to be ultra sparse and as a result has higher
%	RMSE~\cite{mymedialite}.
It has been noted before that Epinions is an ultra sparse dataset
	and hence has a higher RMSE \cite{mymedialite}.

\spara{Learning Price and Valuation Distributions}
The prices on Epinions cannot be mapped to a ground-truth time-series
	as users bought the item from many different sellers.
To circumvent this, we apply the \emph{kernel density estimation}
	(KDE) method~\cite{silverman86} to estimate price distributions.
Consider an arbitrary item $i$ and let
	$\langle p_1, p_2, \dotsc, p_{n_i} \rangle$ be the list of prices
	reported for $i$.
In KDE, it is assumed that the $n_i$ prices
	are i.i.d. with a probability distribution whose density
	function takes the form
%\begin{align}
$\hat{f}_i(x) = \frac{1}{n_i h}\sum_{j=1}^{n_i} \kappa\left(\frac{x-p_j}{h}\right)$,
%\end{align}
where $\kappa(\cdot)$ is the kernel function and $h>0$ is a
	parameter called \emph{bandwidth}, controlling the scale of
	smoothing.
We apply the Gaussian kernel \cite{jiang07} by setting
	$\kappa(x)$ to be the standard Gaussian density function
	$\phi(x) = \frac{1}{\sqrt{2\pi}} \mathrm{exp}(-\frac{x^2}{2})$.
The optimal bandwidth $h$ for the Gaussian kernel can be
	determined by Silverman's rule of thumb~\cite{silverman86}:
	$h^* = ( \frac{4\hat{\sigma}^5}{3 n_i} )^{\frac{1}{5}}$,
	where $\hat{\sigma}$ is the empirical standard deviation.
Then, from the estimated $f_i$, we
	generate $T=7$ samples and treat the samples as
	if they were the prices of $i$ in a week.
For each item $i$ we also use $f_i$ as
	a proxy for its valuation distribution.
Note that the distribution $f_i$ remains Gaussian with mean
	$\mu_i = \sum_{j=1}^{n_i} p_j / (n_i h)$ and
	variance $\sigma_i^2 = h$.
Thus for any price value $p$, $\Pr[\val_{ui} \ge p] = 1 - \hat{F}_i(p) = \frac{1}{2} (1-\mathrm{erf}(\frac{p-\mu_i}{\sqrt{2}\sigma_i}))$, where $\mathrm{erf}(\cdot)$ is the Gauss error function.

\begin{figure*}[t!]
\begin{tabular}{cccc}
    %\hspace{-6mm}
    \includegraphics[width=.24\textwidth]{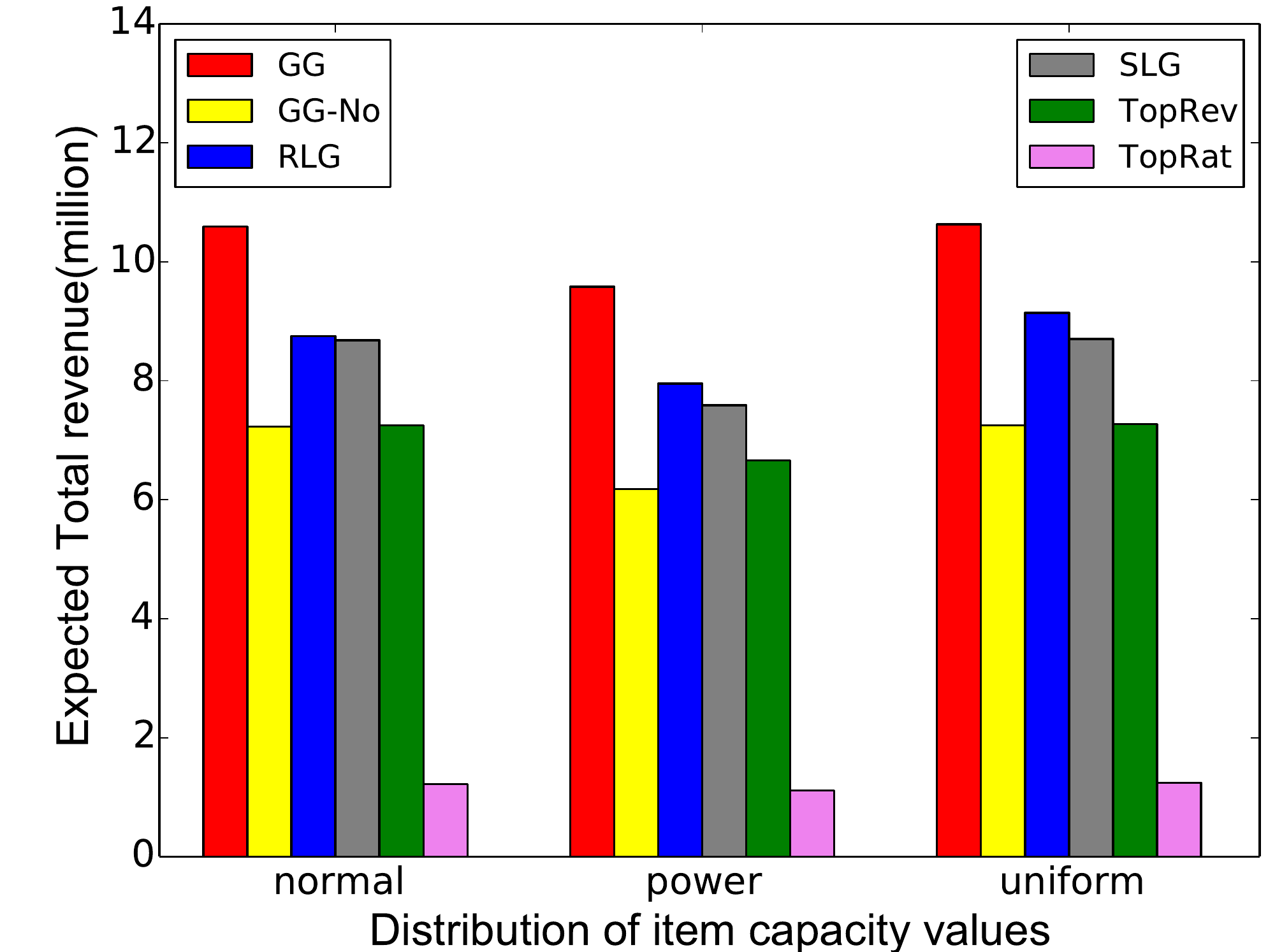}&
    \hspace{-2mm}\includegraphics[width=.24\textwidth]{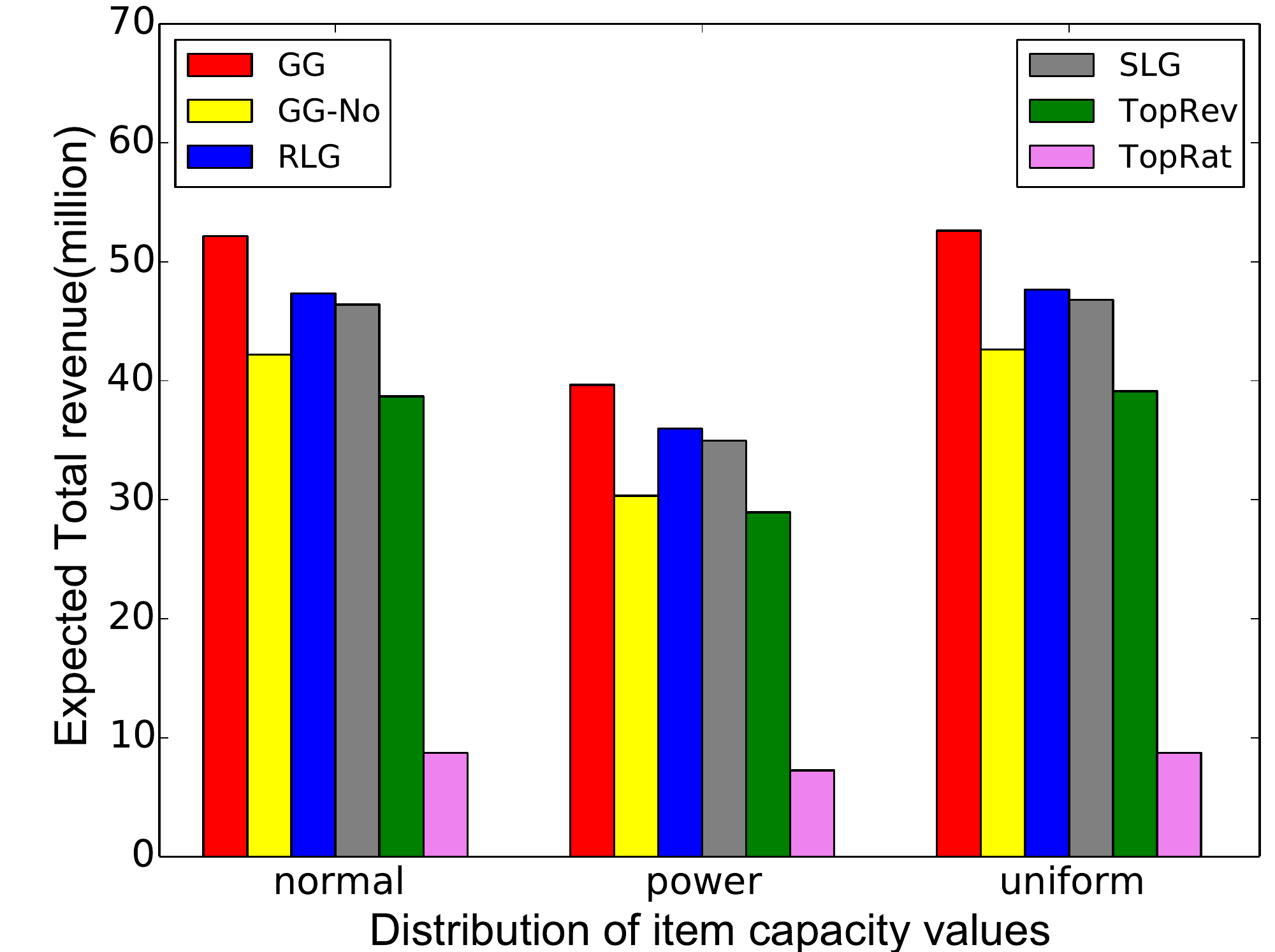}&
    \hspace{-2mm}\includegraphics[width=.24\textwidth]{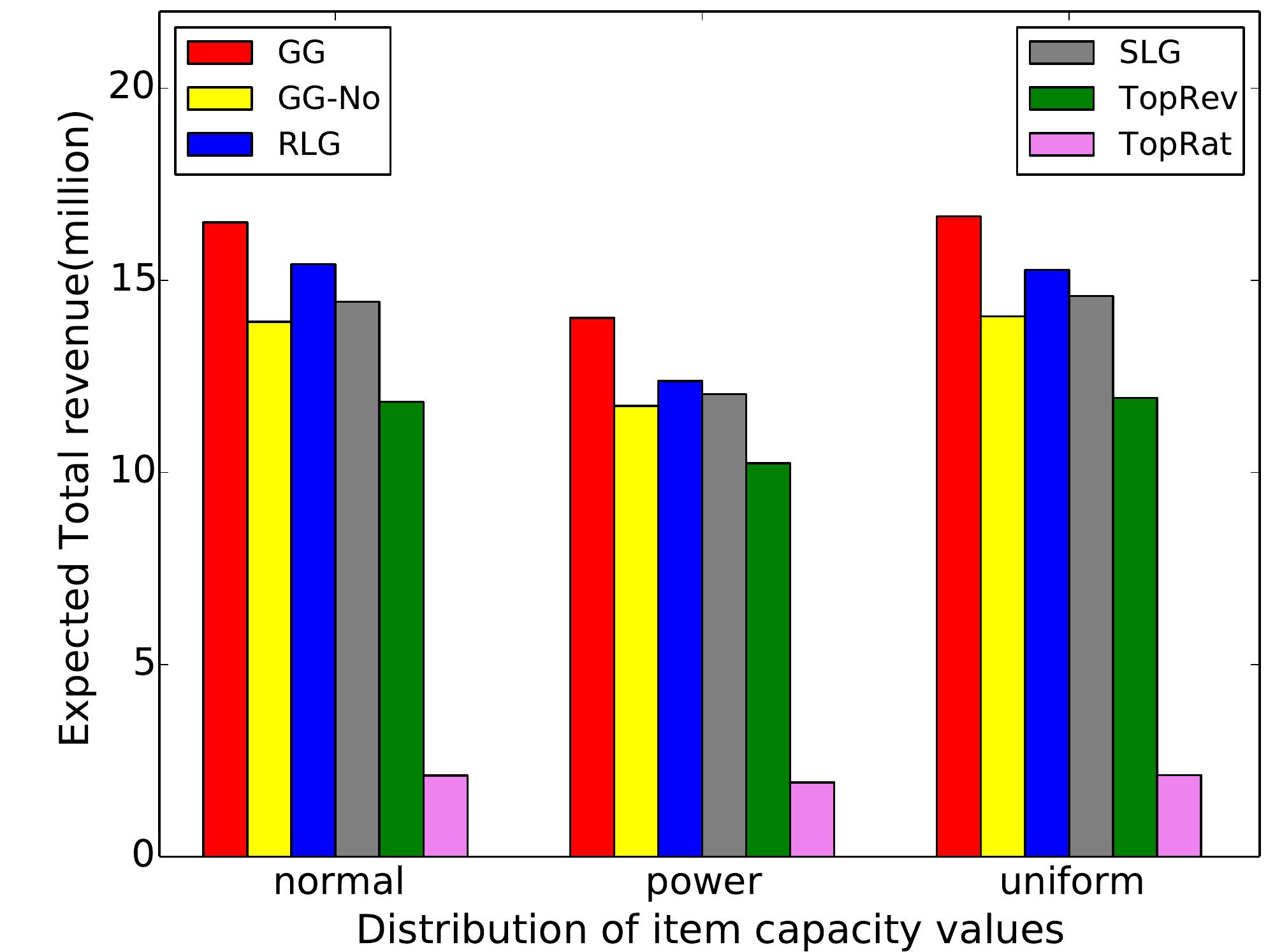}&
    \hspace{-2mm}\includegraphics[width=.24\textwidth]{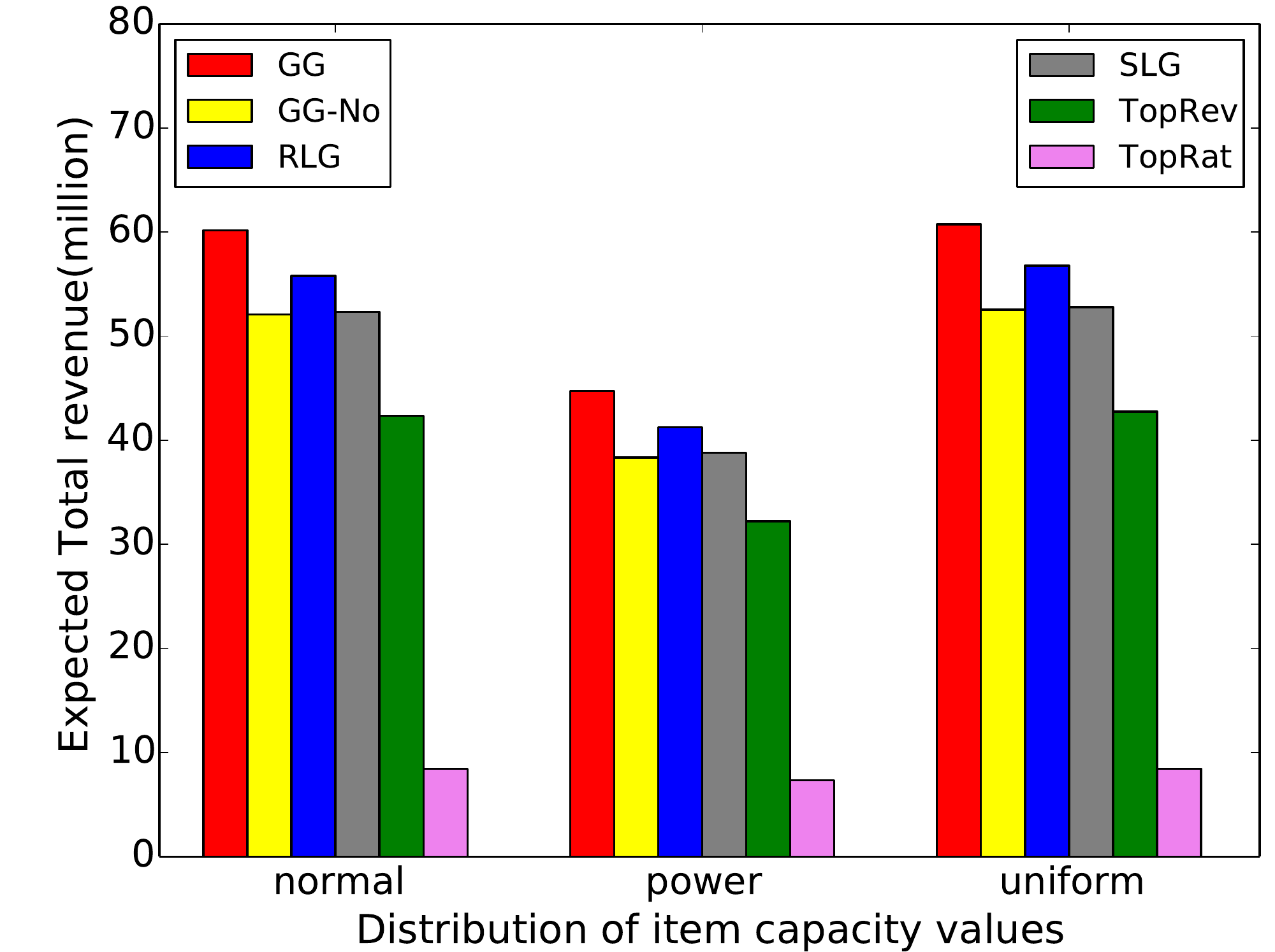}\\
	(a) Amazon & (b) Epinions  & (c) Amazon, class size $1$ & (d) Epinions, class size $1$  \\
\end{tabular}
\caption{ Expected total revenue achieved, with each $\beta_i$ chosen uniformly at random from $[0,1]$.}
\label{fig:uniformBeta}
\end{figure*}

\begin{figure*}[t!]
\begin{tabular}{cccc}
    %\hspace{-6mm}
    \includegraphics[width=.24\textwidth]{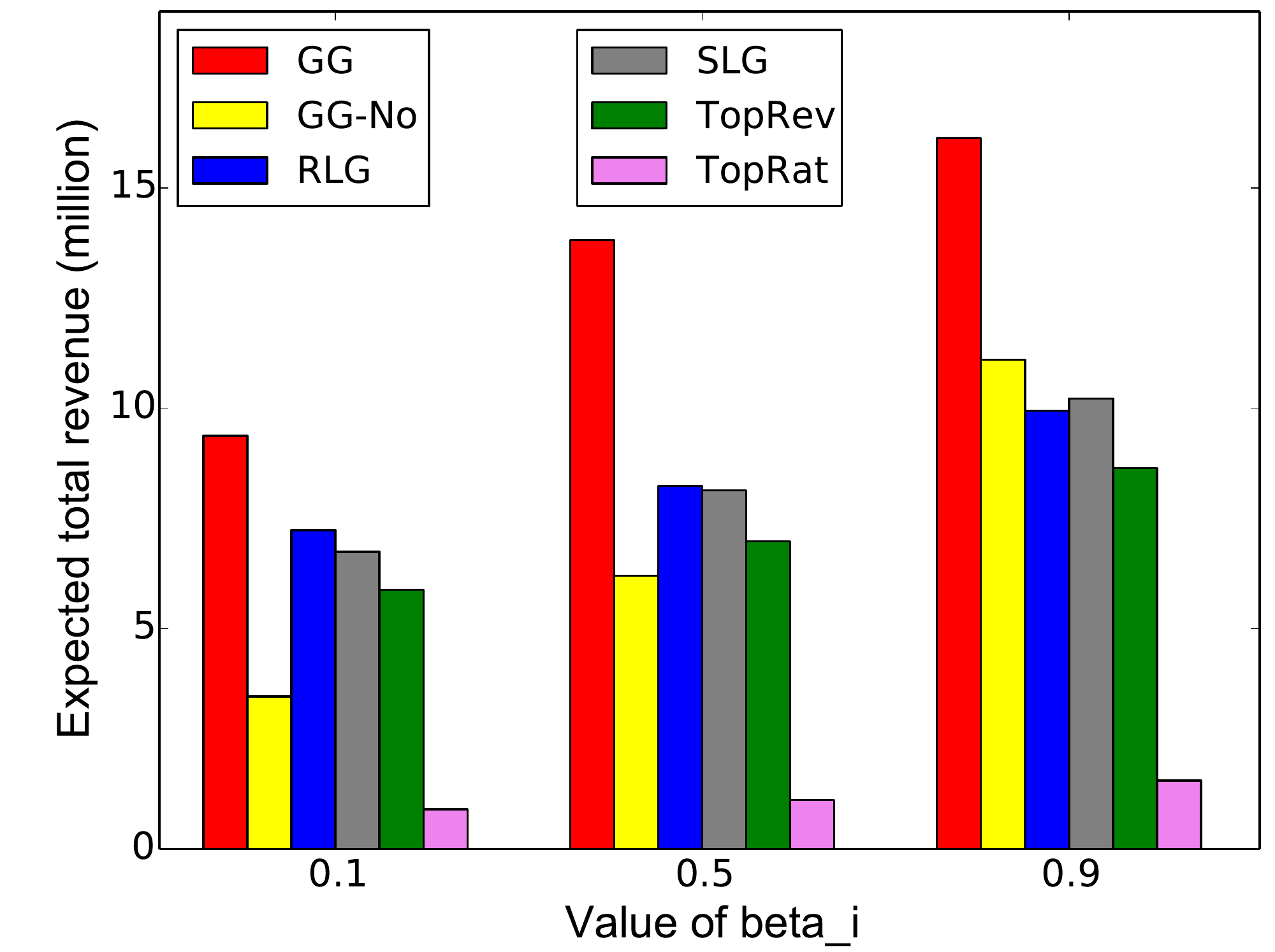}&
    \hspace{-2mm}\includegraphics[width=.24\textwidth]{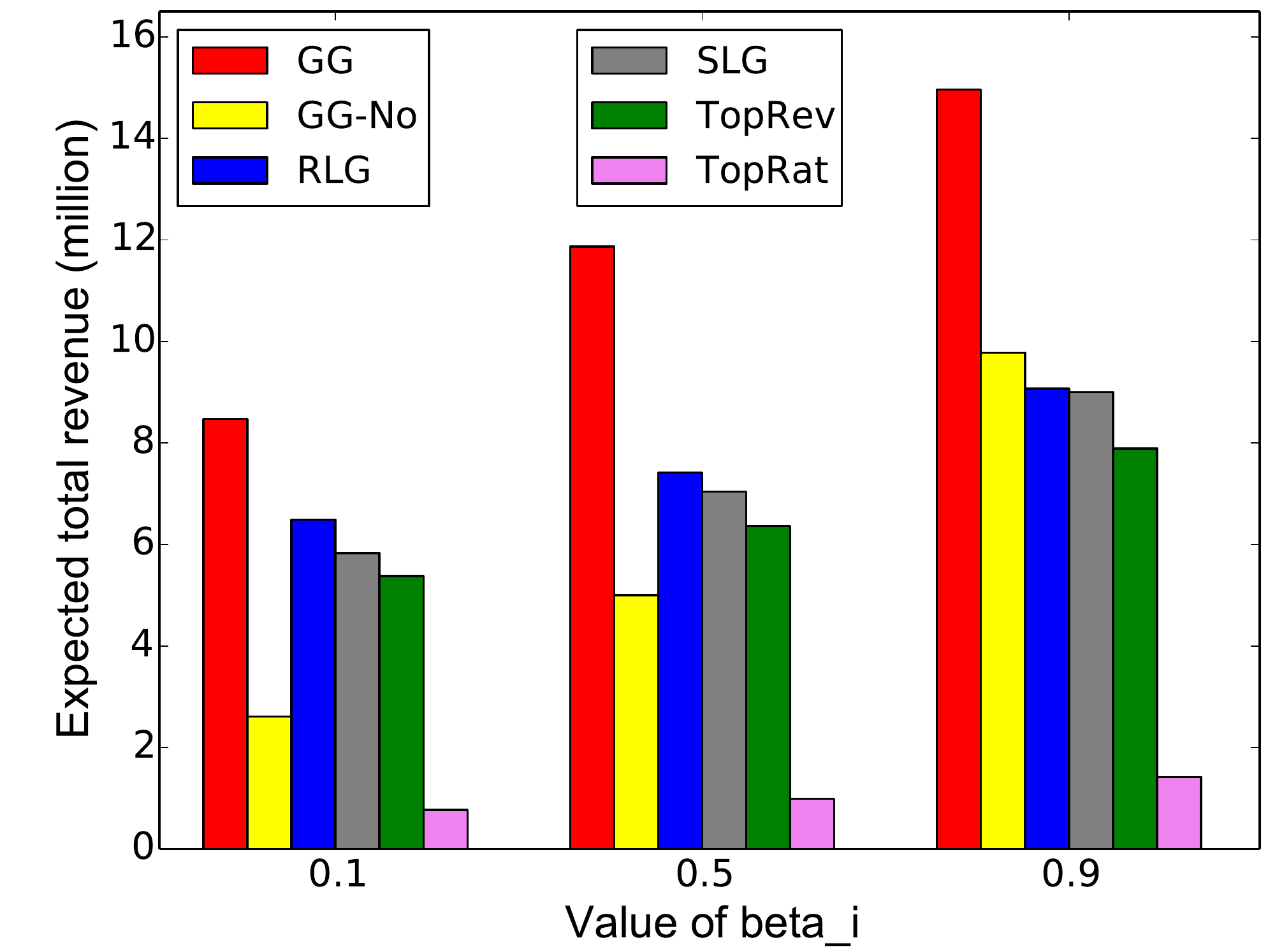}&
    \hspace{-2mm}\includegraphics[width=.24\textwidth]{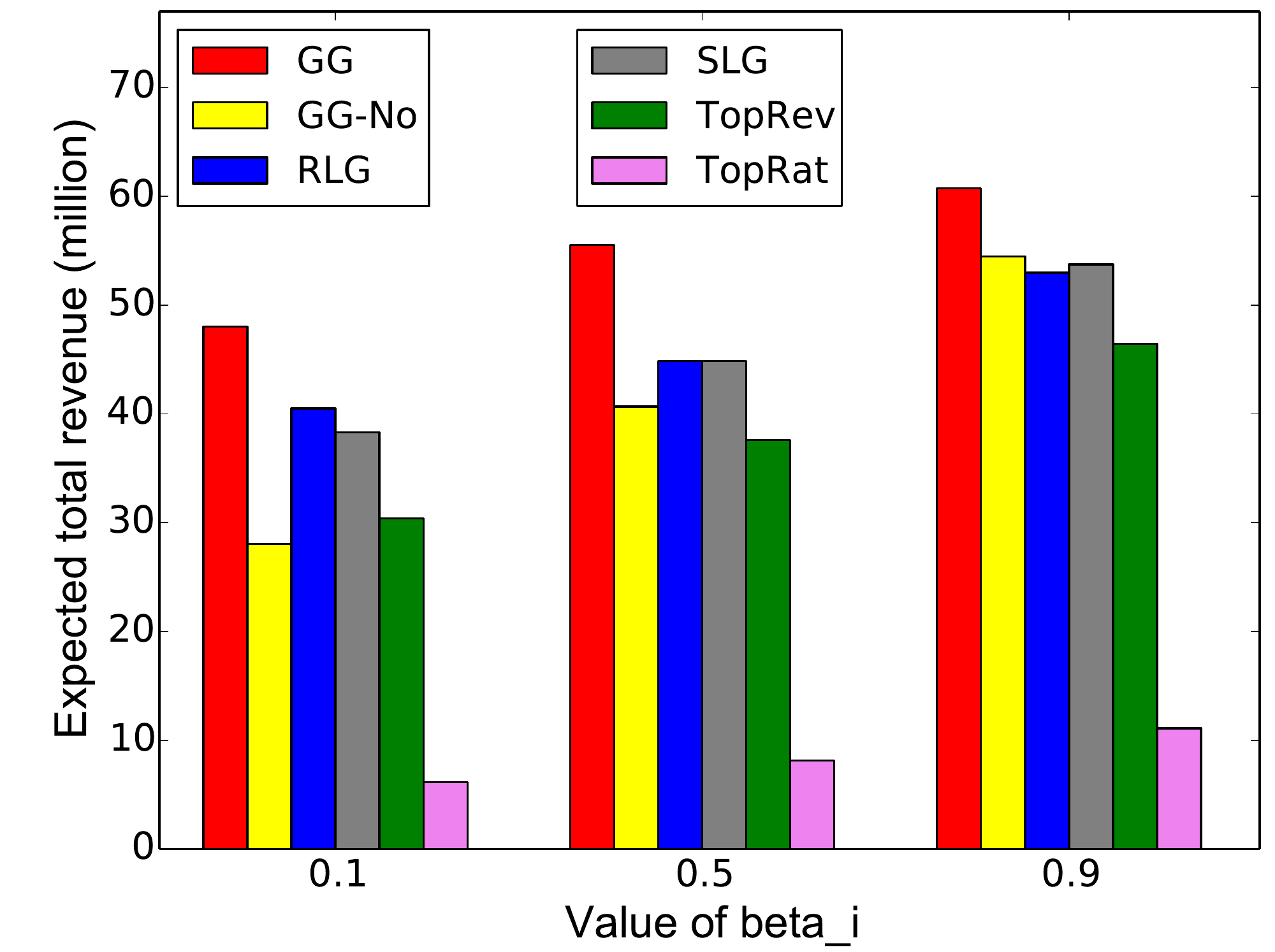}&
    \hspace{-2mm}\includegraphics[width=.24\textwidth]{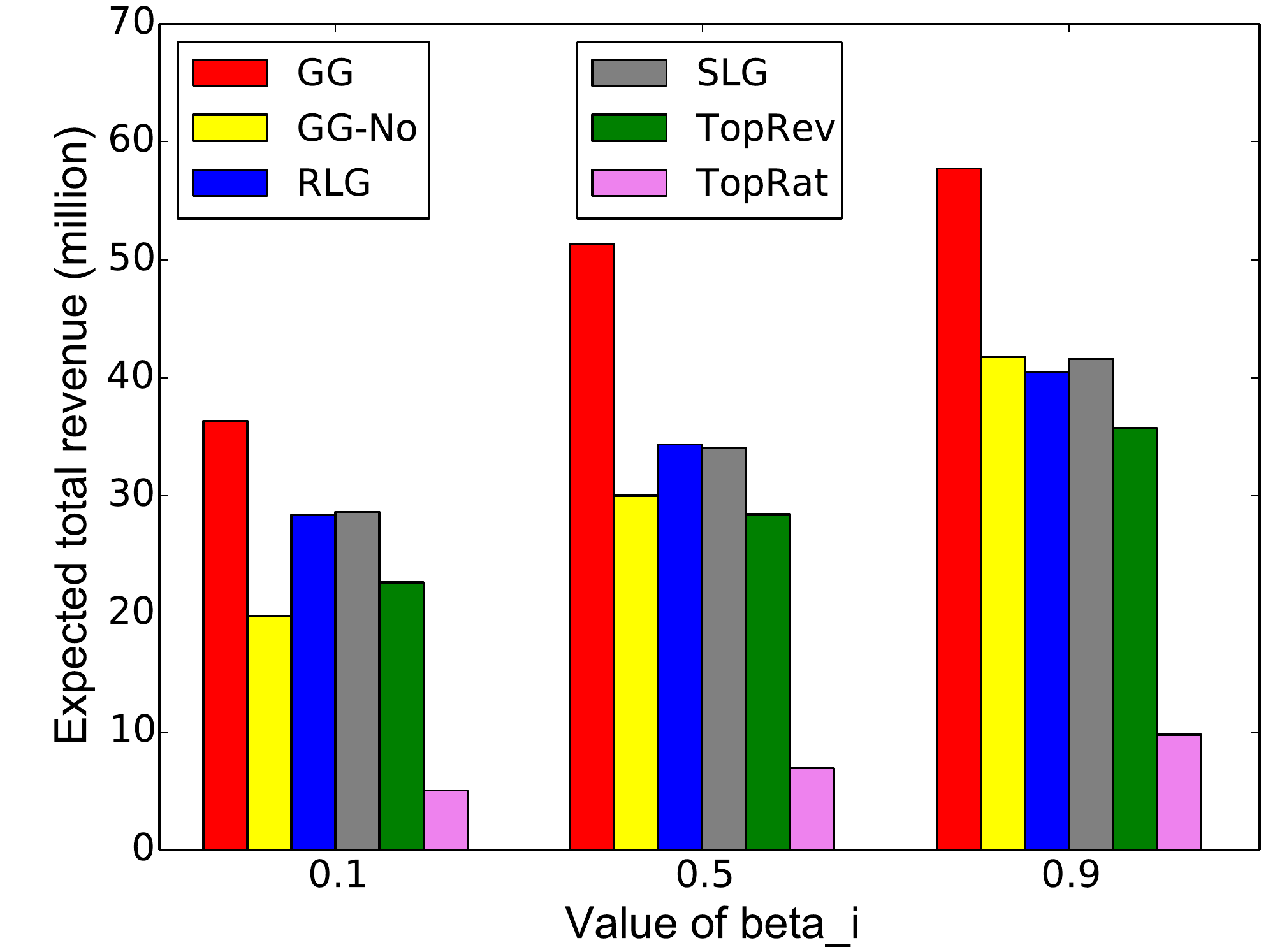} \\
    { (a) Amazon (Gaussian)} & { (b) Amazon (Exponential)} & { (c) Epinions (Gaussian)}  &  { (d) Epinions (Exponential) }  \\
\end{tabular}
\caption{Expected total revenue with varying saturation strength, item class size $> 1$.}
\label{fig:BetaClass}
\end{figure*}

\begin{figure*}[t!]
\begin{tabular}{cccc}
    %\hspace{-6mm}
    \includegraphics[width=.24\textwidth]{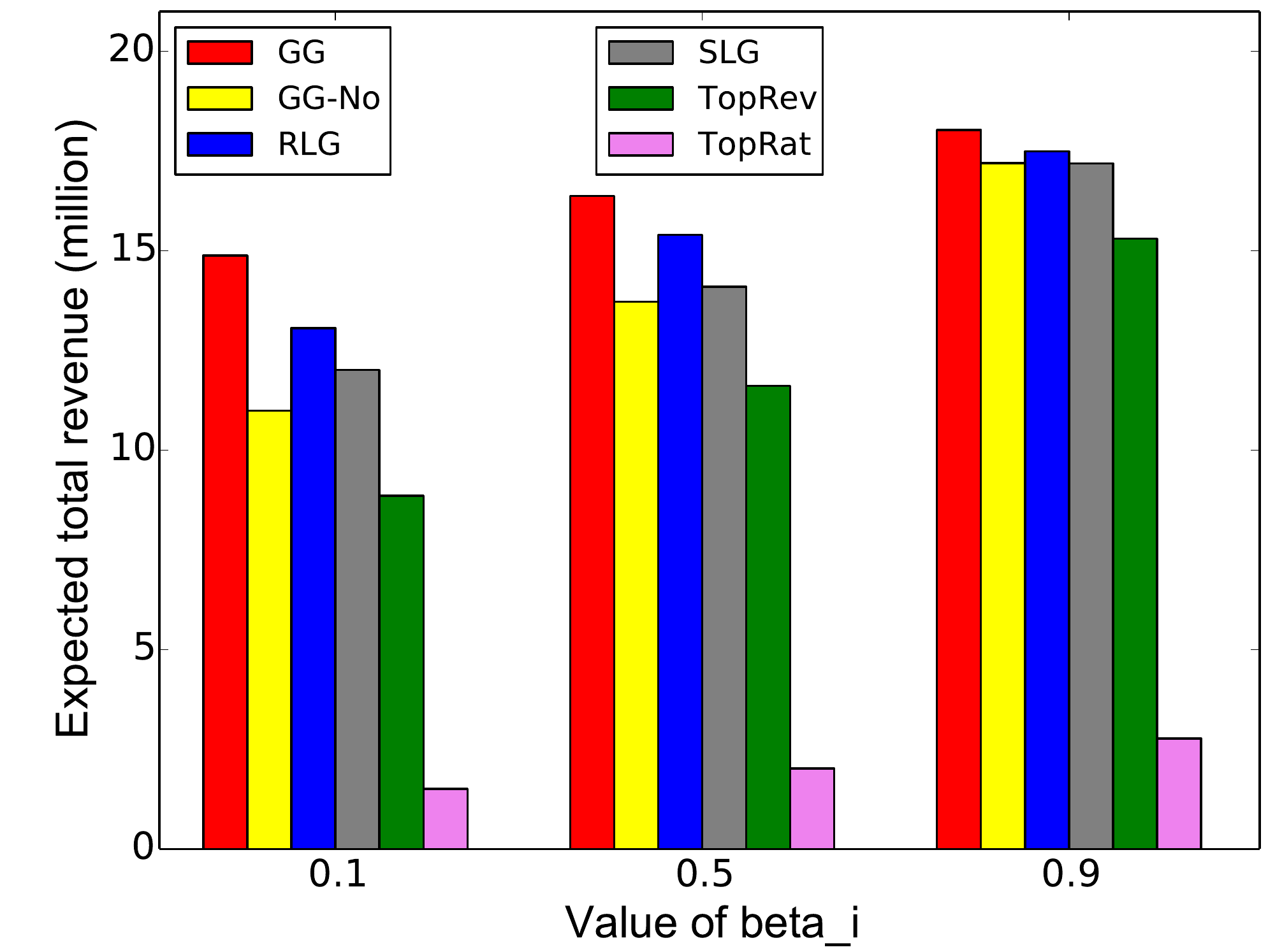}&
    \hspace{-2mm}\includegraphics[width=.24\textwidth]{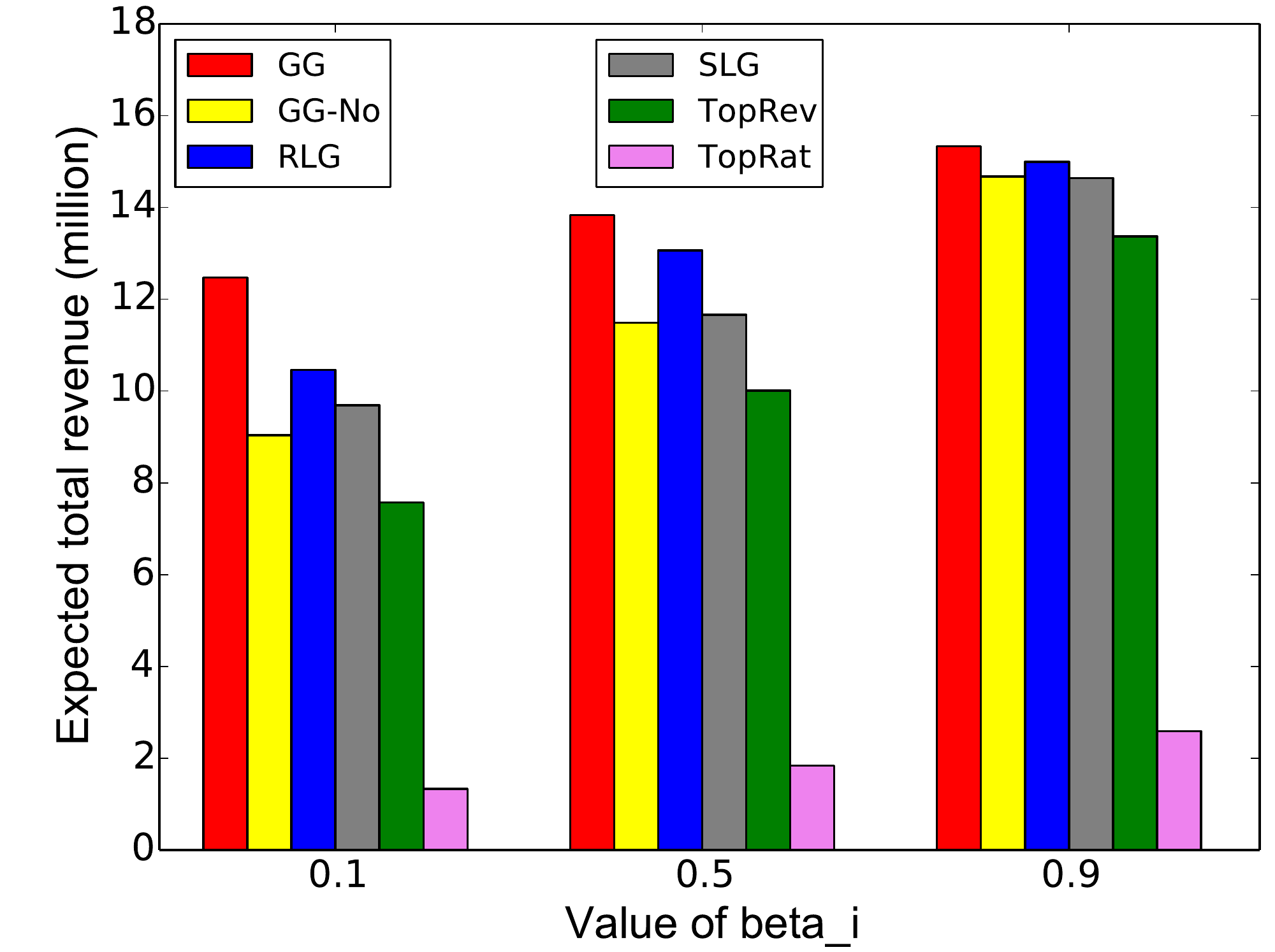}&
    \hspace{-2mm}\includegraphics[width=.24\textwidth]{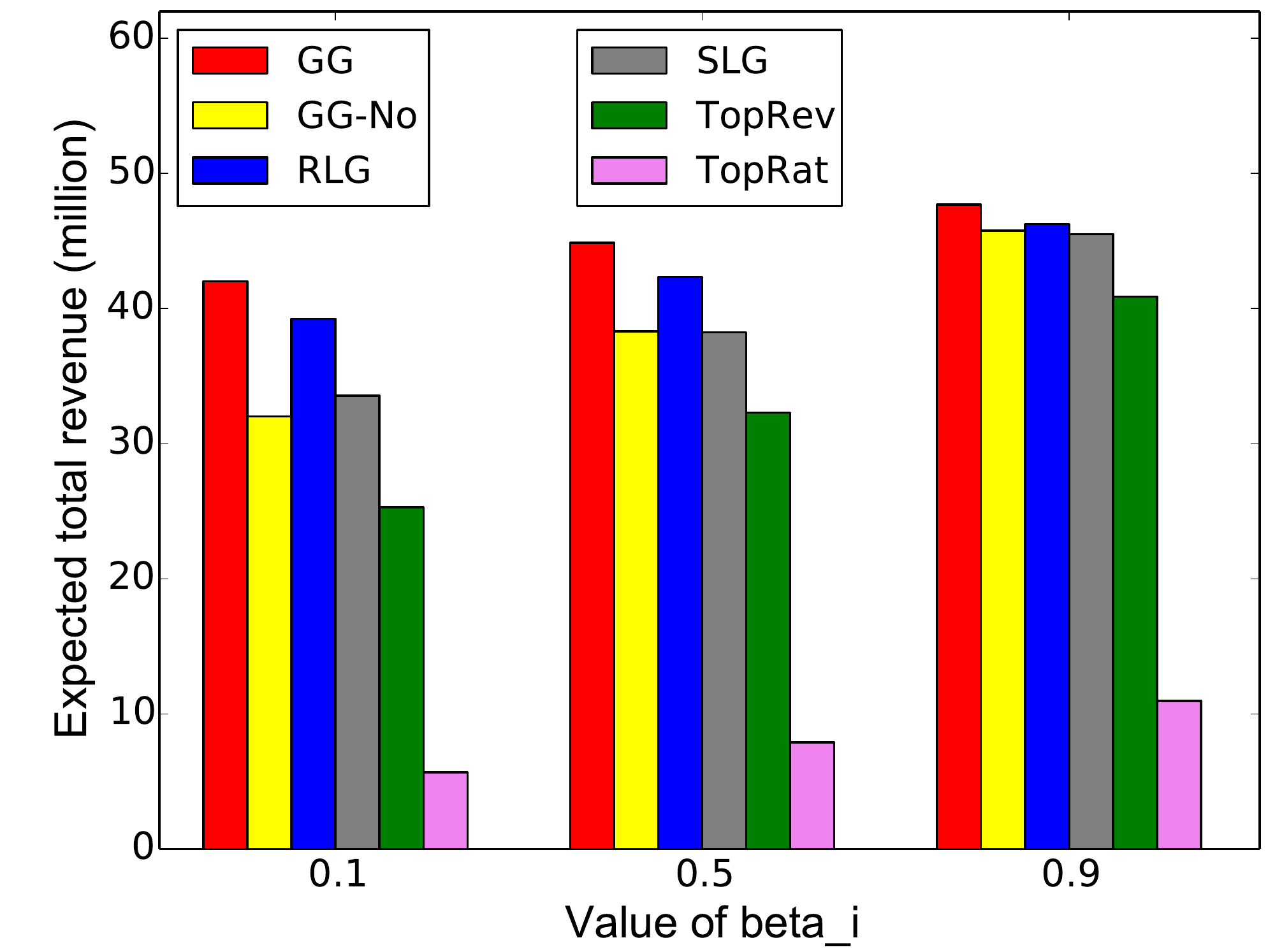}&
    \hspace{-2mm}\includegraphics[width=.24\textwidth]{plots_pdf/rev_epinions_power_noclass} \\
    { (a) Amazon (Gaussian)} & { (b) Amazon (Exponential)} & { (c) Epinions (Gaussian)}  &  { (d) Epinions (Exponential) }  \\
\end{tabular}
\caption{Expected total revenue with varying saturation strength, item class size $= 1$.}
\label{fig:BetaNoClass}
\end{figure*}

\spara{Synthetic Data}
We use synthetic datasets much larger than Amazon and Epinions to gauge the scalability of algorithms.
We do \emph{not} report revenue achieved on this data, since the generation process is artificial and \textsl{this dataset is solely used for testing scalability}.
In total there are five datasets with $|U| = 100K, 200K, \ldots, 500K$.
$T$ is set to $5$.
The item-set is the same: $|I| = 20K$ and for each item $i$, choose a value $x_i$ uniformly at random from $[10,500]$, and for each time $t$, sample $\price(i,t)$ uniformly at random from $[x_i, 2 x_i]$.
For each user $u$, we randomly choose 100 items to be the highest rated, and for each such item, sample $T$ (primitive) adoption probability values from a Gaussian distribution with $\mu = y_i$ and $\sigma^2 = 0.1$, where $y_i$ itself is chosen randomly from $[0,1]$; then we match adoption probabilities with prices so that anti-monotonicity holds.
Input size will then be $100T|U|$ ({\em cf.} Table~\ref{table:dataset}). Recall, the number of nonzero adoption probabilities is the critical factor in deciding the scalability of any \revmax algorithm.  

%\subsection{Experiments Setup}

\spara{Parameter Settings for Amazon and Epinions}
%Computing primitive probability $\ap(u,i,t)$ requires
%	user threshold $\theta_u$ (assuming a probabilistic
%	MF based RS is used), which we conservatively set as the $100$-th highest
%	predicted rating of user $u$, since it is unlikely a user would be interested items ranked lower than 100. 
It is unrealistic and unlikely that a user would be interested in buying all items.
That is, for each user $u$, $\ap(u,i,t)$ will be nonzero for a small subset of items:
we rank the items based on their predicted ratings for $u$ and compute adoption
probabilities for the top-100 items.
We also need the saturation factor $\beta_i$ for computing marginal revenue.
We test the following two cases.
First, we hard-wire a uniform value for all items, testing
	three different cases: $0.1$, $0.5$, and $0.9$, representing
	strong, medium, and weak effect of saturation.
Second, for each item $i\in I$, its $\beta_i$ is chosen uniformly at
	random from $[0,1]$ -- this effectively
	averages over possible values for all the items to model
	our lack of knowledge.
For capacity constraints $q_i$, we consider two probability
	distributions from which $q_i$ is sampled:
(1) Gaussian: $\cl{N}(5000, 200)$ for Epinions and $\cl{N}(5000, 300)$ for Amazon;
(2) exponential with inverse scale $2\times 10^{-3}$ (mean $5000$).
%(3) {\em Uniform}: $\cl{U}(5000, 5200)$ for Epinions and $\cl{U}(5000, 5300)$ for Amazon.
We believe the above scenarios are representative and worth testing.

\spara{Algorithms Evaluated}
We  compare \GGreedy, \SGreedy, and \RGreedy 
	with the following natural baselines.
\topkrating (for Top RAting) recommends to every user the $k$ items with
	highest predicted rating by MF;
\topkrev (for Top REvenue) recommends to every user the $k$ items with
	highest ``expected revenue'' (price $\times$ primitive adoption probability).
Since \topkrating is inherently ``static'',
	to evaluate the expected revenue it yields over $[T]$,
	the recommended items are repeated in all $T$ time steps.
We also consider a ``degenerated'' version of \GGreedy,
	which we call \GlobalNo: it ignores saturation effects when selecting triples.
That is, when computing marginal revenue and selecting triples, \GlobalNo will
	behave as though $\beta_i = 1$ for all $i\in I$,  but when we compute the final revenue
	yielded by its output strategy, the true $\beta_i$ values will be used.
This is to measure how much revenue would be lost if 
	saturation effects are present but ignored.

\begin{figure*}[t!]
\begin{tabular}{cccc}   
    \includegraphics[width=.245\textwidth]{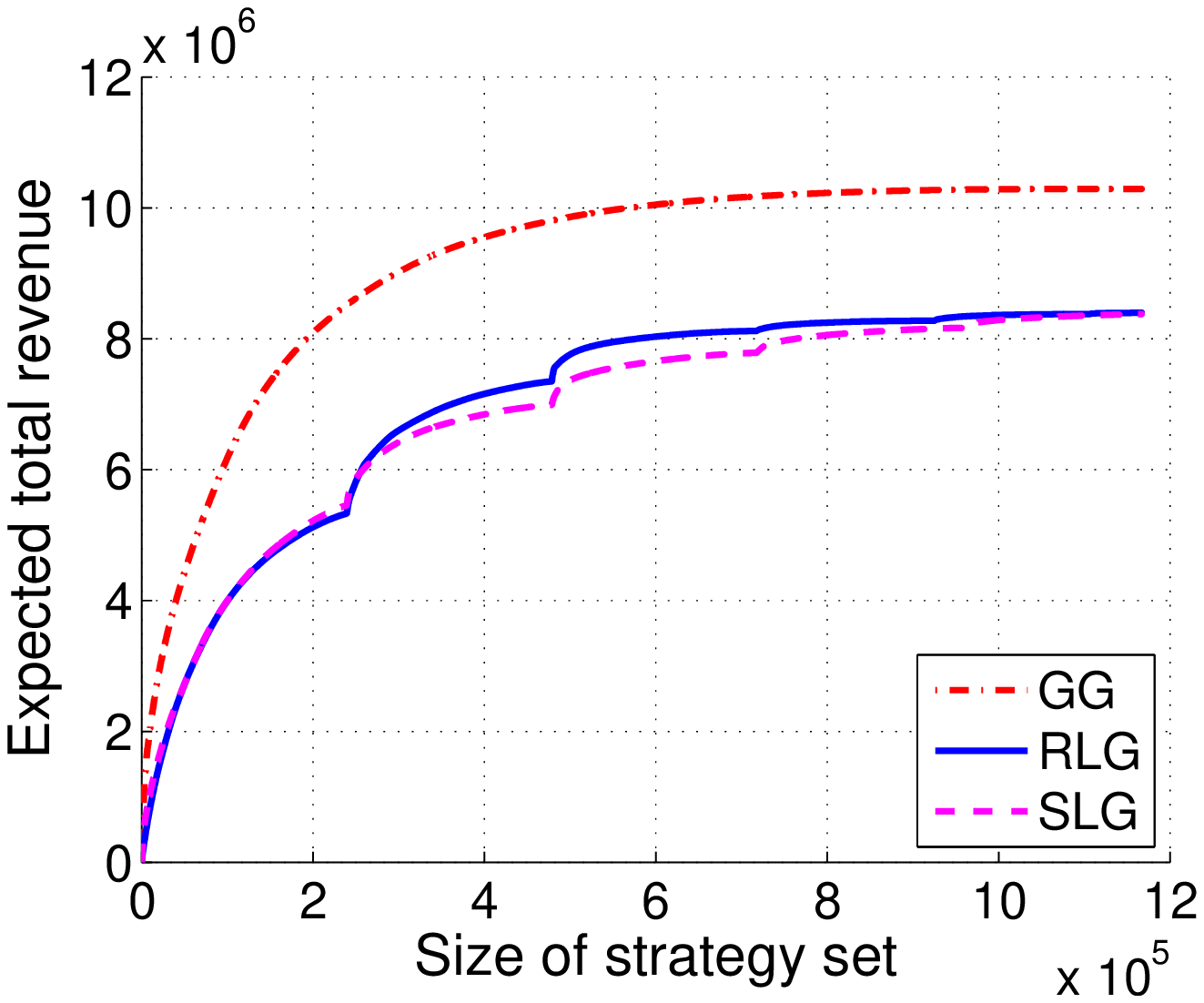}&
    \hspace{-2mm} \includegraphics[width=.245\textwidth]{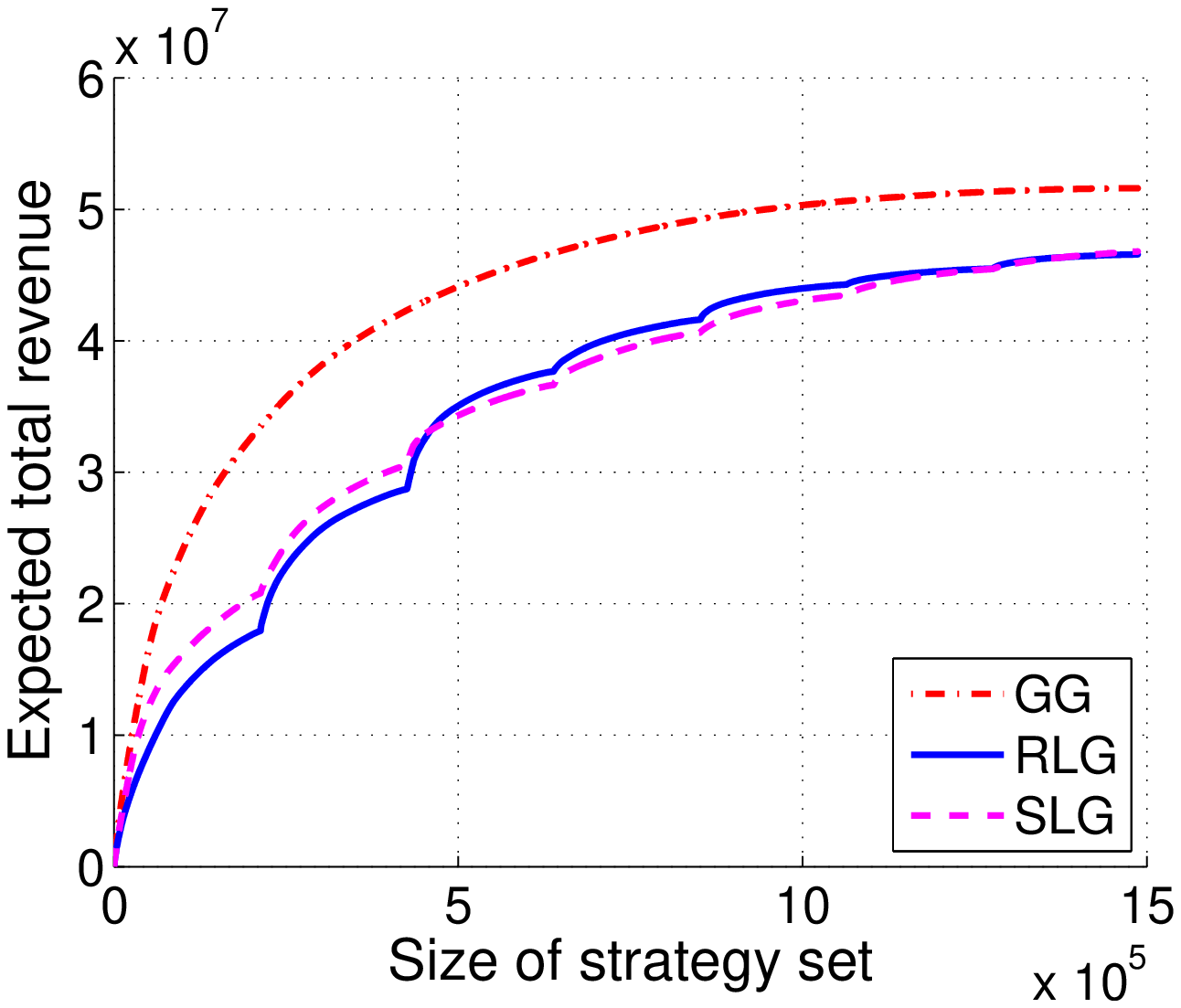}&
    \hspace{-2mm} \includegraphics[width=.245\textwidth]{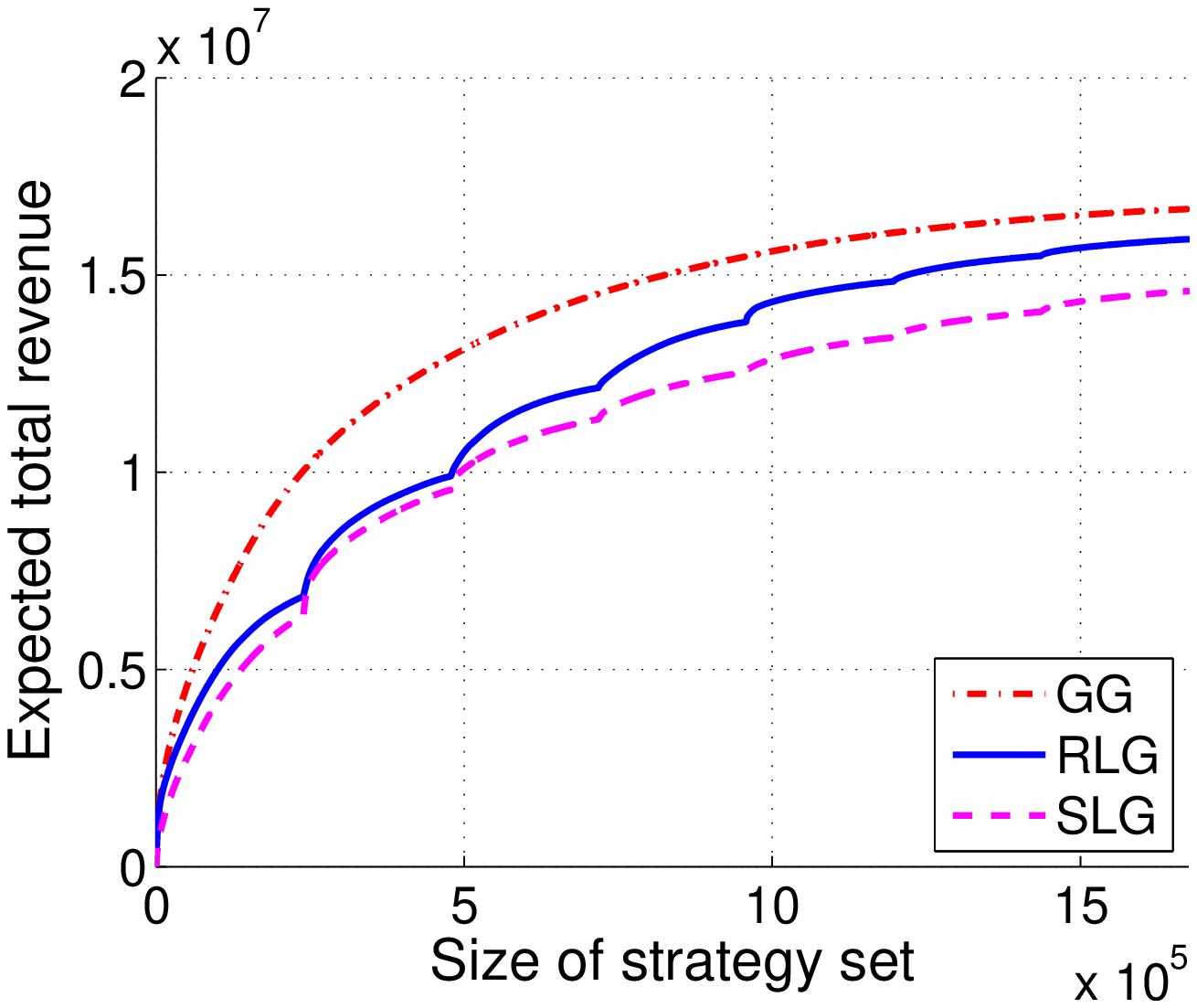}&
    \hspace{-2mm} \includegraphics[width=.245\textwidth]{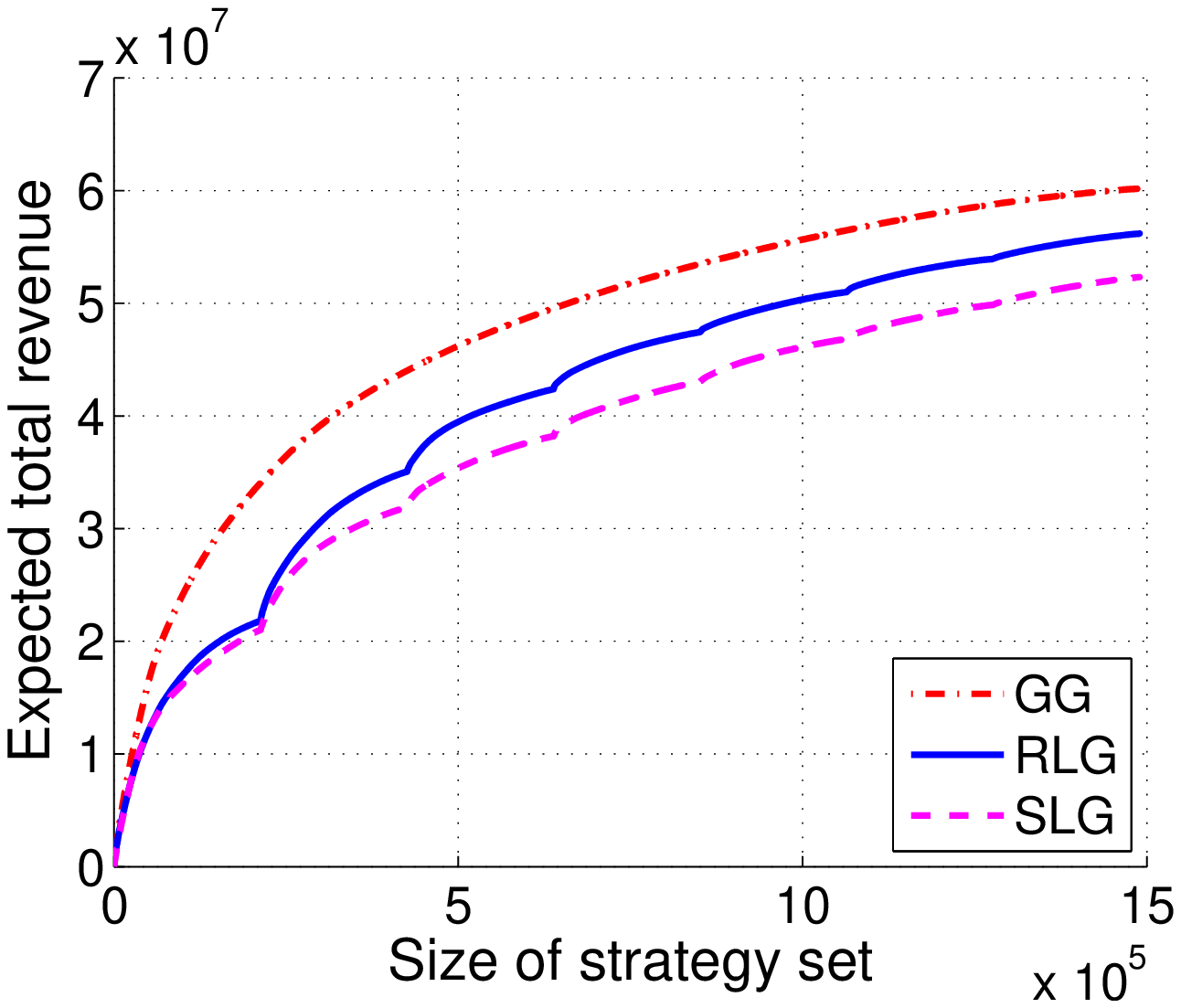}\\
	(a) Amazon & (b) Epinions & (c) Amazon, class size $1$ & (d) Epinions, class size $1$  \\
\end{tabular}
\caption{Expected total revenue of $\GGreedy$, $\SGreedy$, $\RGreedy$ vs.\ solution size ($|S|$).}
\label{fig:submod}
\end{figure*}

\begin{figure*}[t!]
\begin{tabular}{cccccc}   
    \includegraphics[width=.18\textwidth]{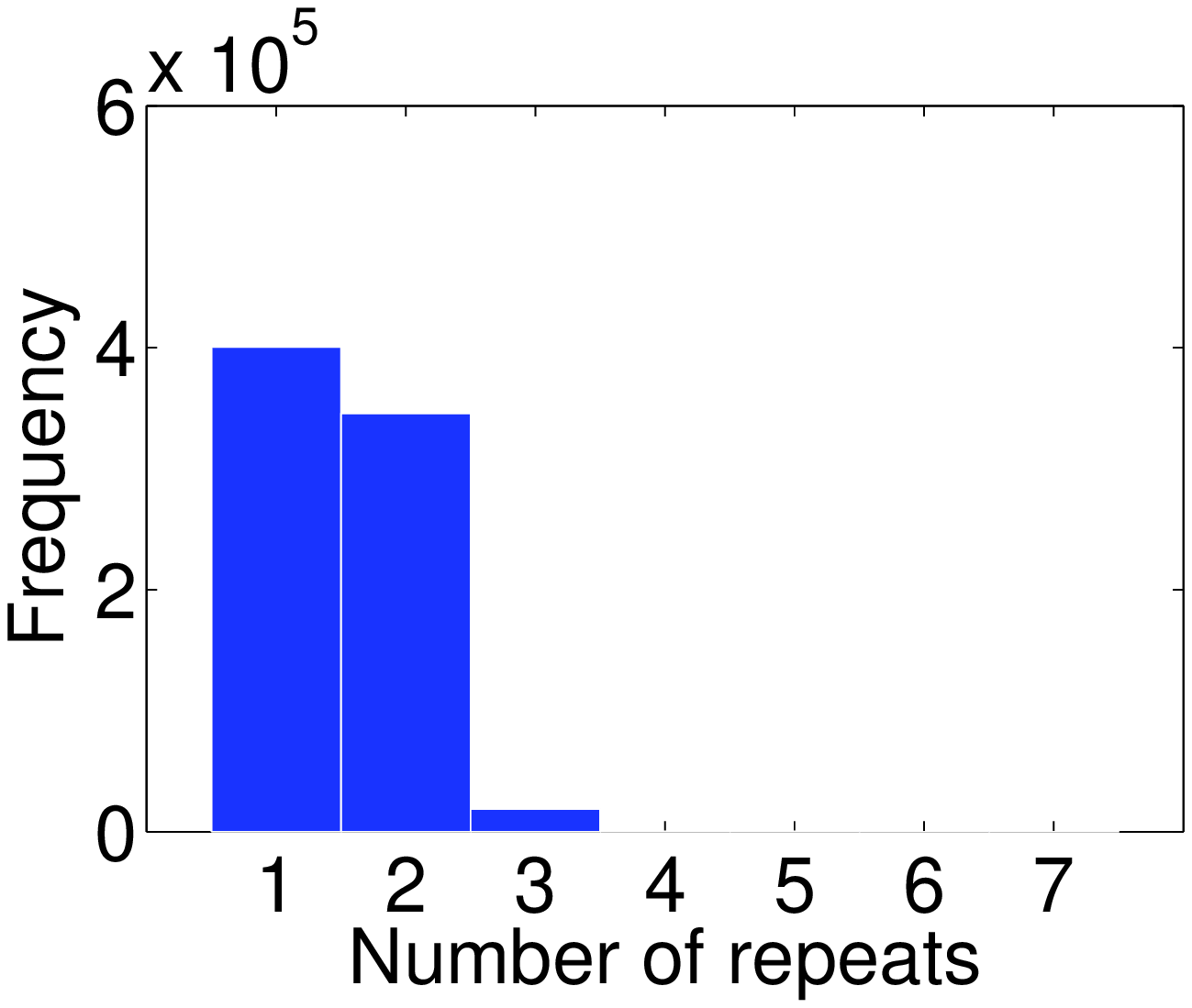}&
    \hspace{-6mm}\includegraphics[width=.18\textwidth]{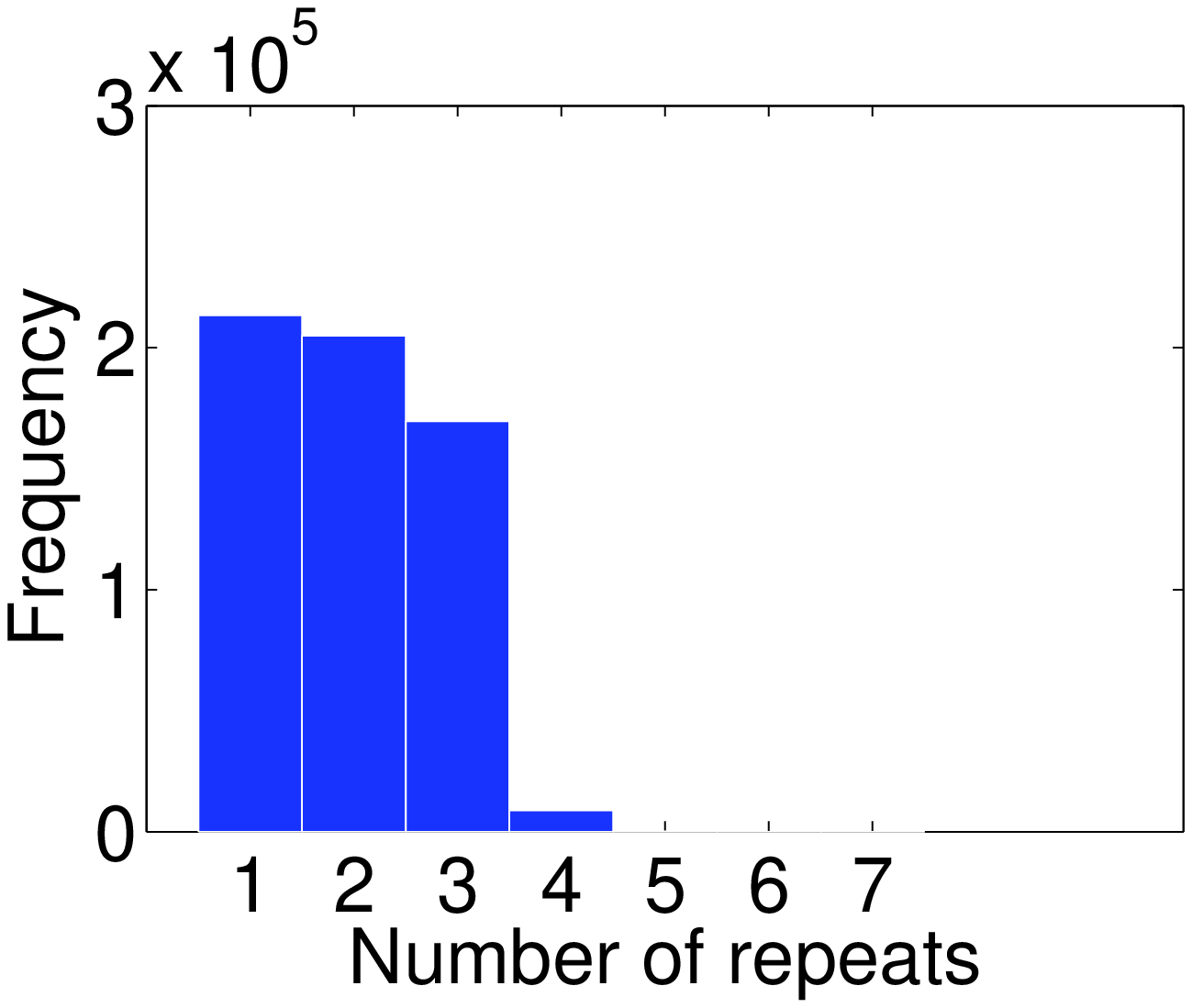}&
    \hspace{-6mm}\includegraphics[width=.18\textwidth]{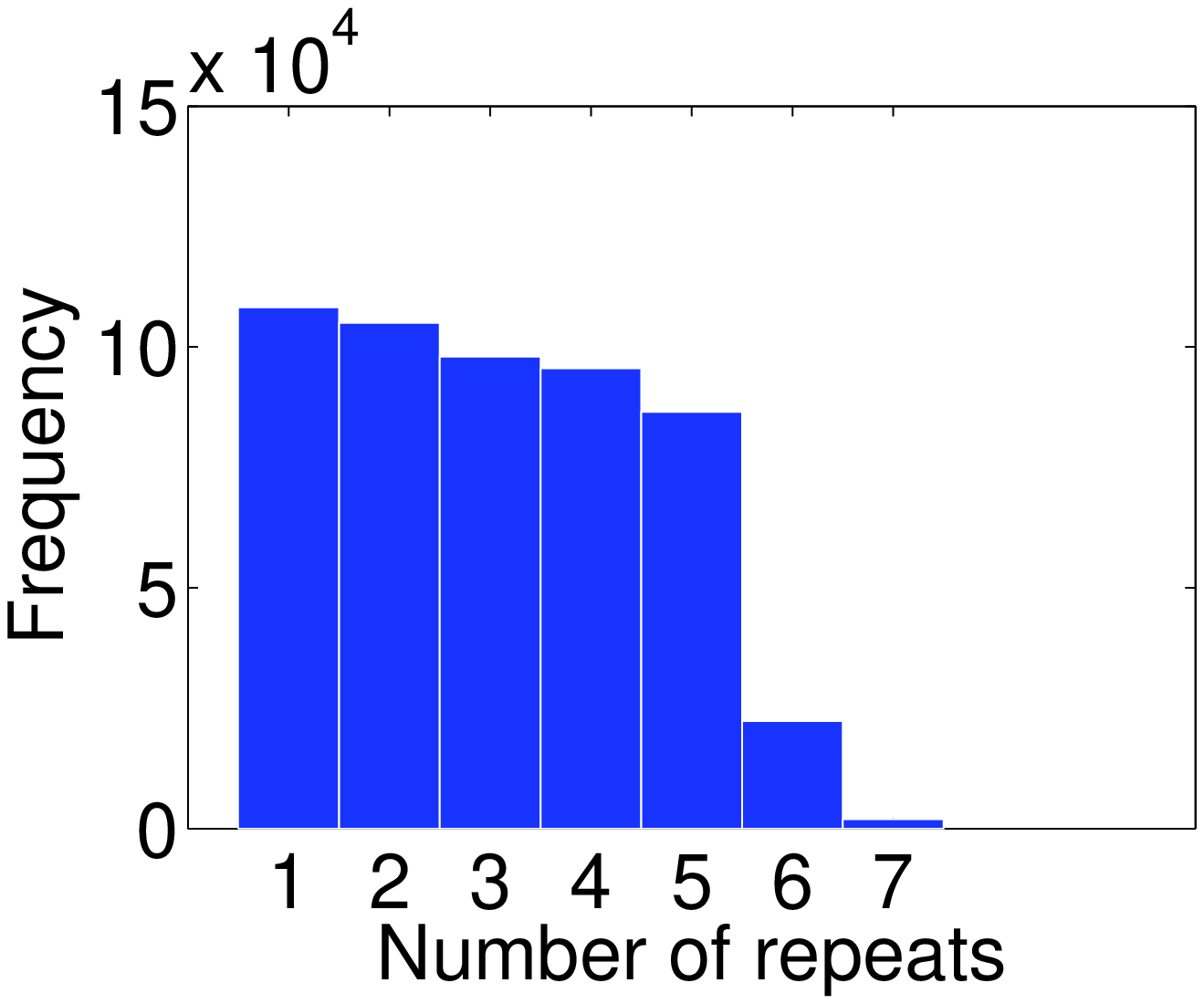}&
    \hspace{-6mm}\includegraphics[width=.18\textwidth]{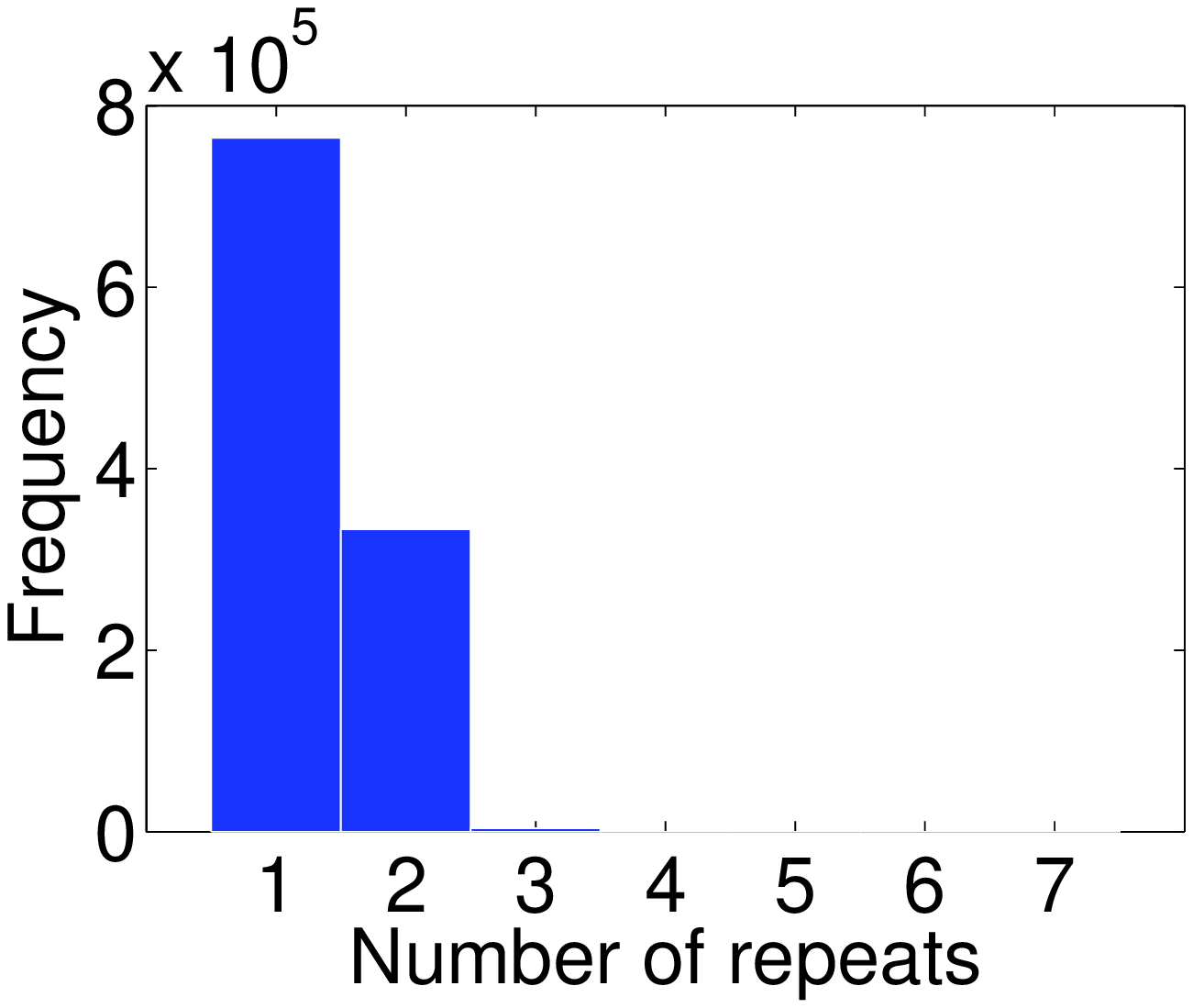}&
    \hspace{-6mm}\includegraphics[width=.18\textwidth]{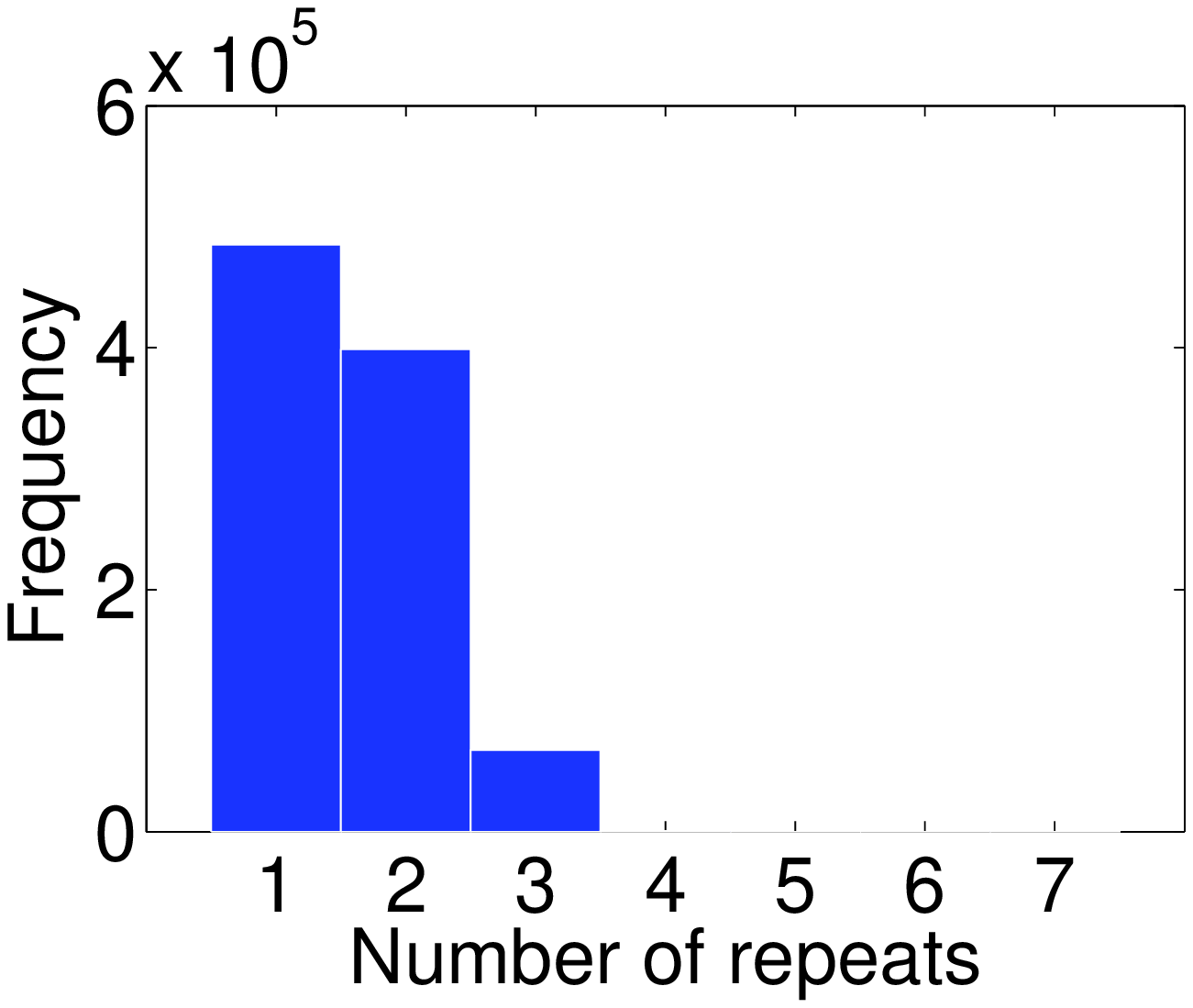}&
    \hspace{-6mm}\includegraphics[width=.18\textwidth]{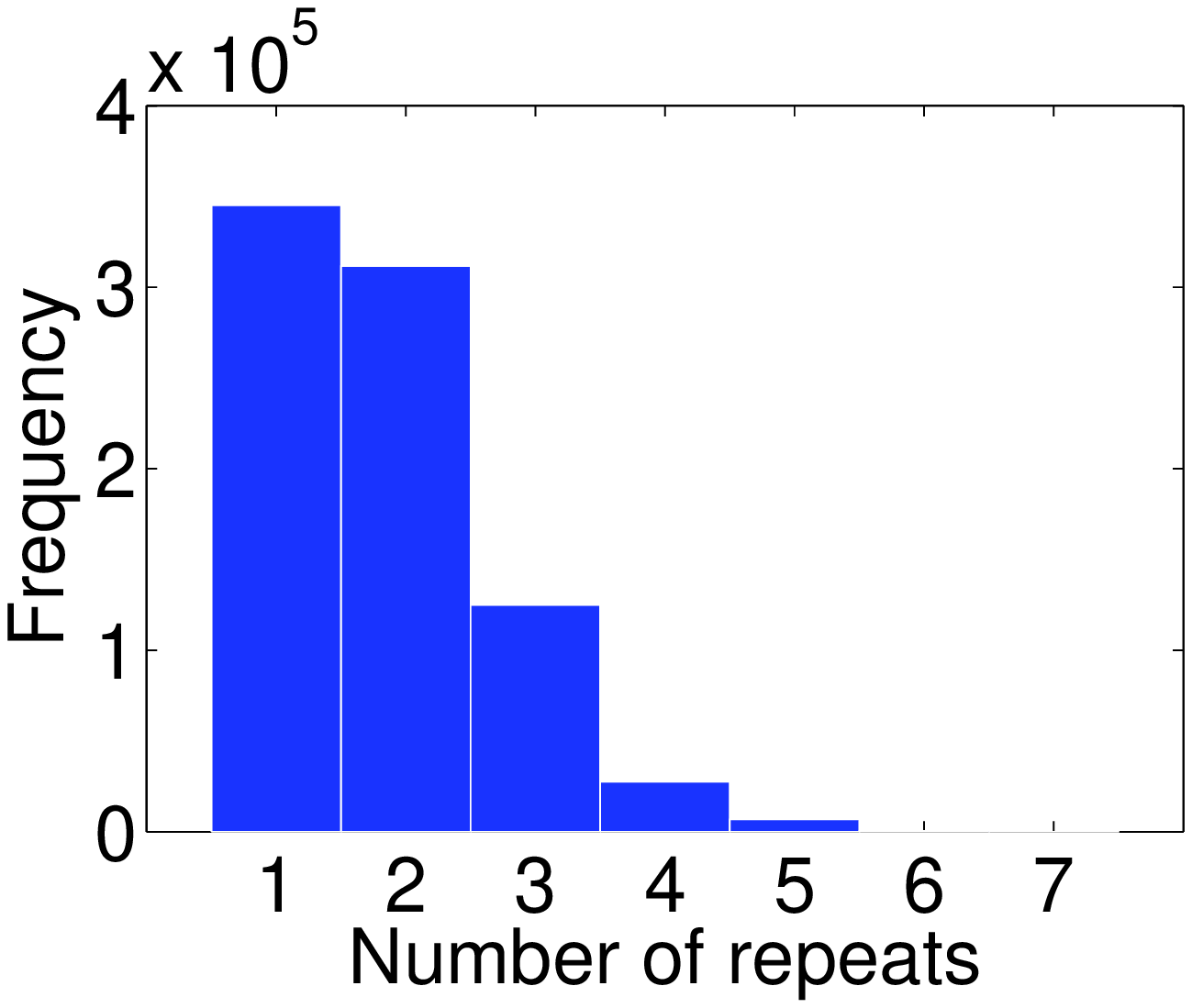}\\
    (a) {\scriptsize Amazon ($\beta_i = 0.1$)} & (b) {\scriptsize Amazon ($\beta_i = 0.5$)} & (c) {\scriptsize Amazon ($\beta_i = 0.9$)} & (d) {\scriptsize Epinions ($\beta_i = 0.1$)} & (e) {\scriptsize Epinions ($\beta_i = 0.5$)} & (f) {\scriptsize Epinions ($\beta_i = 0.9$) } \\
\end{tabular}
\caption{Histogram on the number of repeated recommendations made by $\GGreedy$ for each user-item pair.}
\label{fig:histClass}
\vspace{-3mm}
\end{figure*}

\subsection{Results and Analysis}\label{sec:result}

Experimental results are reported on two metrics: expected revenue achieved (quality of recommendations, {\em cf.} Def. \ref{def:rev}) and number of repeated recommendations made.
In addition to using ground-truth information on item classes from the data, we also test the scenario with class size $= 1$, i.e., each item is in its own category. %: $\cl{C}(i) \neq \cl{C}(j)$ whenever $i \neq j$.
For \RGreedy, we generate $N=20$ permutations of $[T]$.
All implementations are in Java and programs were run on a Linux server with Intel Xeon CPU X5570 (2.93GHz) and 256GB RAM.

\spara{Quality of Recommendations}
Figure~\ref{fig:uniformBeta} shows the expected total revenue achieved by various algorithms and baselines when $\beta_i$ is chosen uniformly at random from $[0,1]$.
As can be seen, \GGreedy \emph{consistently} yields the best revenue, leading the runner-up \RGreedy by non-trivial margins (about $10\%$ to $20\%$ gain).
\GlobalNo is always behind \GGreedy (about $10\%$ to $30\%$ loss), and so is \SGreedy compared to \RGreedy (about $1\%$ to $6\%$ behind).
\topkrev and \topkrating are always outperformed by all greedy algorithms (though it is less so for \topkrev).
\GGreedy is typically $30\%$ to $50\%$ better than \topkrev.
Note that the expected revenues are in the scale of tens of millions dollars, so even a small relative gain could translate to large revenue.
%: e.g., $1\%$ over \RGreedy on Amazon amount to about one hundred thousand more 

Figure~\ref{fig:BetaClass} and Figure~\ref{fig:BetaNoClass} show the comparisons of revenue with uniform $\beta_i$ values: $0.1$, $0.5$, and $0.9$.
The purpose is to examine how algorithms ``react'' to different strength of saturation effects.
As can be seen, the hierarchy of algorithms (ranked by revenue) is quite consistent with that in Figure \ref{fig:uniformBeta}, and importantly \GGreedy is always the top performer.
The gap between \GGreedy and the rest is larger with smaller $\beta_i$ (stronger saturation).
%plot the expected total revenue as a function of $\beta_i$ (which is the same for all items), ranging from $0.1, 0.2$ to $0.9$.
%The overall trend in all 12 cases is similar to that in Figure~\ref{fig:uniformBeta}.
%Still, the three greedy algorithms are always better than the two top-$k$ baselines.
%The gap between $\GGreedy$ and $\GlobalNo$ shrinks as $\beta_i$ goes up (saturation strength diminishes), which is intuitive.
%The performance of $\RGreedy$ and $\SGreedy$ are still quite close.
In Figure~\ref{fig:BetaNoClass} (class size $1$), though $\SGreedy$ is always behind $\RGreedy$, the difference becomes smaller as $\beta_i$ increases.
This intuitively suggests that $\RGreedy$ makes better decisions when it comes to repeated recommendations, as it is less sensitive to strong saturation.

In Figure~\ref{fig:submod}, we plot the growth of revenue as the greedy algorithms increment $S$ (Gaussian item quantities, $\beta_i$ uniform in $[0,1]$).
The lines for $\GGreedy$ clearly illustrate the phenomenon of diminishing marginal returns, empirically illustrating submodularity.
Interestingly enough, $\SGreedy$ and $\RGreedy$ have similar overall trends but also have ``segments'', corresponding to switches in time steps; submodularity can be observed within each ``segment''.
%Non-monotonicity can be also observed as the algorithms terminate when the marginal revenue of all remaining triples drops to zero.
%We notice that \GGreedy actually stops slightly earlier than the other two (but still yields a better revenue!).
%This shows non-monotonicity and the advantage of \GGreedy, namely planning a strategy in a \emph{holistic} manner.

Finally, Figure~\ref{fig:histClass} presents histograms on the number of repeated recommendations made by $\GGreedy$ for each user-item pair (item class size $>$ 1; other cases are similar and hence omitted).
The Y-axis is normalized to show the percentage instead of absolute counts.
Note that in both datasets, when $\beta_i = 0.1$, it happens more often that an item is recommended only once or twice to a user, since the dynamic adoption probability drops rapidly with small $\beta_i$.
As $\beta_i$ becomes large, the histogram becomes less skewed as more repeats show up, especially on Amazon with $\beta_i = 0.9$. This experiment shows that \GGreedy takes advantage of (lack of) saturation for making repeat recommendations to boost revenue.

\begin{figure}[t!]
	\centering
	\includegraphics[width=0.4\textwidth]{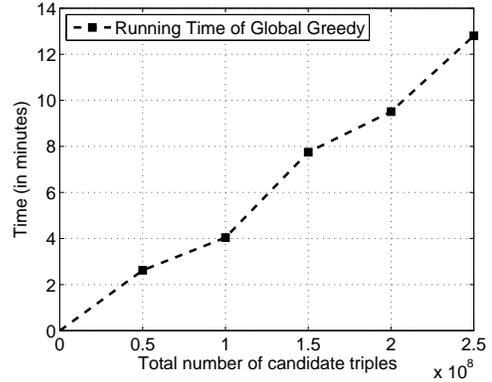}
\vspace{-3mm}
	\caption{Running time of \GGreedy on synthetic data}
	\label{fig:syn_time}
\end{figure}

\begin{table}
\centering
\begin{tabular}{| c | c | c | c | c | c |}
	\hline
		& {\bf GG}  & {\bf RLG} & {\bf SLG} & {\bf TopRE} & {\bf TopRA}   \\ \hline
	\textbf{Amazon} 	& 4.67  & 6.81 & 7.95 & 0.78 & 0.45	\\ \hline
	\textbf{Epinions} 	& 2.35  & 3.00 & 2.71 & 0.68 & 0.16  \\ \hline
\end{tabular}
\caption{ Running time (in mins) comparison}
\label{table:time}
\end{table}

\spara{Running Time and Scalability Tests}
Table \ref{table:time} reports the running time of various algorithms on Amazon and Epinions.
For lack of space, we only show the numbers for cases with uniform random $\beta_i$ and Gaussian item capacities (other cases are similar and hence omitted).
On both datasets, \GGreedy, \SGreedy and \RGreedy show scalability and efficiency and all finish under 10 minutes for Amazon and under 5 minutes for Epinions.
%\SGreedy and \RGreedy (one run) is more efficient than \GGreedy which operates on a heap seven times larger.
%However, the increased running time for \GGreedy is neglectable compared to the additional revenue generated by its superior solution.
%\GlobalNo is slower than \GGreedy as it ignores saturation, resulting in potentially more heap operations in lazy forward due to larger marginal revenue values.
%%%%%%%%%%%%%%%%
\eat{
\textcolor{red}{Particularly, \GGreedy (with two-level heaps and lazy-forward optimizations) is the most efficient, and is twice as fast as the version without two-level heaps.} 
\note[Laks]{Should we even talk about GG'? What about covering GG' and GG comparison briefly in the response and saying nothing about it here?} 

\note[Wei]{Re.: I agree on moving GG vs GG' to responses, and thus we can remove the red sentence， and that red column in Table 2}
}

In addition, we run \GGreedy on synthetic datasets that are much larger than Amazon and Epinions.
Figure \ref{fig:syn_time} illustrates the running time of \GGreedy on synthetic data with 100K, 200K, 300K, 400K, and 500K users with $T=5$ and each user having 100 items with non-zero adoption probability.
This means, for example, the dataset with 500K users has 250 million triples to select from.
The growth rate in Figure \ref{fig:syn_time} is almost linear and it takes about 13 minutes to finish on the largest one, which clearly demonstrates the scalability of \GGreedy. To put things in perspective, Netflix, the largest public ratings dataset, has only 100 million known ratings. 

In sum, our experiments on Amazon and Epinions data show that the proposed greedy algorithms are effective and efficient, producing far superior solutions than the baselines.
In particular, the most sophisticated \GGreedy consistently outperforms the rest and easily scales to (synthetic) data 2.5 times the size of the Netflix dataset. 

Next, in the following subsection, we report additional experimental results conducted in the settings where information about product prices is not completely available at the beginning of a time horizon.

\subsection{Experimental Results: Incomplete Product Prices}

So far, our experiments have focused on the setting where all product prices, i.e., $\price(i,t)$, for all $t$ in time horizon $[T]$, are available as input to algorithms like \GGreedy and \RGreedy when they are making recommendation decisions. 
However, this may not always be the case in practice, as prices are dynamic and exact values may not be available much in advance.
Notice that this does not affect \SGreedy as it only requires prices for the current time. 
%As such, we are interested in gauging how our algorithms perform in the following settings. 
As such, we are interested in gaugaing how our algorithms perform in the setting where product prices become available in batches and in time order.

\begin{figure*}[t!]
\begin{tabular}{cccc}
    \includegraphics[width=.24\textwidth]{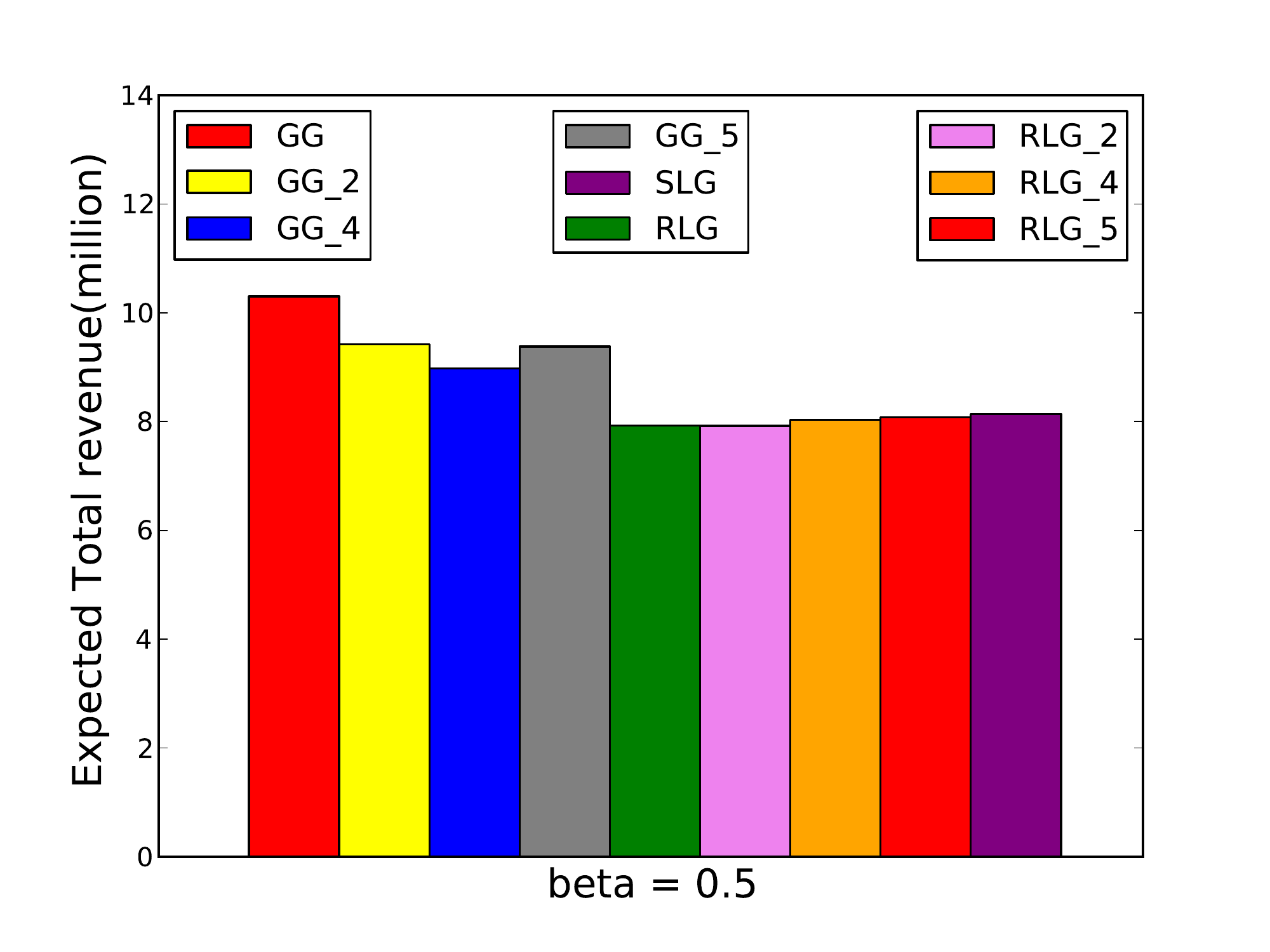}&
    \hspace{-2mm}\includegraphics[width=.24\textwidth]{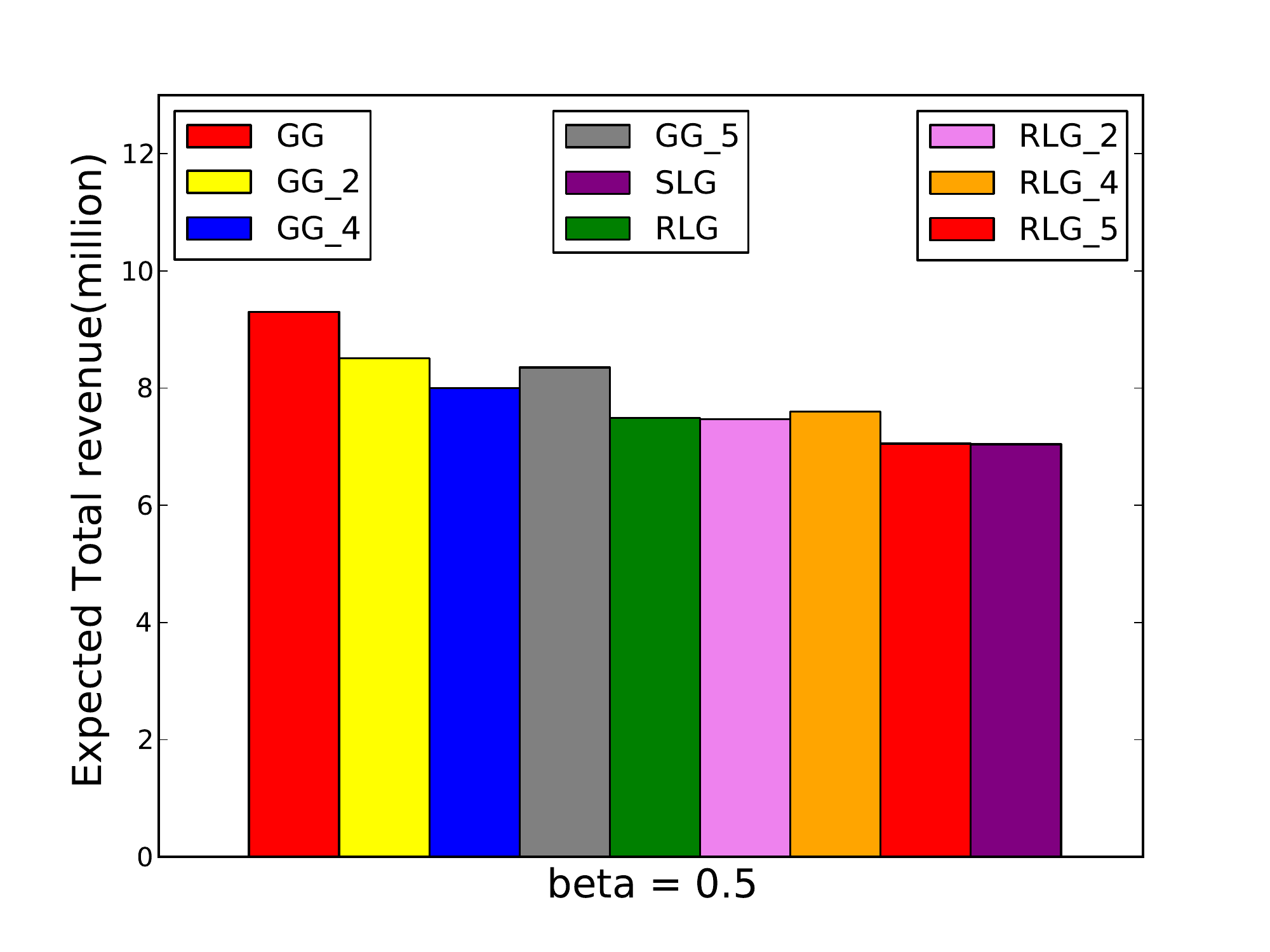}&
    \hspace{-2mm}\includegraphics[width=.24\textwidth]{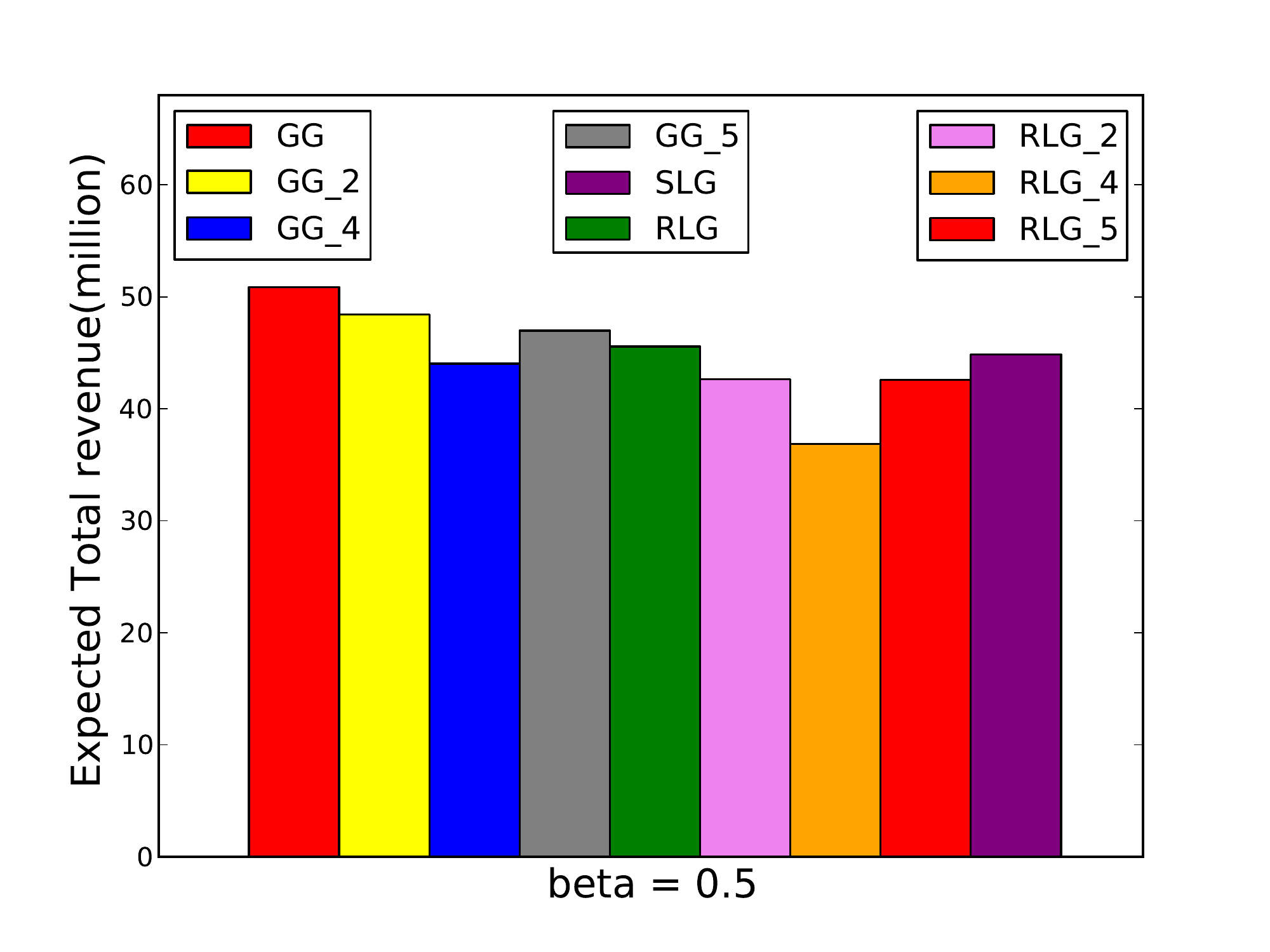}&
    \hspace{-2mm}\includegraphics[width=.24\textwidth]{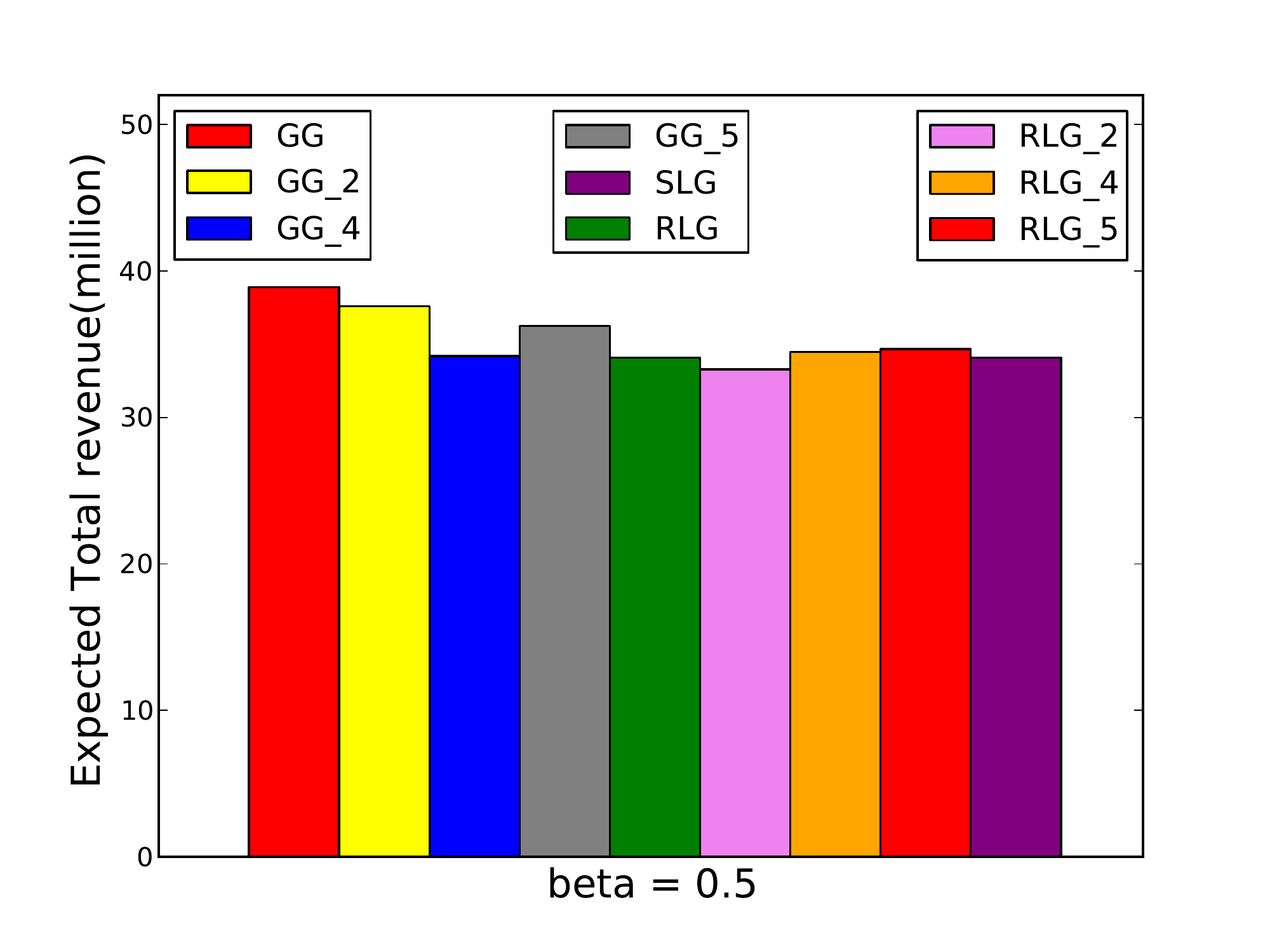}\\
	(a) Amazon (Gaussian)  & (b) Amazon (power-law)   & (c) Epinions (Gaussian) & (d) Epinions (power-law)  \\
\end{tabular}
\caption{Revenue comparison for complete prices with gradual availability. In the legend, $\mathrm{GG}_i$ and $\mathrm{RLG}_i$ means that the first sub-horizon is from time step $1$ to $i$, and the second sub-horizon is from time step $i+1$ to $T$.}
\label{fig:newExp1}
\end{figure*}

\eat{
\begin{figure*}[t!]
\begin{tabular}{cccc}
    \includegraphics[width=.24\textwidth]{plots_pdf/amazon_normal_incomplete}&
    \hspace{-2mm}\includegraphics[width=.24\textwidth]{plots_pdf/amazon_power_incomplete}&
    \hspace{-2mm}\includegraphics[width=.24\textwidth]{plots_pdf/epinions_normal_incomplete}&
    \hspace{-2mm}\includegraphics[width=.24\textwidth]{plots_pdf/epinions_power_incomplete}\\
    (a) Amazon (Gaussian)  & (b) Amazon (power-law)   & (c) Epinions (Gaussian) & (d) Epinions (power-law) \\
\end{tabular}
\caption{Revenue comparison for incomplete prices.}
\label{fig:newExp2}
\end{figure*}
}

%\smallskip\noindent\textbf{Complete prices with gradual availability.}
%\LL{
%In this case, product prices are available in batches and in time order.
More specifically, the time horizon $[T]$ is ``divided'' into sub-horizons $[T_1], [T_2], \ldots, [T_r]$, and prices become available sub-horizon after sub-horizon.
For example, suppose $T=7$, $[T_1] = \{1,2,3\}$, and $[T_2] = \{4,5,6,7\}$.
Initially, only prices for time steps $1$ through $3$ are known to the recommender system, and only at $t=4$, prices for time steps $4$ through $7$ will be known.
To adapt \GGreedy and \RGreedy, the algorithms will first come up with recommendations for $[T_1]$, and then given those, consider recommendations in $[T_2]$.
We expect both \GGreedy and \RGreedy to yield less revenue under this setting because they are no longer able to select triples in a holistic manner for the entire horizon $[T]$.
Also, note that \SGreedy is not affected at all since it already makes recommendations in chronological order.
%}

%\LL{
In our experiments, we set $T=7$ and split it into two sub-horizons, with cut-off time at $2$, $4$, and $5$, respectively.
For example, if cut-off is $2$, then $[T_1] = \{1,2\}$ and $[T_2] = \{3,4,5,6,7\}$.
Figure~\ref{fig:newExp1} shows the revenue achieved by \GGreedy and \RGreedy w.r.t. the three different cut-off time steps, and compares with the revenue yielded by these algorithms when prices are available all at once (hereafter referred to as the ``original setting'').
Both Amazon and Epinions datasets are used.
As can be seen from Figure~\ref{fig:newExp1}, \GGreedy ($\mathrm{GG}_2$, $\mathrm{GG}_4$, and $\mathrm{GG}_5$) still outperforms \RGreedy ($\mathrm{RLG}_2$, $\mathrm{RLG}_4$, and $\mathrm{RLG}_5$) and \SGreedy.
In the case of \GGreedy, $\mathrm{GG}_2$, $\mathrm{GG}_4$, and $\mathrm{GG}_5$ are all worse than the original setting (which is not surprising), and the loss is the greatest when the cut-off time is $4$, the most even split on $[T]$. This may be because when the split is not even (e.g., $\mathrm{GG}_2$  and $\mathrm{GG}_5$), the algorithm gets to see a larger bulk of the price information together compared with the even split ($\mathrm{GG}_4$). 
Similar trends can be observed for \RGreedy.  

\section{Extension to Random Prices}\label{sec:discuss}

So far, we have assumed that we have access to an exact pricing model, whereby  the exact price of items within a short time horizon is known. While there is evidence supporting this assumption and there are microeconomics studies on exact pricing model as pointed out in \textsection\ref{sec:intro}, the question arises what if the price prediction model is probabilistic in nature. More precisely, by treating prices $\price(i,t)$ as random variables, such a model only predicts them to within a distribution. Can we leverage the theory and techniques developed in this paper to deal with this case? We settle this question in this section. 
It is easy to see \revmax remains NP-hard, by restricting to the case where all prices are known exactly w.p. 1. The interesting question is how the  algorithms can be leveraged for the random price model. 
\eat{
A nice extension to our model is to consider a model with random prices instead of fixed ones.
That is, instead of assuming access to all prices $\price(i,t)$, we only know the distribution of $\price(i,t)$ (which is effectively a random variable).
Other ingredients are the same as described in \textsection\ref{sec:revmax}.
} 
In this case, the expected total revenue yielded by a strategy would be an expectation taken over the randomness of adoption events (as in the exact-pricing model) and over the randomness of prices. A major complication is that a closed form for the revenue function may not exist depending on the distribution from which prices are drawn. An obvious way of dealing with this is to treat the expected price (or most probable price) as a proxy for exact price and find a strategy that optimizes the expected revenue w.r.t. this assumption, using algorithms in \textsection\ref{sec:algo}. It is a heuristic and will clearly be suboptimal w.r.t. the true expected revenue. A more principled approach is to appeal to the Taylor approximation method used in convergence analysis~\cite{casella}.
It has the advantage of being distribution independent. 

Consider any valid strategy $S$, and a triple $z = (u, i, t) \in S$.
Define $[z]_S := \{(u,j,t') \in S : \cl{C}(j) = \cl{C}(i) \wedge t' \le t\}$.  %$\cup \{(u,i,t'')\in S: t'' < t\}$. 
That is, $[z]_S$ contains the triples that ``compete'' with $z$ under $S$.
%Also, let $\bp_{[z]}$  be the vector of random variables associated with elements in $[z]$.
Now consider the price random variables corresponding to all triples in $[z]_S$.
E.g., if items $i_1, i_2, i_3$ are from the same class, $S = \{(u,i_1,t_1), (u,i_2,t_2), (u,i_3,t_3)\}$, $t_1 < t_2 < t_3$, and $z=(u, i_3, t_3)$, then the revenue contribution of $z$ to $S$ will be dependent on the price vector $(\price(i_1,t_1), \price(i_2,t_2), \price(i_3,t_3))$.
For notational simplicity, we let $\bz$ be this price vector for $z$, and use $\bz_a$ to denote the $a$-th coordinate of $\bz$.
In the above example, if $a = 2$, then $\bz_a = \price(i_2, t_2)$.
The contribution from triple $z$ to the overall revenue of $S$ is clearly a function of the price vector $\bz$, and we denote it by $g(\bz)$.
For all $a = 1, 2,\ldots, |\bz|$, let $\bar{\bz}_a$ and $\Var(\bz_a)$ be the mean and variance of $\bz_a$ respectively.
Also, let $\Cov(\bz_a, \bz_b)$ be the covariance of two prices $\bz_a$ and $\bz_b$, where $a\neq b$.
Finally, by $\bar{\bz}$, we mean the vector of means of the price random variables in $\bz$.
In our running example, $\bar{\bz} = (\bar{\bz}_1, \bar{\bz}_2, \bar{\bz}_3)$, where $\bz_1 = \price(i_1, t_1)$, etc.
We then expand $g(\bz)$ at $\bar{\bz}$ using {\em Taylor's Theorem}:
\begin{align}\label{eqn:taylor}
g(\bz) &= g(\bar{\bz}) + \sum_{a=1}^{|\bz|} \frac{\partial g(\bar{\bz})}{\partial \bz_a}(\bz_a - \bar{\bz}_a) + \nonumber \\
		& \frac{1}{2} \sum_{a=1}^{|\bz|}  \sum_{b=1}^{|\bz|} \frac{\partial^2 g(\bar{\bz})}{\partial \bz_a \partial \bz_b}(\bz_a - \bar{\bz}_a)  (\bz_b- \bar{\bz}_b) + r(\bz),
\end{align}
where $r(\bz)$ is the remainder (consisting of higher order terms) and by Taylor's theorem it satisfies
$\lim_{\bz \to \bar{\bz}} \frac{r(\bz)}{(\bz-\bar{\bz})^2} = 0$.
Following standard practice \cite{casella}, this remainder is ignored since we are interested in an efficient approximation.

Disregarding the remainder and taking expectation over both sides of \eqref{eqn:taylor} gives the expected revenue contribution of triple $z$:
\begin{align} 
	&\E[g(\bz)] \approx \E[g(\bar{\bz})] + \sum_{a=1}^{|\bz|} \frac{\partial g(\bar{\bz})}{\partial \bz_a} \E[(\bz_a - \bar{\bz}_a)] + \nonumber \\  
	&\qquad \frac{1}{2} \sum_{a=1}^{|\bz|}  \sum_{b=1}^{|\bz|}   \frac{\partial^2 g(\bar{\bz})}{\partial \bz_a \partial \bz_b} \E[(\bz_a - \bar{\bz}_a)  (\bz_b- \bar{\bz}_b) ] \nonumber \\ 
	&= g(\bar{\bz}) + \frac{1}{2} \sum_{a=1}^{|\bz|}  \Var(\bz_a) +  \sum_{1 \le a < b \le |\bz|} \Cov(\bz_a, \bz_b), \label{eqn:taylor2}
\end{align}
where we have applied the linearity of expectation to get $\E[\bz_a - \bar{\bz}_a] = 0, \forall a$. 
Thus, for any strategy $S$, its expected total revenue, denoted $\mathit{Rand}\Rev(S)$, is 
$\mathit{Rand}\Rev(S) = \sum_{z\in S} \E[g(\bz)]$.

%\note[Wei]{This marks the end of Wei's notations}

The first three summands in \eqref{eqn:taylor2} correspond to mean (first-order), variance, and covariance (second-order) respectively. They are  used to estimate the true revenue function and the reminder is ignored.
More precisely, the algorithms for the exact-price model can be used, with the calculation of revenue changed by adding the extra variance and covariance terms as shown in \eqref{eqn:taylor2}. In principle, we can incorporate as many terms from the Taylor expansion as dictated by the accuracy desired and simply use the algorithms in \textsection\ref{sec:algo} for finding strategies with large revenue. 
\eat{
While statisticians usually expand up to the first-order term only \cite{casella}, for more accurate estimations of true expected revenue, one can always expand $g_z(\bp_{[z]})$ to even higher order terms. 
We point out that a practical challenge of working with the random price model is to obtain sufficient appropriate data to learn the parameters (mean, variance, covariance) accurately. 
} 
%To the best of our knowledge, no public accessible datasets satisfy the needs, and thus we leave empirical evaluation of the random price model as future work.

%Thus, for the random price model, only the calculation of revenue is changed compared to the fixed price model, and thus all of our proposed algorithms can be easily adapted here. 
%Especially, if only mean values are used, then the equivalence to fixed price is obvious.

%A more accurate estimation would require to know all variances and covariances, which are more difficult for an outsider (to an e-commerce business) to collect data and learn.
%We leave this as an interesting future work.

%Lastly, we stress that the above estimation method works for any distribution prices may follow.

\section{Conclusions and Future Work}\label{sec:concl}

In this work, we investigate the business-centric perspective of RS, and propose a
	dynamic revenue model by incorporating many crucial aspects such as price, 
	 competition, constraints, and saturation effects.
Under this framework, we study a novel problem \revmax, which asks to find a
	recommendation strategy that maximizes the expected
	total revenue over a given time horizon.
%The problem is NP-hard and we obtain an approximation algorithm
%	for a slightly relaxed version ($R$-\revmax).
We prove that \revmax is NP-hard and develop an approximation algorithm
	for a slightly relaxed version ($R$-\revmax) by establishing an elegant connection to matroid theory.
We also design intelligent greedy algorithms to tackle the original
	\revmax and conduct extensive experiments on Amazon and Epinion data
	to show the effectiveness and efficiency of these algorithms.
Furthermore, using synthetic data, we show that the \GGreedy algorithm 
scales to datasets 2.5 times the size of the Netflix dataset.

For future work, on the theoretical side, %it is open whether the original \revmax problem is  approximable. 
it is worth asking if \revmax remains NP-hard when every item belongs to
	its own class.
For the random price model, it is interesting to investigate if the Taylor approximation method can yield a strategy with a guaranteed approximation to the optimal solution w.r.t. true expected revenue. 
In reality, it is possible that prices, saturation, and competition may interact. Modeling and learning the interactions present is an interesting challenge. 
On the practical side, 
%one may consider extending the algorithmic framework of \revmax to a distributed setting to scale up to even more massive datasets.
an interesting challenge is to find suitable real datasets from which to learn the parameters for the
random price model and conduct empirical evaluations on it. Finally, here we have focused on revenue-maximizing recommendation problem, given an exogenous price model. Conversely, to find optimal pricing in order to maximize the expected revenue in the context of a given RS is an interesting problem which has clear connections to algorithmic game theory \cite{agtbook, klbbook}.

%\bibliographystyle{abbrv}
%\bibliography{mybib}

\begin{thebibliography}{10}
\smallskip

\bibitem{rssurvey05}
G.~Adomavicius and A.~Tuzhilin.
\newblock Toward the next generation of recommender systems: A survey of the
  state-of-the-art and possible extensions.
\newblock {\em IEEE Trans. Knowl. Data Eng.}, 17(6):734--749, 2005.

\bibitem{wsj}
J.~Angwin and D.~Mattioli.
\newblock Coming soon: Toilet paper priced like airline tickets.
\newblock {\em Wall Street Journal}, September 5, 2012, \url{http://on.wsj.com/1lECovl}

\bibitem{azaria13}
A.~Azaria et al. %, A.~Hassidim, S.~Kraus, A.~Eshkol, O.~Weintraub, and I.~Netanely.
\newblock Movie recommender system for profit maximization.
\newblock In {\em RecSys}, pages 121--128, 2013.

\bibitem{camel}
Camelytics.
\newblock {\em Prices Always Change}, 2012 (accessed May 9, 2014).
\newblock \url{http://bit.ly/1jznFgL}.

\bibitem{casella}
G.~Casella and R.~L. Berger.
\newblock {\em Statistical Inference (2nd edition)}.
\newblock Duxbury, 2002.

\bibitem{chen08}
L.-S. Chen et~al.
\newblock Developing recommender systems with the consideration of product
  profitability for sellers.
\newblock {\em Inf. Sci.}, 178(4):1032--1048, 2008.

\bibitem{das09}
A.~Das, C.~Mathieu, and D.~Ricketts.
\newblock Maximizing profit using recommender systems.
\newblock {\em CoRR}, abs/0908.3633, 2009.

\bibitem{dassarma12}
A.~Das-Sarma et al. %, S.~Gollapudi, R.~Panigrahy, and L.~Zhang.
\newblock Understanding cyclic trends in social choices.
\newblock In {\em WSDM}, pages 593--602, 2012.

\bibitem{even75}
S.~Even, A.~Itai, and A.~Shamir.
\newblock On the complexity of timetable and multi-commodity flow problems.
\newblock In {\em FOCS}, pages 184--193, 1975.

\bibitem{gabow83}
H.~Gabow.
\newblock An efficient reduction technique for degree-constrained subgraph and
  bidirected network flow problems.
\newblock In {\em STOC}, pages 448--456, 1983.

\bibitem{gabow89}
H.~Gabow and R.~Tarjan.
\newblock Faster scaling algorithms for network problems.
\newblock {\em SIAM J. Comput.}, 18(5):1013--1036, 1989.

\bibitem{mymedialite}
Z.~Gantner et~al.
\newblock Mymedialite: a free recommender system library.
\newblock In {\em RecSys}, pages 305--308, 2011.

\bibitem{gareyJohnson}
M.~R. Garey and D.~S. Johnson.
\newblock {\em Computers and Intractability: A Guide to the Theory of
  NP-Completeness}.
\newblock W. H. Freeman, 1979.

\bibitem{jiang07}
A.~X. Jiang et al. %and K.~Leyton-Brown.
\newblock Bidding agents for online auctions with hidden bids.
\newblock {\em Machine Learning}, 67(1-2):117--143, 2007.

\bibitem{kalish85}
S.~Kalish.
\newblock A new product adoption model with price, advertising, and
  uncertainty.
\newblock {\em Management Science}, 31(12):1569--1585, 1985.

\bibitem{kapoor13}
K.~Kapoor et al.  %, N.~Srivastava, J.~Srivastava, and P.~R. Schrater.
\newblock Measuring spontaneous devaluations in user preferences.
\newblock In {\em KDD}, pages 1061--1069, 2013.

\bibitem{koren10}
Y.~Koren.
\newblock Collaborative filtering with temporal dynamics.
\newblock {\em Commun. ACM}, 53(4):89--97, 2010.

\bibitem{mfsurvey}
Y.~Koren et al. %, R.~M. Bell, and C.~Volinsky.
\newblock Matrix factorization techniques for recommender systems.
\newblock {\em IEEE Computer}, 42(8):30--37, 2009.

\bibitem{lee10}
J.~Lee, V.~S. Mirrokni, V.~Nagarajan, and M.~Sviridenko.
\newblock Maximizing nonmonotone submodular functions under matroid or knapsack
  constraints.
\newblock {\em SIAM J. Discrete Math.}, 23(4):2053--2078, 2010.

\bibitem{leskovec07}
J.~Leskovec, L.~A. Adamic, and B.~A. Huberman.
\newblock The dynamics of viral marketing.
\newblock {\em TWEB}, 1(1), 2007.

%\bibitem{techrep-vldb14}
%W.~Lu, S.~Chen, K.~Li, L.V.S.~Lakshmanan.
%\newblock Show me the money: dynamic recommendations for revenue maximization (technical report).
%\newblock {\em CoRR}, arXiv:1409.0080, 2014.

\bibitem{manshadi13}
F.~Manshadi et~al.
\newblock A distributed algorithm for large-scale generalized matching.
\newblock {\em PVLDB}, 6(9):613--624, 2013.

\bibitem{minoux78}
M.~Minoux.
\newblock Accelerated greedy algorithms for maximizing submodular set
  functions.
\newblock In {\em IFIP Conf. on Optimization Techniques}, page 234--243, 1978.

\bibitem{gianmarco11}
G.~D.~F. Morales, A.~Gionis, and M.~Sozio.
\newblock Social content matching in mapreduce.
\newblock {\em PVLDB}, 4(7):460--469, 2011.

\bibitem{agtbook}
N.~Nisan et al. %, T.~Roughgarden, E.~Tardos, and V.~Vazirani.
\newblock {\em Algorithmic Game Theory}.
\newblock Cambridge University Press, New York, NY, USA, 2007.

\bibitem{Porteus90}
E.~L. Porteus.
\newblock Stochastic inventory theory.
\newblock %In D.~Heyman and M.~Sobel, editors, 
In {\em Stochastic Models}, volume~2, 
 % of {\em Handbooks in Operations Research and Management Science}, pages 605 -- 652.
Elsevier, 1990.

\bibitem{RS:hb10}
F.~Ricci, L.~Rokach, B.~Shapira, and P.~B. Kantor, editors.
\newblock {\em Recommender Systems Handbook}.
\newblock Springer, 2011.

\bibitem{richardson02}
M.~Richardson and P.~Domingos.
\newblock Mining knowledge-sharing sites for viral marketing.
\newblock In {\em KDD}, pages 61--70, 2002.

\bibitem{klbbook}
Y.~Shoham and K.~Leyton-Brown.
\newblock {\em Multiagent Systems - Algorithmic, Game-Theoretic, and Logical
  Foundations}.
\newblock Cambridge University Press, 2009.

\bibitem{silverman86}
B.~W. Silverman.
\newblock {\em Density estimation for statistics and data analysis}, volume~26.
\newblock CRC press, 1986.

\bibitem{snyder08}
C.~Snyder and W.~Nicholson.
\newblock {\em Microeconomic Theory, Basic Principles and Extensions (10th
  ed)}.
\newblock South-Western Cengage Learning, 2008.

\bibitem{wang13}
J.~Wang and Y.~Zhang.
\newblock Opportunity model for e-commerce recommendation: right product; right
  time.
\newblock In {\em SIGIR}, pages 303--312, 2013.

\bibitem{zhao12}
G.~Zhao, M.-L. Lee, W.~Hsu, and W.~Chen.
\newblock Increasing temporal diversity with purchase intervals.
\newblock In {\em SIGIR}, pages 165--174, 2012.

\end{thebibliography}

%\section*{Appendix: Responses to First-Round Reviews}
%\input{appendix}

\end{document}